\pdfoutput=1
\documentclass[draft]{lmcs} %
\usepackage[utf8]{inputenc}
\pdfoutput=1

\usepackage{lastpage}
\lmcsdoi{18}{1}{43}
\lmcsheading{}{\pageref{LastPage}}{}{}%
{Mar.~09,~2020}{Mar.~29,~2022}{}


\usepackage{jmsdelim}
\usepackage{macros}
\usepackage{cool-typography}

\begin{document}

\title{A Cubical Language for Bishop Sets}

\author[J.~Sterling]{Jonathan Sterling\rsuper{a}}  %
\author[C.~Angiuli]{Carlo Angiuli\rsuper{b}}  %
\author[D.~Gratzer]{Daniel Gratzer\rsuper{a}}  %

\address{Aarhus University}  %
\email{\{jsterling,gratzer\}@cs.au.dk}  %
\address{Carnegie Mellon University}  %
\email{cangiuli@cs.cmu.edu}  %

\begin{abstract}
  \noindent
  We present XTT, a version of Cartesian cubical type theory specialized for
  Bishop sets \`a la Coquand, in which every type
  enjoys a \emph{definitional} version of the uniqueness of identity proofs.
  Using cubical notions, XTT reconstructs many of the ideas underlying
  Observational Type
  Theory, a
  version of intensional type theory that supports function extensionality. We
  prove the canonicity property of XTT (that every closed boolean is
  definitionally equal to a constant) using Artin gluing.
\end{abstract}

\maketitle

\section{Introduction}

Little attention has been paid to notions of liberty and fraternity in dependent
type theory, but the same cannot be said about equality. Why? To define the
typing judgment $a:A$ we must determine which types are equal---because terms of
type $A$ may be cast (coerced) to any type $A'$ equal to $A$---but in the
presence of dependency, equality of types is contingent on equality of terms.
In this way, dependency transmutes term equality from a purely semantic
consideration to a core aspect of syntax.

As a practical matter, it is desirable to automate as many of these coercions
as possible. To that end, type theorists have spent decades refining decision
procedures for type equality modulo \eg, $\alpha$-, $\beta$-, $\delta$-, and
certain $\eta$-laws~\cite{coquand:1991,coquand:1996,stone-harper:2000,harper-pfenning:2005,stone-harper:2006,abel-coquand-dybjer:2008:mpc,abel-coquand-pagano:2009,abel-scherer:2012,abel:2013,abel-oehman-vezzosi:2017,gratzer-sterling-birkedal:2019}.
Unfortunately, not all desirable coercions can be automated---for instance,
mathematical equality of functions $\mathbb{N}\to\mathbb{N}$ is famously
undecidable. The collection of \emph{automated} equations in a given type
theory is called \emph{definitional equality};%
  \footnote{Historically, philosophical considerations have motivated
  explanations of \emph{definitional equality} as a scientific concept
  independent of a specific theory, sometimes leading to a notion of
  ``equality'' that is not a congruence (\eg, not respected by
  $\lambda$)~\cite{martin-lof:1975}; we argue that our theory-specific notion,
  based on the phenomenal aspect of automated conversion, is more reflective of
  everyday practice. As a programmatic matter, we moreover rule out any kind of
  ``equality'' for which the operations of type theory are not functional.}
the more equations are definitional, the less time users must spend providing
coercions.

Next, one must determine which coercions, if any, are recorded in \emph{terms}. A
priori, which coercions are recorded is independent of which equations are
definitional, but in practice these considerations are inextricably
linked---silent coercions along non-definitional equations typically disrupt
the aforementioned decision procedures. The canonical example is extensional
type theory~\cite{hofmann:1997}, in which one can silently coerce terms from any type to any other
under a contradictory assumption. In particular, elements of $A$ and $A\to A$
are identified in context $x:\CKwd{0}$, allowing users to encode fixed point
combinators; as a result, any decision procedure which relies on
$\beta$-reduction will no longer terminate.

A third consideration is the relationship between judgmental and propositional
equality. Following Martin-L\"of~\cite{martin-lof:1987}, type theorists arrange
concepts of interest into judgments, or top-level forms of assertion, such as typehood
($A\ \mathrm{type}$), membership ($a:A$), entailment
($\Gamma\vdash\mathcal{J}$), and \emph{judgmental equality} of types ($A=A'\
\mathrm{type}$) and terms ($a=a':A$). To account for higher-order concepts,
rather than admit higher-order judgments, we usually internalize judgmental notions as
types: dependent sums internalize context extension, dependent products
internalize entailment, and \emph{propositional equality} should in some sense
internalize judgmental equality.

\subsection{Notions of equality in type theory}

In the past fifty years, researchers have considered myriad presentations of
equality in type theory. Almost always, judgmental equality is a congruence
(reflexive, symmetric, transitive, and respected by all type and term formers)
along which coercion is silent, expressed by the \emph{conversion} rule that if
$\EqTy{A}{A'}$ and $\IsTm{a}{A}$ then $\IsTm{a}{A'}$. However, formulations of
coercion, definitional equality, and propositional equality differ widely; we
proceed by outlining several existing approaches.

\subsubsection{Equality reflection}

\NewDocumentCommand\CTyEq{mmm}{\CKwd{Eq}_{#1}\prn{#2,#3}}

The simplest way to internalize judgmental equality as a type is to provide
introduction and elimination rules making the existence of a proof of
$\CTyEq{A}{a}{a'}$ equivalent to the judgment $a=a':A$:

\begin{mathflow}
  \inferrule[introduction]{
    \EqTm{a}{a'}{A}
  }{
    \IsTm{\CKwd{refl}\prn{a}}{\CTyEq{A}{a}{a'}}
  }
  \split
  \inferrule[elimination]{
    \IsTm{p}{\CTyEq{A}{a}{a'}}
  }{
    \EqTm{a}{a'}{A}
  }
  \split
  \inferrule[unicity]{
    \IsTm{p}{\CTyEq{A}{a}{a}}
  }{
    \EqTm{p}{\CKwd{refl}\prn{a}}{\CTyEq{A}{a}{a}}
  }
\end{mathflow}

The elimination rule above is known as \emph{equality reflection}, and is
characteristic of extensional versions of type theory~\cite{martin-lof:1984}. Reflection immediately endows $\CTyEq{A}{a}{a'}$ with
many desirable properties: it is automatically a congruence, admits coercion
(via conversion), and enjoys \emph{uniqueness of identity proofs} (UIP, that any
two elements of $\CTyEq{A}{a}{a'}$ are equal).

Unfortunately, because propositional equality is undecidable, equality
reflection ensures that it is undecidable whether a judgment of the theory
holds; worse, as hinted previously, even the ``definitional fragment'' of the resulting
judgmental equality can no longer be automated, because $\beta$-reduction of
open terms may diverge.%
  \footnote{If, on the other hand, one omits $\beta$-equivalence from judgmental
  equality, it is possible to retain decidability of judgments in a
  dependently-typed language with divergent terms, as in the \textsc{Zombie}
  language~\cite{szoeberg-weirich:2015}.}
Therefore, proof assistants for extensional type theories cannot support type
checking, and rely instead on tactic-based construction of typing derivations.

One exemplar of this approach is the Nuprl proof assistant~\cite{constable:1986}
along with its descendants, including MetaPRL~\cite{hickey:2001} and
\textcolor[rgb]{.91,.31,.27}{Red}PRL~\cite{redprl:2018:lfmtp}. These type
theories are designed as ``windows on the truth'' of a single intended semantics
inspired by Martin-L\"of's computational meaning explanations, interpreting
types as partial equivalence relations over untyped terms~\cite{martin-lof:1979,allen:1987:thesis}. Nuprl-family proof assistants employ a
host of reasoning principles not validated by other models of type theory,
including intuitionistic continuity principles~\cite{rahli-bickford:2016},
computational phenomena such as exceptions and partiality~\cite{crary:1998}, and
``direct computation'' rules which use untyped rewrites to establish
well-formedness subgoals.

As a consequence of its fundamentally untyped nature, formalizing a theorem in
Nuprl does not imply the correctness of the corresponding theorem in standard
classical mathematics (the global mathematics of constant or discrete sets),
nor even in most forms of constructive mathematics (the local mathematics of
variable and cohesive sets). It is worth noting that the problem is \emph{not}
located in the presence of anti-classical principles (which are interpretable
in logic over a variety of topoi), and rather arises from the commitment to
untyped ontology.

The creators of the Andromeda proof
assistant~\cite{bauer-gilbert-haselwarter-pretnar-stone:2018} have introduced
another approach to implementing equality reflection, in which judgmental
equality is negotiated by means of algebraic effects and
handlers~\cite{bauer-pretnar:2015}; in essence, handlers allow users to provide
the out-of-band proofs of judgmental equality that are present in the
derivations (but not the terms) of extensional type theory. In contrast to
Nuprl, a proof formalized in Andromeda can be seen to imply the corresponding
informal statement in any variety of classical or constructive mathematics, a
consequence of the fact that an interpretation of Andromeda's extensional type
theory may be found lying over any topos.

\subsubsection{Intensional type theory}

\NewDocumentCommand\CTyId{mmm}{\CKwd{Id}_{#1}\prn{#2,#3}}

The \emph{identity type} of intensional type theory (ITT)~\cite{martin-lof:itt:1975,nordstrom-peterson-smith:1990} offers a much more
restrictive internalization of judgmental equality, characterized by the
following rules:

\begin{mathflow}
  \inferrule{
    \IsTm{a}{A}
  }{
    \IsTm{\CKwd{refl}(a)}{\CTyId{A}{a}{a}}
  }
  \split
  \inferrule{
    \IsTy[\Gamma,x:A,y:A,z:\CTyId{A}{x}{y}]{C(x,y,z)} \\
    \IsTm{p}{\CTyId{A}{a}{a'}} \qquad
    \IsTm[\Gamma,x:A]{c}{C(x,x,\CKwd{refl}(x))}
  }{
    \IsTm{\CKwd{J}_{x,y,z.C}(p;x.c)}{C(a,a',p)}
  }
  \split
  \inferrule{
  }{
    \EqTm{\CKwd{J}_{x,y,z.C}(\CKwd{refl}(a);x.c)}{\Subst{a}{x}{c}}{C(a,a,\CKwd{refl}(a))}
  }
\end{mathflow}

The elimination form $\CKwd{J}$ endows $\CTyId{A}{a}{a'}$ with many properties
expected of an equality connective, including symmetry, transitivity, and
coercion, which can be defined as follows:

\begin{mathflow}
  \CKwd{coerce}\ (p : \CTyId{\mathcal{U}}{A}{B}) : A\to B \defeq
  \CKwd{J}_{x,y,z.x\to y}(p;\_.\CTmAbs{x}{x})
\end{mathflow}

The tradeoffs of ITT are well-understood. By requiring explicit coercion for
non-$\alpha/\beta/\delta/\eta$ equations, ITT presents a theory with decidable
judgments. In practice, however, explicit coercions accumulate in types and
terms, requiring even more explicit coercions to mediate between previously-used
coercions. This is because coercions simplify only when applied to identity
proofs of the form $\CKwd{refl}(A)$; a coercion may mediate between
definitionally equal types and nevertheless fail to reduce (\eg, for a variable
of type $\CTyId{\mathcal{U}}{A}{A}$).

In addition to these practical considerations, identity types also fail to
evince several properties which some may expect of an equality connective, such
as UIP and function extensionality, the principle that
$(x:A)\to\CTyId{B}{f(x)}{g(x)}$ implies $\CTyId{A\to
B}{f}{g}$~\cite{streicher:1994,hofmann-streicher:1998}. In light of homotopy
type theory~\cite{hottbook}, it is reasonable to consider an equality
connective without UIP, but theorists and practitioners alike generally agree
that function extensionality is desirable. These shortcomings are sometimes
addressed by adjoining axioms for function extensionality or UIP (or
univalence), but axioms in the identity type cause even more coercions to become
irreducible.

\subsubsection{Setoids}

Another way to avoid the shortcomings of identity types in ITT is to work in
\emph{setoids}~\cite{hofmann:1995}, or Bishop sets~\cite{bishop:1967}, an exact
completion which replaces types by pairs of a carrier type $|A|$ and a
type-valued ``equivalence relation'' $=_A$. Each type former is lifted to
setoids extensionally: the setoid of functions $(|A|,=_A)\to (|B|,=_B)$
consists of functions $f:|A|\to|B|$ equipped with proofs $f_= : (x,y:|A|)\to
x=_A y \to f(x)=_B f(y)$ that they respect equivalence.

The framework of setoids allows users of type theory to ensure their
constructions are appropriately extensional, at the cost of manually proving
those conditions. In contrast, respect for the identity type is automatic (by
its elimination principle) but insufficiently powerful to imply function
extensionality. An ideal treatment of equality should, unlike setoids, take
advantage of the fact that constructions in type theory \emph{do} respect
function extensionality; syntactically well-behaved examples of this approach
include Observational Type Theory, cubical type theory, and XTT, discussed
below. Another recent proposal is to translate a type-theoretic language,
\emph{setoid type theory}, into setoids in ITT;\@ however, it is unknown
whether this language is complete or enjoys good syntactic properties~\cite{altenkirch-boulier-kaposi-tabareau:2019}.

\subsubsection{Observational Type Theory}

The first systematic account of extensional equality types in intensional type
theory was \emph{Observational Type Theory} (OTT)~\cite{altenkirch-mcbride:2006,altenkirch-mcbride-swierstra:2007}, which built
on earlier work by Altenkirch and McBride~\cite{altenkirch:1999,mcbride:1999}.
The main idea of OTT is to consider a closed (inductive-recursive) universe of
types, and to define propositional equality and its operations by recursion on
type structure. Concretely, for any two types $A,B$ there is a type of proofs
that $A$ equals $B$, and a coercion operation sending these proofs to functions
$A\to B$; then, for any $a:A$ and $b:B$, there is a type of proofs that $a$
heterogeneously equals $b$, and a coherence operation stating that terms
heterogeneously equal their coercions. Propositional equality in OTT satisfies
a definitional form of UIP\@.

Because OTT's equality types are defined by recursion, they will unfold into
complex types (reminiscent of the equality relations on setoids) when
sufficiently specialized; in XTT, path types do not unfold, but are easily
characterized up isomorphism when necessary. In both OTT and
XTT, however, any algorithm to check definitional equality of coercions must
rely on type constructors being injective up to the equality type, which we
have ensured by adding a type-case operator to the universe.  Although
type-case is acceptable or even desirable in programming~\cite{constable:1982,constable-zlatin:1984,harper-morrisett:1995,dagand:2013},
it is not justified by standard interpretations of the universe as a
Grothendieck universe.

Observational Type Theory also pioneered the idea that coercions should compute
on non-reflexive proofs of equality, a design principle which also plays a
significant role in the usability of cubical path types in contrast to standard
identity types.
Recently, McBride and collaborators have made progress toward a cubical version
of OTT based on a different cube category and coercion operation than the ones
considered in XTT~\cite{chapman-forsberg-mcbride:2018}.

We discuss in more detail the relationship to observational type theory, as
well as the type-case problem, in Section~\ref{sec:outlook}.

\subsubsection{Cubical type theory}

Homotopy type theory arises from the observation that ITT is compatible with
Voevodsky's \emph{univalence axiom}, which states that every equivalence (coherent
isomorphism) between types $A$ and $B$ gives rise to an identity proof
$\CTyId{\mathcal{U}}{A}{B}$~\cite{kapulkin-lumsdaine:2021,hottbook}. Univalence
solves longstanding difficulties with type-theoretic universes---transforming
them into object classifiers in the sense of higher topos theory~\cite{lurie:2009}---and contradicts UIP because two types can be equivalent in
several inequivalent ways. Unfortunately, adding univalence to ITT as an axiom
results in coercions that are ``stuck'' on non-reflexive identity proofs, as in
the case of adding axioms for function extensionality or UIP\@.

To address this problem, researchers have developed a number of \emph{cubical
type theories}~\cite{cchm:2017,abcfhl:2021,angiuli-favonia-harper:2018} whose
propositional equality and coercion operations support univalence in a computationally
well-behaved way. The core idea is to introduce a judgmental notion of equality
proof which is then internalized as the \emph{path type}.

Concretely, cubical type theories extend type theory with an abstract interval
$\DIM$ populated by \emph{dimension variables} $i:\DIM$ and constant endpoints
$\Dim0,\Dim1:\DIM$. A type parametrized by a dimension variable
$\IsTy[i:\DIM]{A}$ represents a proof that $\Subst{\Dim0}{i}{A}$ and
$\Subst{\Dim1}{i}{A}$ are equal types, and a term $\IsTm[i:\DIM]{a}{A}$ is a
heterogeneous equality proof between $\Subst{\Dim0}{i}{a}:\Subst{\Dim0}{i}{A}$
and $\Subst{\Dim1}{i}{a}:\Subst{\Dim1}{i}{A}$. Coercion in cubical type theory
is a primitive operation which computes based on the structure of the proof
$\IsTy[i:\DIM]{A}$, and admits an OTT-style ``coherence'' operation as a special
case. Because cubical type theory defines propositional equality in $A$ using
parametrized elements of $A$, propositional equality automatically inherits the
properties of each type and therefore satisfies function extensionality and
related principles.

There are several versions of cubical type theory. For instance, the Cubical Agda proof assistant~\cite{vezzosi-mortberg-abel:2019} implements a variant of the De Morgan version of cubical type theory~\cite{cchm:2017}, which equips $\DIM$ with negation and binary minimum and maximum
operations. Cartesian cubical type theory~\cite{abcfhl:2021,angiuli-favonia-harper:2018}, implemented in the
\textcolor[rgb]{.91,.31,.27}{Red}PRL~\cite{redprl:2018:lfmtp} and
\texttt{\textcolor[rgb]{.91,.31,.27}{red}tt}~\cite{redtt:2018:dagstuhl} proof
assistants, imposes no further structure on $\DIM$ but requires a stronger
coercion operation. In addition, Awodey, Cavallo, Coquand, Riehl, and Sattler~\cite{riehl:2019:hott} have recently proposed an \emph{equivariant} Cartesian
cubical type theory which is homotopically well-behaved, and Cavallo,
M\"ortberg, and Swan~\cite{cavallo-mortberg-swan:2020} have developed a common
generalization of the De Morgan and Cartesian type theories.

\subsubsection{Our contribution: XTT}

We describe XTT, a type theory without equality reflection whose propositional
equality connective satisfies function extensionality and definitional UIP\@. XTT
was introduced in a preliminary version of this work, which appeared in the 4th
International Conference on Formal Structures for Computation and Deduction
under the title \emph{Cubical Syntax for Reflection-Free Extensional Equality}~\cite{sterling-angiuli-gratzer:2019}.

Using ideas from Cartesian cubical type theory, XTT reconstructs the decisive aspects of
OTT in a more modular, judgmental fashion. For instance, instead of defining
equality separately at each type, we define path types uniformly in terms of
dimension variables; similarly, we impose UIP by means of a \emph{boundary
separation} rule which does not mention path types.

Compared to other cubical type theories~\cite{cchm:2017,abcfhl:2021,angiuli-favonia-harper:2018}, XTT has clear
advantages and disadvantages. Whereas other cubical type theories have two
separate connectives for path types and identity types, XTT's path types
strictly satisfy the rules of Martin-L\"of's identity types definitionally. In
addition, the rules governing composition---the most complex rules of every
cubical type theory---are substantially simpler in XTT than in Cartesian cubical
type theory. On the other hand, these simplifications are only possible because
XTT is concerned with \emph{Bishop sets}, a specific kind of cubical set
analogous to a setoid~\cite{coquand:2017:bish}, of which univalent universes and
higher inductive types are not instances.

In addition to the XTT calculus, our second contribution is an abstract
canonicity proof for XTT using the language of categorical gluing, summarized
in Section~\ref{sssec:metatheory-contribution}.

\subsection{Metatheory}

Type checkers rely on global invariants of type theories that are easily
disrupted---indeed, we have already seen that the equality reflection rule
single-handedly destroys the decidability of type checking. Consequently,
type theorists devote much effort to proving that various calculi are well-behaved, in the form of
the \emph{canonicity} and \emph{normalization} metatheorems.

Canonicity states that any closed term of boolean (or natural number) type is
judgmentally equal to either $\CTmTt$ or $\CTmFf$ (resp., a numeral). Canonicity
expresses a weak form of completeness for base types which is analogous to the
existence property of intuitionistic logic~\cite{troelstra-vandalen:1988}.
Although most type theories (including ETT, ITT, OTT, and XTT) enjoy canonicity,
it can fail when a type theory is extended by a new construct whose behavior is
not sufficiently determined by new equations, as in the extension of ITT by
function extensionality or univalence axioms. Indeed, the main motivation behind
cubical type theory was to develop a univalent type theory satisfying
canonicity.

Normalization is a generalization of canonicity to open terms which
characterizes the open terms of every type up to judgmental equality. These
characterizations can be quite complex: the normal forms of boolean type include
constants $\CTmTt$ and $\CTmFf$, variables $x:\CTyBool$, projections of
variables $x:\CTyBool\times\CTyBool$, \etc. Unlike canonicity, normalization does
not measure the strength of judgmental equality: normalization theorems can hold for
ITT extended with axioms, and hold trivially if one limits judgmental equality
to $\alpha$-equivalence.%
\footnote{In fact, without further constraints on what, precisely, constitutes a normal form,
  normalization can be quite trivial. If one defines a normal form to be simply the terms taken up
  to definitional equality then normalization is trivial but useless. In practice, type theorists
  are careful to isolate normal forms so as to make proving properties like decidability,
  injectivity of type constructors, \etc trivial.}

Conversely, while a failure of canonicity may indicate that judgmental equality
is too weak, a failure of normalization usually indicates that judgmental
equality is altogether intractable. Consider the judgmental injectivity of type
constructors, typically a consequence of normalization: if $\EqTy{A\to B}{A'\to B'}$ then
$\EqTy{A}{A'}$ and $\EqTy{B}{B'}$. Injectivity is crucial for type checking
because it enables the well-typedness of the application of $f : A\to B$ to $a :
A'$ to be reduced to checking whether $A=A'$. A priori, two function types may
be equal because both are equal to a third type $C$ by a sequence of
$\beta/\eta$ equalities; ruling this out generally requires a full
characterization of equality via normalization.

\subsubsection{Categories of models}

The rules of type theory are a complex mutual definition and simultaneous
quotient of the collections of contexts, types, and terms. Type theorists often
make such definitions precise by passing to a specialized setting known as a
\emph{logical framework}, offloading the bureaucratic aspects of the theory,
including in many cases the treatment of variable binding and hypothetical
judgment (as in the Edinburgh Logical
Framework~\cite{harper-honsell-plotkin:1993}), but far more importantly, the
compatibility of every operation with definitional equality (as in
Martin-L\"of's Logical Framework~\cite{nordstrom-peterson-smith:1990} and
Cartmell's generalized algebraic theories~\cite{cartmell:1986}).

Recall that canonicity and normalization will not hold for an arbitrary model
of the type theory but only for the `smallest' model, containing only what is
forced by the collection of rules; they may be easily refuted by
adding rules. Thus, for the purposes of metatheory, the construction of a type
theory must come equipped with an \emph{induction principle} stating in what
sense it is the smallest. These induction principles are in fact up for debate, as
choosing an induction principle is tantamount to fixing the range of possible
interpretations of the syntax. For example, in an Edinburgh LF encoding of type
theory, judgmental equality can have no special status and therefore admits non-trivial
interpretations, whereas mathematicians generally require that it be
interpreted as mathematical equality.

In the language of category theory, these considerations amount to specifying a
category of models of a type theory and exhibiting an initial object in that
category. Luckily, the rules of XTT are sufficiently non-exotic as to allow us
to obtain its functorial semantics by appealing to general existence
theorems~\cite{cartmell:1986,kaposi-kovacs-altenkirch:2019,uemura:2019,uemura:2021:thesis},
thereby sidestepping the so-called conjecture of initiality famously raised by
Voevodsky~\cite{voevodsky:2016:templeton}.

In previous work~\cite{sterling-angiuli-gratzer:2019}, we specified an early
version of XTT's logical semantics by regarding XTT as a generalized algebraic theory
(GAT) in the sense of Cartmell~\cite{cartmell:1986}. Models of GATs determine
choices of objects up to equality, and morphisms of models preserve these
choices strictly; models of the GAT of type theory thus determine \emph{up to
isomorphism} a category of contexts, and \emph{up to equality} the context
extension. This is a much stronger notion of ``category'' than can be
comfortably manipulated using the language of category theory: universal
properties determine object-level structure only up to canonical isomorphism,
and category-level structure only up to categorical equivalence (``isomorphism
up to isomorphism''). Accordingly, in the GAT discipline, one cannot usually
leverage general existence theorems but must instead provide explicit constructions in
order to strictly determine and preserve object-level structures.

We advocate for a more categorical viewpoint, in which morphisms of models
preserve structures only up to coherent isomorphism. In this paper, following
Sterling and Angiuli~\cite{sterling-angiuli:2020}, we instantiate Uemura's
framework~\cite{uemura:2019} to generate a functorial semantics for XTT;\@ models
form a $2$-category with a bi-initial object, and morphisms satisfy
compatibilities like $F(\Gamma.A) \cong F(\Gamma).F(A)$ which generalize
pseudomorphisms of natural models~\cite{clairambault-dybjer:2014,newstead:2018}.  While it may seem at first that
these canonical isomorphisms would incur additional bureaucracy, these weaker
morphisms in fact enable us to work much more abstractly, choosing
representations of objects only locally and as needed. Consequently, we have
managed to avoid nearly all the concrete computations that characterized the
technical development of our previous work on
XTT~\cite{sterling-angiuli-gratzer:2019}.

\subsubsection{Artin gluing}

Many important metatheorems famously cannot be proven by a straightforward induction on
the rules of type theory but require instead a more semantic induction
principle, such as the method of computability pioneered by Tait for the simply
typed $\lambda$-calculus~\cite{tait:1967} and further developed by Girard~\cite{girard:1971,girard:1972}, Martin-L\"of~\cite{martin-lof:itt:1975}, and
others. These methods associate to each type/context a proof-irrelevant
predicate or relation over its elements, and then establish that every element
satisfies the predicate associated to its type. To prove canonicity, one defines
the elements of a type to be its closed terms, and the predicate over
$\IsTm[\CxEmp]{b}{\CTyBool}$ states that $b\Downarrow\CTmTt$ or
$b\Downarrow\CTmFf$ where $\Downarrow$ is a deterministic evaluation relation
contained in judgmental equality.

These techniques have a few major disadvantages in the context of dependent type
theory. First, evaluation must be defined on typed terms modulo
$\alpha$-equivalence, not judgmental equality, because evaluation draws
distinctions between $\beta$-equivalent terms (\eg, $\CTmTt$ is an output of
evaluation whereas $\prn{\CTmLam{x}{x}}\prn{\CTmTt}$ is not); therefore, evaluation
cannot be studied using the machinery of models of type theory.%
  \footnote{This subequational aspect of evaluation is particularly thorny in
  cubical type theories because evaluation does not strictly respect dimension
  substitution, necessitating a technical condition known as closure under
  ``coherent expansion''~\cite{huber:2018,angiuli:2019}.}
Secondly, the predicate over $\IsTm[\CxEmp]{A}{\mathcal{U}}$ should intuitively
state that $A$ determines a type and thence a predicate over its elements, but a
(proof-irrelevant) predicate for $\mathcal{U}$ cannot store the data of a
predicate for each $A$; instead, we define a global lookup table for predicates~\cite{allen:1987:thesis,harper:1992}, and store in the predicate over
$\IsTm[\CxEmp]{A}{\mathcal{U}}$ the assertion that $A$ has an entry in the
table. However, constructing these type systems requires fixing a collection of
types at the outset, making these proofs brittle and difficult to extend.

Recently, type theorists have discovered that these difficulties can be
overcome by considering instead \emph{proof-relevant} predicates~\cite{altenkirch-kaposi:2016:nbe,coquand:2019,shulman:2015}, and that the
resulting constructions are best understood as instances of \emph{Artin
gluing}~\cite[Expos\'e I, Ch.~9]{sga:4}.\footnote{Researchers have been aware of connections
between computability predicates and gluing for much longer, but restricted to
the \emph{fiberwise proof-irrelevant} fragment of a gluing category~\cite{mitchell-scedrov:1993,jung-tiuryn:1993,fiore-simpson:1999}.}

Gluing-based techniques for type theory are perhaps most developed in the
context of \emph{weak} metatheorems such as homotopy canonicity~\cite{kapulkin-sattler:2019,shulman:2015} and homotopy parametricity~\cite{uemura:2017}, where it suffices to consider mathematically natural
notions of model in which substitution does not strictly commute with the
constructs of type theory~\cite{joyal:2017}. Canonicity and normalization are
also susceptible to gluing arguments, but these arguments have generally relied
on explicit constructions and computations rather than leveraging categorical
results as in the weak case~\cite{coquand:2019,kaposi-huber-sattler:2019}.
Subsequent to the writing of this paper, Sterling and Angiuli~\cite{sterling-angiuli:2021} have used proof-relevant logical predicates to establish
normalization for cubical type theory, building on Sterling and Harper's
observation that such arguments can be carried out in the internal type theory
of a gluing model~\cite{sterling-harper:2021}.

\subsubsection{Our contribution}%
\label{sssec:metatheory-contribution}

In this paper, we prove a canonicity theorem for XTT stating that any closed
term of boolean type in the initial model is judgmentally equal to either
$\CTmTt$ or $\CTmFf$. Our canonicity proof builds on results of Sterling and
Angiuli~\cite{sterling-angiuli:2020} concerning the gluing of models of type
theory along a flat functor. We emphasize the conceptual nature of our
canonicity proof, which avoids the explicit computations that pervaded both our
prior work on XTT~\cite{sterling-angiuli-gratzer:2019} and much of the related
work.

\section{XTT:\texorpdfstring{\@}{} a cubical language for Bishop sets}

We begin by introducing the XTT language informally (Figure~\ref{fig:grammar}) and sketching how we
recover (and improve upon) ordinary type-theoretic equality reasoning for Bishop sets. For the sake
of exposition, we elide structural rules, congruence rules, and obvious premises to equational
rules; formally, we define XTT as the bi-initial object in a 2-category of models in
Section~\ref{sec:functorial-semantics}, presented by the signature in Uemura's logical
framework~\cite{uemura:2019,uemura:2021:thesis} in Section~\ref{sec:xtt:lf}. Our presentation differs slightly from
the original formulation of XTT~\cite{sterling-angiuli-gratzer:2019}; we remark on these differences
as they appear.

\begin{figure}
  \begin{grammar}
    contexts & \Gamma,\Delta & \cdot \GrmSep \Gamma,i:\DIM \GrmSep \Gamma,\phi \GrmSep \Gamma, x:A
    \\
    dimensions & r, s & i\GrmSep \Dim0 \GrmSep \Dim1
    \\
    face formulas & \phi,\psi & r=s \GrmSep \phi\lor\psi
    \\
    types & A,B & \CEl{a} \GrmSep \CTyPi{x}{A}{B}\GrmSep \CTySg{x}{A}{B} \GrmSep \CTyPath{i.A}{a}{b}
      \GrmSep \CTyBool \GrmSep \CU \\
    terms & a, b & x \GrmSep \CTmLam{x}{a} \GrmSep a\ b \GrmSep \CTmPair{a}{b} \GrmSep \CTmFst{a}
      \GrmSep \CTmSnd{a} \GrmSep \CTmLam{i}{a} \GrmSep a\ r \GrmSep
      \CTmTt \GrmSep \CTmFf \GrmSep\\
    \GrmContinue &
      \CTmIf{x.A}{a}{b}{b'} \GrmSep
      \Abort \GrmSep \Split{\phi\to a}{\psi\to b} \GrmSep
      \CTmComp{r}{r'}{s}{i.a}{i.b} \GrmSep \\
    \GrmContinue &
      \CTmCoe{r}{r'}{i.a}{b} \GrmSep
      \CTmCodePi{x}{a}{b} \GrmSep \CTmCodeSg{x}{a}{b} \GrmSep \CTmCodePath{i.a}{b}{b'} \GrmSep
      \CTmCodeBool \GrmSep \\
    \GrmContinue &
      \CTmTypecase{x.A}{a}{
        \CKwd{pi}(x,x')\mapsto a'\mid\cdots\mid \CKwd{bool}\mapsto a''}
  \end{grammar}
  \caption{A summary of the informal syntax of XTT\@. Note that some binders, \eg{} those in
    $\CTmLam{i}{a}$ and $\CTyPath{i.A}{a}{b}$, range over %
    dimensions rather than terms.%
  }%
  \label{fig:grammar}
\end{figure}

\subsection{Judgmental structure of XTT}

Like other cubical type theories~\cite{cchm:2017,abcfhl:2021,angiuli-favonia-harper:2018,riehl-shulman:2017}, XTT
extends the judgmental apparatus of type theory with an abstract interval $\DIM$
and a collection $\COND$ of \emph{face formulas}, or propositions ranging over
the interval. Neither $\DIM$ nor $\COND$ are types, but we can extend contexts
by assumptions of either sort, in addition to ordinary typing assumptions.
(Previously, we collected assumptions of $\DIM$ and $\COND$ in a separate
context $\Psi$ to the left of $\Gamma$.)

\begin{mathflow}
  \inferrule{
  }{
    \IsCx{\CxEmp}
  }
  \split
  \inferrule{
    \IsCx
  }{
    \IsCx{\CxExtDim}
  }
  \split
  \inferrule{
    \IsCx & \IsCond
  }{
    \IsCx{\CxExtPrf}
  }
  \split
  \inferrule{
    \IsCx & \IsTy
  }{
    \IsCx{\CxExtEl}
  }
\end{mathflow}

Assumptions of all three sorts are subject to the structural rules of
hypothesis, substitution, and weakening. In addition to dimension variables
$i:\DIM$, the interval has two global elements $\CDim0$ and $\CDim1$ representing
its endpoints. (We call such an interval \emph{Cartesian} because it is the free
finite-product theory on two generators~\cite{awodey:2018}.) The face formulas
$\COND$ are closed under disjunction (unlike~\cite{sterling-angiuli-gratzer:2019}) and equality of dimensions.

\begin{mathflow}
  \inferrule{
  }{
    \IsDim{\CDim0}
  }
  \split
  \inferrule{
  }{
    \IsDim{\CDim1}
  }
  \split
  \inferrule{
    \IsDim{r} & \IsDim{s}
  }{
    \IsCond{\CondEq{r}{s}}
  }
  \split
  \inferrule{
    \IsCond{\phi}&\IsCond{\psi}
  }{
    \IsCond{\CondOr{\phi}{\psi}}
  }
\end{mathflow}

Given $\IsCond{\phi}$, we write $\IsTrue{\phi}$ when $\phi$ holds under the
assumptions in $\Gamma$. The rules governing this judgment are the evident ones,
with the caveat that under an assumption of $\CondEq{r}{s}$, one obtains a
judgmental equality $r=s:\DIM$; we may safely adopt this principle because,
unlike propositional equality in arbitrary types, $\CondEq{r}{s}$ is decidable.

\begin{mathflow}
  \inferrule{
    \EqDim{r}{s}
  }{
    \IsTrue{\CondEq{r}{s}}
  }
  \split
  \inferrule{
    \IsTrue{\CondEq{r}{s}}
  }{
    \EqDim{r}{s}
  }
  \split
  \inferrule{
    \IsTrue{\phi}
  }{
    \IsTrue{\phi\lor\psi}
  }
  \split
  \inferrule{
    \IsTrue{\psi}
  }{
    \IsTrue{\phi\lor\psi}
  }
  \split
  \inferrule{
    \IsTrue{\phi\lor\psi} &
    \IsTrue[\CxExtPrf{\phi}]{\chi} &
    \IsTrue[\CxExtPrf{\psi}]{\chi}
  }{
    \IsTrue{\chi}
  }
\end{mathflow}

In XTT, maps out of $\DIM$ correspond to equality proofs: $\IsTy[i:\DIM]{A}$ is
a proof that $\Subst{\CDim0}{i}{A}$ and $\Subst{\CDim1}{i}{A}$ are equal types, and
$\IsTm[i:\DIM]{a}{A}$ is a proof that $\Subst{\CDim0}{i}{a}:\Subst{\CDim0}{i}{A}$ and
$\Subst{\CDim1}{i}{a}:\Subst{\CDim1}{i}{A}$ are equal elements, modulo the proof $A$ that
$\Subst{\CDim0}{i}{A}$ and $\Subst{\CDim1}{i}{A}$ are equal types. Assumptions of face
formulas act as \emph{constraints}, restricting the domain of maps out of
$\DIM^n$. A hypothesis of $\CondEq{i}{\CDim0}$ sets $i$ to $\CDim0$ in the
hypotheses and conclusions that follow, whereas a type/element under the false
constraint $\CondEq{\CDim0}{\CDim1}$ is nothing at all. Finally, an element under a
disjunction $\CondOr{\phi}{\psi}$ is a pair of elements under $\phi$ and $\psi$
that agree on the overlap $\phi,\psi$.

Unlike~\cite{sterling-angiuli-gratzer:2019}, and following~\cite{cchm:2017}, we
include syntax for these \emph{partial elements} defined on nullary and binary
disjunctions:

\begin{mathflow}
  \inferrule{
    \IsTrue{\CondEq{\CDim0}{\CDim1}} &
    \IsTy{A}
  }{
    \IsTm{\Abort}{A}
  }
  \split
  \inferrule{
    \IsTrue{\CondEq{\CDim0}{\CDim1}} &
    \IsTm{a}{A}
  }{
    \EqTm{a}{\Abort}{A}
  }
  \split
  \inferrule{
    \IsTrue{\phi\lor\psi} &
    \IsTm[\CxExtPrf{\phi}]{a_\phi}{A} &
    \IsTm[\CxExtPrf{\psi}]{a_\psi}{A} &
    \EqTm[\CxExtPrf[\CxExtPrf{\phi}]{\psi}]{a_\phi}{a_\psi}{A}
  }{
    \IsTm{\Split{\phi\to a_\phi}{\psi\to a_\psi}}{A}
  }
  \split
  \inferrule{
    \IsTrue{\phi}
  }{
    \EqTm{\Split{\phi\to a_\phi}{\psi\to a_\psi}}{a_\phi}{A}
  }
  \split
  \inferrule{
    \IsTrue{\psi}
  }{
    \EqTm{\Split{\phi\to a_\phi}{\psi\to a_\psi}}{a_\psi}{A}
  }
  \split
  \inferrule{
    \IsTrue{\phi\lor\psi} &
    \IsTm{a}{A}
  }{
    \EqTm{a}{\Split{\phi\to a}{\psi\to a}}{A}
  }
\end{mathflow}

\begin{nota}[Boundary]\label{not:boundary}
  The interval has more generalized points than $\CDim0$ and $\CDim1$;
  therefore, it is not the case that $\IsTrue[i:\DIM]{i=\CDim0\lor i=\CDim1}$.
  This formula, called the \emph{boundary} of $i$, is important for expressing
  the rules of path types and compositions; we therefore impose the following
  notation:
  \[
    \Bdry{r} \defeq r=\CDim0\lor r=\CDim1
    \qedhere
  \]
\end{nota}

\begin{nota}[Judgmental restriction]\label{not:judgmental-restriction}
  We often want to consider a total term whose subcube coincides with some other
  term. We will write $\IsTm{a}{A}[\phi\to b]$ to abbreviate that $a$ is a term
  that restricts on $\phi$ to $b$, that is:
  \[
    \inferrule|double|{
      \IsTm{a}{A} &
      \EqTm[\CxExtPrf]{a}{b}{A}
    }{
      \IsTm{a}{A}[\phi\to b]
    }
    \qedhere
  \]
\end{nota}

\begin{exa}\label{ex:judgmental-restriction:bdry}
  Combining Notations~\ref{not:boundary}~and~\ref{not:judgmental-restriction}, we
  may succinctly express the situation where $p(i) : A$ exhibits a \emph{path}
  (proof of equality) between two elements $a_0,a_1:A$, writing
  $\IsTm[i:\DIM]{p(i)}{A}[\Bdry{i}\to a]$ where $a\defeq \brk{i=\CDim0\to a_0\mid
  i=\CDim1\to a_1}$.
\end{exa}

\subsection{Dependent path types in XTT}

The rules for dependent product, dependent sum, and boolean types in XTT are
completely standard and are located in Figure~\ref{fig:pi-sigma-bool}. Cubical
type theories internalize the judgmental ``equality situation'' of
Example~\ref{ex:judgmental-restriction:bdry} by means of \emph{(dependent) path
types}; a path type is, in essence, a dependent function out of the interval
subject to a restriction on the boundary of this function (\ie the behavior
of this function on the endpoints $\CDim0,\CDim1:\DIM$).

\begin{figure}
  \begin{mathflow}
    \inferrule{
      \IsTy{A} &
      \IsTy[\CxExtEl<x>{A}]{B}
    }{
      \IsTy{\CTyPi{x}{A}{B}}
    }
    \split
    \inferrule{
      \IsTm[\CxExtEl<x>{A}]{b}{B}
    }{
      \IsTm{\CTmLam{x}{b}}{\CTyPi{x}{A}{B}}
    }
    \split
    \inferrule{
      \IsTm{f}{\CTyPi{x}{A}{B}} &
      \IsTm{a}{A}
    }{
      \IsTm{f\prn{a}}{\Subst{a}{x}{B}}
    }
    \split
    \inferrule{
    }{
      \EqTm{\prn{\CTmLam{x}{b}}\prn{a}}{\Subst{a}{x}{b}}{\Subst{a}{x}{B}}
    }
    \split
    \inferrule{
    }{
      \EqTm{f}{\CTmLam{x}{f\prn{x}}}{\CTyPi{x}{A}{B}}
    }
    \break %
    \inferrule{
      \IsTy{A} &
      \IsTy[\CxExtEl<x>{A}]{B}
    }{
      \IsTy{\CTySg{x}{A}{B}}
    }
    \split
    \inferrule{
      \IsTy[\CxExtEl<x>{A}]{B} &
      \IsTm{a}{A} &
      \IsTm{b}{\Subst{a}{x}{B}}
    }{
      \IsTm{\CTmPair{a}{b}}{\CTySg{x}{A}{B}}
    }
    \split
    \inferrule{
      \IsTm{p}{\CTySg{x}{A}{B}}
    }{
      \IsTm{\CTmFst{p}}{A}
    }
    \split
    \inferrule{
      \IsTm{p}{\CTySg{x}{A}{B}}
    }{
      \IsTm{\CTmSnd{p}}{\Subst{\CTmFst{p}}{x}{B}}
    }
    \split
    \inferrule{
    }{
      \EqTm{\CTmFst{\CTmPair{a}{b}}}{a}{A}
    }
    \split
    \inferrule{
    }{
      \EqTm{\CTmSnd{\CTmPair{a}{b}}}{b}{\Subst{a}{x}{B}}
    }
    \split
    \inferrule{
    }{
      \EqTm{p}{\CTmPair{\CTmFst{p}}{\CTmSnd{p}}}{\CTySg{x}{A}{B}}
    }
    \break %
    \inferrule{
    }{
      \IsTy{\CTyBool}
    }
    \split
    \inferrule{
    }{
      \IsTm{\CTmTt}{\CTyBool}
    }
    \split
    \inferrule{
    }{
      \IsTm{\CTmFf}{\CTyBool}
    }
    \split
    \inferrule{
      \IsTy[\CxExtEl{\CTyBool}]{C} &
      \IsTm{b}{\CTyBool} &
      \IsTm{c_{\CTmTt}}{\Subst{\CTmTt}{x}{C}} &
      \IsTm{c_{\CTmFf}}{\Subst{\CTmFf}{x}{C}}
    }{
      \IsTm{\CTmIf{x.C}{b}{c_{\CTmTt}}{c_{\CTmFf}}}{\Subst{b}{x}{C}}
    }
    \split
    \inferrule{
    }{
      \EqTm{\CTmIf{x.C}{\CTmTt}{c_{\CTmTt}}{c_{\CTmFf}}}{c_{\CTmTt}}{\Subst{\CTmTt}{x}{C}}
    }
    \split
    \inferrule{
    }{
      \EqTm{\CTmIf{x.C}{\CTmFf}{c_{\CTmTt}}{c_{\CTmFf}}}{c_{\CTmFf}}{\Subst{\CTmFf}{x}{C}}
    }
  \end{mathflow}
  \caption{Rules for dependent products, dependent sums, and booleans.}%
  \label{fig:pi-sigma-bool}
\end{figure}

Given a line of types $\IsTy[i:\DIM]{A}$ and two elements $a_0 :
\Subst{\CDim0}{i}{A}, a_1:\Subst{\CDim1}{i}{A}$, the type of paths between
$a_0$ and $a_1$ is written $\CTyPath{i.A}{a_0}{a_1}$; because the type $A$ can
depend on $i$, path types express a kind of heterogeneous equality (though
different from the one proposed by McBride~\cite{mcbride:1999}). The rules for
path types are summarized below:

\begin{mathflow}
  \inferrule[formation]{
    \IsTy[\CxExtDim<i>]{A} &
    \IsTm{a_0}{\Subst{\CDim0}{i}{A}} &
    \IsTm{a_1}{\Subst{\CDim1}{i}{A}}
  }{
    \IsTy{\CTyPath{i.A}{a_0}{a_1}}
  }
  \split
  \inferrule[introduction]{
    \IsTm[\CxExtDim<i>]{a}{A}
  }{
    \IsTm{\CTmAbs{i}{a}}{
      \CTyPath{i.A}{\Subst{\CDim0}{i}{a}}{\Subst{\CDim1}{i}{a}}
    }
  }
  \split
  \inferrule[elimination]{
    \IsTm{p}{\CTyPath{i.A}{a_0}{a_1}} &
    \IsDim{r}
  }{
    \IsTm{p\prn{r}}{\Subst{r}{i}{A}}[\Bdry{r}\to\Split{r=\CDim0\to a_0}{r=\CDim1\to a_1}]
  }
  \split
  \inferrule[computation]{
  }{
    \EqTm{\prn{\CTmAbs{i}{a}}\prn{r}}{\Subst{r}{i}{a}}{\Subst{r}{i}{A}}
  }
  \split
  \inferrule[uniqueness]{
  }{
    \EqTm{p}{\CTmAbs{i}{p\prn{i}}}{\CTyPath{i.A}{a_0}{a_1}}
  }
\end{mathflow}

\begin{rem}
  One can also express the data of a path type as a line of types
  $\IsTy[i:\DIM]{A}$ together with a partial element
  $\IsTm[i:\DIM,\Bdry{i}]{a}{A}$; then, the elements of the path type
  $\CKwd{path}_{i.A}\prn{i.a}$ would consist in elements
  $\IsTm[i:\DIM]{p}{A}[\Bdry{i}\to a]$. In fact, we use exactly this style of
  definition in our mathematical version of the syntax of XTT (see
  Section~\ref{sec:functorial-semantics}).
\end{rem}

We are already prepared to see one of the advantages of cubical type theories
over intensional Martin-L\"of type theory. Because we use maps out of the
interval to represent equality, equations in $A$ naturally take the form of
(parametrized) elements of $A$; therefore, the ``introduction rules'' for
equality in $A$ are the same as the introduction rules for $A$ itself.

\begin{exa}\label{ex:funext}
  Function extensionality, provable in cubical type theories, provides a
  particularly convincing example, considering that it can be derived directly
  using only the rules for dependent function and path types. Given two
  functions $f,g:\CTyPi{x}{A}{B}$ and a family of paths
  $h:\CTyPi{x}{A}{\CTyPath{\_.B}{f(x)}{g(x)}}$, we have:
  \[
    \CTmAbs{i}{\CTmAbs{x}{h\prn{x}\prn{i}}} : \CTyPath{\_.\CTyPi{x}{A}{B}}{f}{g}
    \qedhere
  \]
\end{exa}

\subsection{Universe of Bishop sets}

XTT is equipped with a universe of \emph{Bishop sets}, \ie types that satisfy
a definitional version of the unicity of identity proofs. As was the case for
Observational Type Theory, it is essential that that this universe is
\emph{closed} --- a matter we will discuss in more detail in
Sections~\ref{sec:typecase}~and~\ref{sec:bad-universes}.

In our original presentation of XTT~\cite{sterling-angiuli-gratzer:2019}, we
required that all types were Bishop sets. Here, we require this property only of
elements of the universe, in order to suggest how one might integrate XTT into a
standard (univalent) Cartesian cubical type theory in which not all types are Bishop sets
(notably, univalent universes and higher inductive types). Additionally, whereas
we previously described an infinite and cumulative hierarchy of universes \`a la
Coquand,\footnote{Universes \`a la Coquand~\cite{coquand:2013:psh-model}
differ from universes \`a la Tarski in a few ways: one eschews the standard
$\IsTy{A}$ judgment for a stratified judgment $\Gamma\vdash A\
\mathit{type}_i$, and then the rules for each universe $\mathscr{U}_i$
exhibit an isomorphism between the collection of types of level $i$ and the
collection of elements of $\mathscr{U}_i$.}
here we have opted to specify only a single universe \`a la Tarski for the sake
of simplicity and clarity of presentation.

We begin with the basic formation rules for the universe of Bishop sets:

\begin{mathflow}
  \inferrule{
  }{
    \IsTy{\CU}
  }
  \split
  \inferrule{
    \IsTm{\hat{A}}{\CU}
  }{
    \IsTy{\CEl{\hat{A}}}
  }
\end{mathflow}

\begin{nota}
  As above, we adopt the convention of writing $\hat{A}$ for an element of $\CU$;
  then, we will write $a\in\hat{A}$ as a shorthand for $a : \CEl{\hat{A}}$.
\end{nota}

\subsubsection{Boundary separation and UIP}

What makes types classified by $\CU$ special is that they satisfy the
\emph{boundary separation} principle below, a modular reconstruction of the
uniqueness of identity proofs:

\begin{mathflow}
  \inferrule[boundary separation]{
    \IsTm{\hat{A}}{\CU} & \IsDim{r} & \EqTm*[\CxExtPrf{\Bdry{r}}]{a}{b}{\hat{A}}
  }{
    \EqTm*{a}{b}{\hat{A}}
  }
\end{mathflow}

To see that boundary separation implies the unicity of identity proofs,
consider the context $\Gamma \defeq
\prn{\Delta,\hat{A}:\CU,a\in\hat{A},b\in\hat{A},p :
\CTyPath{\_.\CEl{\hat{A}}}{a}{b},q:\CTyPath{\_.\CEl{\hat{A}}}{a}{b}}$; we may
derive $\EqTm[\Gamma]{p}{q}{\CTyPath{\_.\CEl{\hat{A}}}{a}{b}}$ as follows:
\begin{prooftree*}
  \infer0{
    \EqTm*[\Gamma,i:\DIM,\Bdry{i}]{
      \brk{i=\CDim0\to a\mid i=\CDim1\to b}
    }{
      \brk{i=\CDim0\to a\mid i=\CDim1\to b}
    }{\hat{A}}
  }

  \infer1{
    \EqTm*[\Gamma,i:\DIM,\Bdry{i}]{
      \brk{i=\CDim0\to p(i)\mid i=\CDim1\to p(i)}
    }{
      \brk{i=\CDim0\to q(i)\mid i=\CDim1\to q(i)}
    }{\hat{A}}
  }

  \infer1{
    \EqTm*[\Gamma,i:\DIM,\Bdry{i}]{p(i)}{q(i)}{\hat{A}}
  }

  \infer1[\textsc{boundary separation}]{
    \EqTm*[\Gamma,i:\DIM]{p(i)}{q(i)}{\hat{A}}
  }
  \infer1{
    \EqTm{\CTmAbs{i}{p(i)}}{\CTmAbs{i}{q(i)}}{\CTyPath{\_.\CEl{\hat{A}}}{a}{b}}
  }
  \infer1{
    \EqTm{p}{q}{\CTyPath{\_.\CEl{\hat{A}}}{a}{b}}
  }
\end{prooftree*}

\subsubsection{Coercion and composition}

Another important aspect of Bishop sets in XTT is that they support
\emph{coercion} and \emph{composition} operations:

\begin{mathflow}
  \inferrule{
    \IsTm[\CxExtDim<i>]{\hat{A}}{\CU}&
    \IsTm*{a}{\Subst{r}{i}{\hat{A}}}
  }{
    \IsTm*{\CTmCoe{r}{r'}{i.\hat{A}}{a}}{\Subst{r'}{i}{\hat{A}}}[r'=r\to a]
  }
  \break
  \inferrule{
    \IsTm[\CxExtDim<i>]{\hat{A}}{\CU}&
    \IsTm*[\CxExtPrf[\CxExtDim<i>]{i=r\lor\Bdry{s}}]{a}{\hat{A}}
  }{
    \IsTm*{\CTmComp{r}{r'}{s}{i.\hat{A}}{i.a}}{\Subst{r'}{i}{\hat{A}}}[
      r'=r\lor\Bdry{s}\to \Subst{r'}{i}{a}
    ]
  }
\end{mathflow}

In essence, these operations implement the action of paths in every set,
simultaneously enabling coercions between equal types, as well providing a way
to compose and invert paths. The coercion operation above allows, in
particular, an element of a set to be transformed into an element of any equal
set: this is the action of $\CTmCoe{\CDim0}{\CDim1}{i.\hat{A}}{a}$. In Observational Type
Theory, there is an additional \emph{coherence} operation that
(heterogeneously) equates $a$ with its coercion $\CTmCoe{\CDim0}{\CDim1}{i.\hat{A}}{a}$;
in XTT (and Cartesian cubical type theories generally), this is accomplished
using another instance of the general coercion operator called a ``filler'':
\[
  \prn{\lambda i.\CTmCoe{\CDim0}{i}{i.\hat{A}}{a}} :
  \CTyPath{i.\CEl{\hat{A}}}{a}{\CTmCoe{0}{1}{i.\hat{A}}{a}}
\]

Composition is analogous to coercion, except that it may additionally constrain
the result to match a partial element defined on a boundary $\Bdry{s}$ for some
$s:\DIM$. (Because of boundary separation and regularity, XTT's composition
operator is substantially simpler than those of other cubical type theories,
which consider partial elements defined on arbitrary $\phi:\COND$.) Composition
can be used to define combinators expressing the symmetry and transitivity of
equality, as well as to implement Martin-L\"of's $\CKwd{J}$ eliminator. In fact,
to express symmetry and transitivity, it suffices to first consider the case
where $\hat{A}$ doesn't depend on $i$, called \emph{homogeneous composition}:
\[
  \CTmHcomp{r}{r'}{s}{\hat{A}}{i.a} \defeq
  \CTmComp{r}{r'}{s}{\_.\hat{A}}{i.a}
\]

\begin{exa}[Symmetry]
  Let $\hat{A}:\CU$ and let $a,b\in\hat{A}$ and let
  $p:\CTyPath{\_.\CEl{\hat{A}}}{a}{b}$. We may use homogeneous composition to
  define an inverse path $\bar{p} : \CTyPath{\_.\CEl{\hat{A}}}{b}{a}$:
  \[
    \bar{p}
    \defeq
    \CTmAbs{i}{
      \CTmHcomp{\CDim0}{\CDim1}{i}{\hat{A}}{
        j.
        \brk{
          j=\CDim0\lor i=\CDim1 \to p(\CDim0)
          \mid
          i=\CDim0\to p(j)
        }
      }
    }
    \qedhere
  \]
\end{exa}

\begin{exa}[Transitivity]
  Let $\hat{A}:\CU$ and let $a,b,c\in\hat{A}$ and let
  $p:\CTyPath{\_.\CEl{\hat{A}}}{a}{b}, q:\CTyPath{\_.\CEl{\hat{A}}}{b}{c}$. We
  may use homogeneous composition to define a composite path $p\cdot q:
  \CTyPath{\_.\CEl{\hat{A}}}{a}{c}$:
  \[
    p\cdot q
    \defeq
    \CTmAbs{i}{
      \CTmHcomp{\CDim0}{\CDim1}{i}{\hat{A}}{
        j.
        \brk{
          j=\CDim0\lor i=\CDim0\to p(i)
          \mid
          i=\CDim1\to q(j)
        }
      }
    }
    \qedhere
  \]
\end{exa}

\begin{rem}
  By boundary separation, the symmetry and transitivity operators act very strictly. For instance,
  given $p : \CTyPath{i.A}{a}{b}$ one has $\bar{p} \cdot p = (\CTmLam{i}{b})$
  definitionally. Similarly, composition of paths is definitionally associative:
  $p\cdot \prn{q\cdot w} = \prn{p\cdot q}\cdot w$. In the absence of boundary separation, these
  coherences would hold up to another path, using a more complex instance of the composition
  operation.
\end{rem}

\begin{exa}[Identity type]\label{ex:identity-type}
  Using composition, we may define a combinator with
  the same type as Martin-L\"of's $\CKwd{J}$ eliminator for the identity type.
  Let $\hat{A} : \CU$ be a set and
  $x\in\hat{A},y\in\hat{A},z:\CTyPath{\_.\CEl{\hat{A}}}{x}{y}\vdash
  \hat{C}(x,y,z) : \CU$ be a motive of induction.
  Fixing $a,b\in\hat{A}$ and $p : \CTyPath{\_.\CEl{\hat{A}}}{a}{b}$ and
  $x:A\vdash c(x) \in \hat{C}(x,x,\lambda\_.x)$, we may define an element
  $\CKwd{J}_{\hat{C}}(p,c)\in \hat{C}(a,b,p)$ as follows:
  \[
    \CKwd{J}_{\hat{C}}(p,c) \defeq
    \CTmCoe{0}{1}{i.
      \hat{C}\prn{
        p(0),p(i),
        \lambda j.
        \CTmHcomp{0}{j}{i}{\hat{A}}{
          k.
          \brk{
            k=\CDim0\lor i=\CDim0\to p(0)
            \mid
            i=\CDim1\to p(k)
          }
        }
      }
    }{
      c(p(0))
    }
    \qedhere
  \]
\end{exa}

What is the behavior of the $\CKwd{J}$ combinator from
Example~\ref{ex:identity-type} on a \emph{reflexive} proof of equality
$\lambda\_.a : \CTyPath{\_.\CEl{\hat{A}}}{a}{a}$? From Martin-L\"of type
theory, we would expect $\CKwd{J}_{\hat{C}}\prn{\lambda\_.a, c}$ to compute to
$c(a)$; in ordinary cubical type theory, this equation only holds up to another
path, but in XTT we can force it to hold using the following \emph{regularity}
principle:

\begin{mathflow}
  \inferrule[coercion regularity]{
    \EqTm[\Gamma,i:\DIM,j:\DIM]{\hat{A}}{\Subst{j}{i}{\hat{A}}}{\CU}
  }{
    \EqTm*{\CTmCoe{r}{r'}{i.\hat{A}}{a}}{a}{\Subst{r'}{i}{\hat{A}}}
  }
\end{mathflow}

\begin{rem}
  Considering the boundary $\Bdry{s}$, the boundary separation rule ensures
  that the standard decomposition of composition into homogeneous composition
  and coercion~\cite{abcfhl:2021,angiuli-favonia-harper:2017} holds
  definitionally:
  \[
    \CTmComp{r}{r'}{s}{i.\hat{A}}{i.a}
    =
    \CTmHcomp{r}{r'}{s}{\Subst{r'}{i}{\hat{A}}}{
      i.\CTmCoe{i}{r'}{i.\hat{A}}{a}
    }
  \]

  Consequently, the following regularity rule for composition is also derivable:
  \begin{gather*}
    \inferrule{
      \EqTm[\Gamma,i:\DIM,j:\DIM]{\hat{A}}{\Subst{j}{i}{\hat{A}}}{\CU} \\
      \EqTm*[\Gamma,i:\DIM,i=r\lor\Bdry{s},j:\DIM,j=r\lor\Bdry{s}]{a}{\Subst{j}{i}{a}}{\hat{A}}
    }{
      \EqTm*{\CTmComp{r}{r'}{s}{i.\hat{A}}{i.a}}{\Subst{r'}{i}{a}}{\Subst{r'}{i}{\hat{A}}}
    }
    \qedhere
  \end{gather*}
\end{rem}

\subsubsection{Closure of the universe under connectives}\label{sec:informal:connectives}

The universe is closed under codes for connectives in the standard way, by
adding introduction forms for each code and equations governing the behavior of
$\CEl{-}$ on codes:

\begin{mathflow}
  \inferrule{
    \IsTm{\hat{A}}{\CU} &
    \IsTm[\Gamma,x\in\hat{A}]{\hat{B}}{\CU}
  }{
    \IsTm{\CTmCodePi{x}{\hat{A}}{\hat{B}}}{\CU}
  }
  \split
  \inferrule{
    \IsTm{\hat{A}}{\CU} &
    \IsTm[\Gamma,x\in\hat{A}]{\hat{B}}{\CU}
  }{
    \IsTm{\CTmCodeSg{x}{\hat{A}}{\hat{B}}}{\CU}
  }
  \split
  \inferrule{
    \IsTm[\CxExtDim<i>]{\hat{A}}{\CU} &
    \IsTm*{a}{\Subst{\CDim0}{i}{\hat{A}}} &
    \IsTm*{b}{\Subst{\CDim1}{i}{\hat{A}}}
  }{
    \IsTm{\CTmCodePath{i.\hat{A}}{a}{b}}{\CU}
  }
  \split
  \inferrule{
  }{
    \IsTm{\CTmCodeBool}{\CU}
  }
  \break
  \begin{array}{l}
    \EqTy{\CEl{\CTmCodePi{x}{\hat{A}}{\hat{B}}}}{\CTyPi{x}{\CEl{\hat{A}}}{\CEl{\hat{B}}}}\\
    \EqTy{\CEl{\CTmCodeSg{x}{\hat{A}}{\hat{B}}}}{\CTySg{x}{\CEl{\hat{A}}}{\CEl{\hat{B}}}}\\
    \EqTy{\CEl{\CTmCodePath{i.\hat{A}}{a}{b}}}{\CTyPath{i.\CEl{\hat{A}}}{a}{b}}\\
    \EqTy{\CEl{\CTmCodeBool}}{\CTyBool}
  \end{array}
\end{mathflow}

Considering that \textsc{boundary separation} applies to all elements of $\CU$,
we are restricted to connectives that preserve the condition of being boundary
separated. Next, we must include equations specifying the behavior of
$\CKwd{coe}$ and $\CKwd{com}$ on each type code. We begin with coercion,
verifying in each case that the equation is compatible with \textsc{coercion
regularity}.
\begin{align*}
  \CTmCoe{r}{r'}{
    i.\CTmCodePi{x}{\hat{A}}{\hat{B}}
  }{f}
  &=
  \CTmLam{x}{
    \CTmCoe{r}{r'}{
      i.
      \Subst{
        \CTmCoe{r'}{i}{i.\hat{A}}{x}
      }{x}{\hat{B}}
    }{
      f\prn{
        \CTmCoe{r'}{r}{i.\hat{A}}{x}
      }
    }
  }
  \\
  \CTmCoe{r}{r'}{
    i.\CTmCodeSg{x}{\hat{A}}{\hat{B}}
  }{p}
  &=
  \CTmPair{
    \CTmCoe{r}{r'}{i.\hat{A}}{\CTmFst{p}}
  }{
    \CTmCoe{r}{r'}{
      i.
      \Subst{
        \CTmCoe{r}{i}{i.\hat{A}}{\CTmFst{p}}
      }{x}{\hat{B}}
    }{\CTmSnd{p}}
  }
  \\
  \CTmCoe{r}{r'}{
    i.\CTmCodePath{j.\hat{A}}{a_0}{a_1}
  }{p}
  &=
  \CTmAbs{j}{
    \CTmComp{r}{r'}{j}{i.\hat{A}}{
      \_.p(j)
    }
  }
\end{align*}

Of course, \textsc{coercion regularity} implies
$\CTmCoe{r}{r'}{i.\CTmCodeBool}{a} = a$. We additionally observe that the
behavior of homogeneous composition (and thence general composition) is totally
determined by the combination of the above and boundary separation; in
particular, the following equations are derivable by pivoting on $\Bdry{s}$:
\begin{align*}
  \CTmHcomp{r}{r'}{s}{\CTmCodePi{x}{\hat{A}}{\hat{B}}}{i.f}
  &=
  \CTmLam{x}{
    \CTmHcomp{r}{r'}{s}{\hat{B}}{i.f(x)}
  }
  \\
  \CTmHcomp{r}{r'}{s}{\CTmCodeSg{x}{\hat{A}}{\hat{B}}}{i.p}
  &=
  \CTmPair{
    \CTmHcomp{r}{r'}{s}{\hat{A}}{i.\CTmFst{p}}
  }{
    \CTmComp{r}{r'}{s}{
      i.
      \Subst{
        \CTmHcomp{r}{i}{s}{\hat{A}}{i.\CTmFst{p}}
      }{x}{\hat{B}}
    }{i.\CTmSnd{p}}
  }
  \\
  \CTmHcomp{r}{r'}{s}{\CTmCodePath{j.\hat{A}}{a}{b}}{i.p}
  &=
  \CTmAbs{j}{
    \CTmHcomp{r}{r'}{s}{\hat{A}}{
      i.
      p(j)
    }
  }
\end{align*}

\NewDocumentCommand\CTmCodeProd{mm}{#1\mathbin{\hat\times}#2}

\subsubsection{Algorithmic type checking and type-case}\label{sec:typecase}

Although a type checking algorithm for XTT is beyond the scope of this paper,
such an algorithm is important to fully substantiate our claim that XTT can act
as a more tractable alternative to extensional type theory. As a step towards
applying existing type checking algorithms to XTT, we include one final construct.

Most type checking algorithms going back to the work of
Coquand~\cite{coquand:1996} check $M : A$ by first evaluating $A$ to a weak-head
normal form, in order to determine whether $A$ is a dependent product type, a
dependent sum type, the booleans, \etc. Such a determination is crucial because
the head constructor of $A$, in turn, determines how to type check $M$ (\eg, by
applying it to an argument, considering its projections, \etc).

Consider the case that we have a variable $x :
\CTyPath{\_.\CU}{\CTmCodeProd{\hat{A}}{\hat{B}}}{\CTmCodeProd{\hat{C}}{\hat{D}}}$
in scope, and we are attempting to check $\CTmPair{u}{v} \in {x(i)}$ for some
variable $i:\DIM$. The equational rules of XTT do not suggest any reductions for
$x(i)$, so we might na\"{\i}vely return a type error: $\CTmPair{u}{v}$ can only be
an element of a product type, but $\CEl{x(i)}$ appears to be neutral.

However, such a strategy is not complete in the presence of the boundary
separation rule. Suppose that in addition, we have
proofs $p : \CTyPath{\_.\CU}{\hat{A}}{\hat{C}}$ and $q :
\CTyPath{\_.\CU}{\hat{B}}{\hat{D}}$ in scope. Then we may form the path
$\CTmAbs{j}{\CTmCodeProd{p(j)}{q(j)}} :
\CTyPath{\_.\CU}{\CTmCodeProd{\hat{A}}{\hat{B}}}{\CTmCodeProd{\hat{C}}{\hat{D}}}$;
by boundary separation, $\CEl{x(i)} = \CEl{p(i)}\times\CEl{q(i)}$
definitionally, and therefore we \emph{must} proceed to check $u \in p(i)$ and
$v\in q(i)$ (and possibly succeed).

Of course, an algorithm cannot guess out of thin air whether such $p,q$ exist! A
way around this impasse, pioneered in
OTT~\cite{altenkirch-mcbride:2006,altenkirch-mcbride-swierstra:2007}, is to
ensure that from a path between $\hat{A}\mathbin{\hat\times}\hat{B}$ and
$\hat{C}\mathbin{\hat\times}\hat{D}$, we can \emph{always} obtain a path between
$\hat{A}$ and $\hat{C}$. Under those circumstances, we have a uniform
strategy to type check terms on neutral equations between product types (\etc):
ignore the proof of equality, and consider only its boundary.

This approach does \emph{not} make sense for mathematical sets or spaces, since
there are more ways for two product sets to be equal than that their components
are equal, but it does make sense for closed universes, such as the
inductive-recursive universes of Martin-L\"of~\cite{martin-lof:1984}. We discuss
the semantic disadvantages of these closed universes in
Section~\ref{sec:bad-universes}.

Concretely, we achieve this ``injectivity up to paths'' of type constructors in
XTT by including a ``type-case'' operator enabling intensional analysis of
sets~\cite{constable:1982, constable-zlatin:1984,harper-morrisett:1995}.
\[
  \inferrule{
    \IsTy[\CxExtEl<u>{\CU}]{C} \\
    \IsTm[\CxExtEl[\CxExtEl<u>{\CU}]<v>{\CEl{u}\to\CU}]{c_\Pi}{\Subst{\CTmCodePi{x}{u}{v(x)}}{u}{C}} \\
    \IsTm[\CxExtEl[\CxExtEl<u>{\CU}]<v>{\CEl{u}\to\CU}]{c_\Sigma}{\Subst{\CTmCodeSg{x}{u}{v(x)}}{u}{C}} \\
    \IsTm[
      \CxExtEl[
        \CxExtEl[
          \CxExtEl<u_0>{\CU}
        ]<u_1>{\CU}
      ]<u_p>{
        \CTyPath{\_.\CU}{u_0}{u_1}
      },
      x_0\in u_0,
      x_1\in u_1
    ]{c_p}{
      \Subst{
        \CTmCodePath{i.u_p(i)}{x_0}{x_1}
      }{u}{C}
    }\\
    \IsTm{c_b}{\Subst{\CTmCodeBool}{u}{C}}
    \\
    \IsTm{\hat{A}}{\CU}
  }{
    \IsTm{
      \CTmTypecase*{u.C}{\hat{A}}{
        \CKwd{pi}(u,v)\to c_\Pi\\
        \CKwd{sg}(u,v)\to c_\Sigma\\
        \CKwd{path}\prn{u_0,u_1,u_p,x_0,x_1}\to c_p\\
        \CKwd{bool}\to c_b
      }
    }{\Subst{\hat{A}}{u}{C}}
  }
\]

This operator is equipped with the obvious reduction rules:
\begin{align*}
  \CTmTypecase{u.C}{\prn{\CTmCodePi{z}{\hat{A}}{\hat{B}}}}{\cdots}
  &=
  \Subst{\hat{A},\CTmLam{z}{\hat{B}}}{u,v}c_\Pi
  \\
  \CTmTypecase{u.C}{\prn{\CTmCodeSg{z}{\hat{A}}{\hat{B}}}}{
    \cdots
  }
  &=
  \Subst{A,\CTmLam{z}{B}}{u,v}c_\Sigma
  \\
  \CTmTypecase{u.C}{\prn{\CTmCodePath{i.\hat{A}}{a}{b}}}{
    \cdots
  }
  &=
  \Subst{
    \Subst{\CDim0}{i}{\hat{A}},
    \Subst{\CDim1}{i}{\hat{A}},
    \CTmLam{i}{\hat{A}},
    a,
    b
  }{
    u_0,u_1,u_p,x_0,x_1
  }{c_p}
  \\
  \CTmTypecase{u.C}{\CTmCodeBool}{
    \cdots
  }
  &=
  c_b
\end{align*}

\subsection{The LF signature of XTT}%
\label{sec:xtt:lf}

Thus far, we have presented the syntax of XTT as a series of informal inference
rules; formally, however, XTT is the bi-initial object in the 2-category of
models defined in Section~\ref{sec:functorial-semantics}. As a middle ground
between those two styles, we now give a definition of XTT as a sequence of
constants (Figure~\ref{fig:signature}) in a logical framework (LF). By results
of Uemura~\cite{uemura:2019,uemura:2021:thesis}, these constants can be systematically elaborated
into the definitions of Section~\ref{sec:functorial-semantics}, and
automatically induce a 2-category of models with a bi-initial object.
Conversely, connecting these LF constants to standard (but highly annotated)
inference rule presentations requires an \emph{adequacy} theorem~\cite{harper-honsell-plotkin:1993}; we are using Uemura's logical
framework~\cite{uemura:2019,uemura:2021:thesis}, which is adequate for a wide variety of type
theories.

\subsubsection{Uemura's logical framework}

Uemura's LF is essentially a fragment of extensional dependent type theory; the
only important difference is that Uemura's LF stratifies types into two kinds:
representable types $\star$ and types $\Box$. Elements $\LFJdg:\star$ correspond
to judgments that can be hypothesized (such as object-level term judgments),
whereas elements $\LFJdg:\Box$ correspond to arbitrary judgments (such as
object-level type judgments, which generally cannot be hypothesized). We
summarize the rules of Uemura's LF in Figure~\ref{fig:lf}, omitting the standard
definitional equalities for dependent products and extensional equality; readers
can consult Section 5 of~\cite{uemura:2019} for a thorough introduction.%
  \footnote{Note that~\cite{uemura:2019} does not explicitly include dependent
  sums, but these are present in the semantics and do not impact any of the
  results of that paper.}
\begin{figure}
  \begin{mathflow}
    \inferrule{
      \IsRepTy{A}
    }{
      \IsNRepTy{A}
    }
    \split
    \inferrule{
      \IsRepTy{A}
      &
      \IsNRepTy[\CxExtEl{A}]{B}
    }{
      \IsNRepTy{\Prod{x:A}{B}}
    }
    \split
    \inferrule{
      \IsNRepTy{A}
      &
      \IsTm{M_0,M_1}{A}
    }{
      \IsNRepTy{M_0 =\Sub{A} M_1}
    }
    \split
    \inferrule{
      \IsNRepTy{A}
      &
      \IsNRepTy[\CxExtEl{A}]{B}
    }{
      \IsNRepTy{\Sum{x:A}{B}}
    }
    \split
    \inferrule{
      \IsTm[\CxExtEl{A}]{M}{B}
    }{
      \IsTm{\CTmLam{x}{M}}{\Prod{x:A}{B}}
    }
    \split
    \inferrule{
      \IsTm{M}{\Prod{x:A}{B}}
      &
      \IsTm{N}{A}
    }{
      \IsTm{M(N)}{\Subst{N}{x}{B}}
    }
    \split
    \inferrule{
      \IsTm{M}{A}
    }{
      \IsTm{\CTmRefl}{M =\Sub{A} M}
    }
    \split
    \inferrule{
      \IsTm{M_0,M_1}{A}
      &
      \IsTm{N}{M_0 =\Sub{A} M_1}
    }{
      \EqTm{M_0}{M_1}{A}
    }
    \split
    \inferrule{
      \IsTm{M_0}{A}
      &
      \IsNRepTy[\CxExtEl{A}]{B}
      &
      \IsTm{M_1}{\Subst{M_0}{x}{B}}
    }{
      \IsTm{\CTmPair{M_0}{M_1}}{\Sum{x:A}{B}}
    }
    \split
    \inferrule{
      \IsTm{M}{\Sum{x:A}{B}}
    }{
      \IsTm{\CTmFst{M}}{A}
      &
      \IsTm{\CTmSnd{M}}{\Subst{\CTmFst{M}}{x}{B}}
    }
  \end{mathflow}
  \caption{A summary of Uemura's logical framework.}%
  \label{fig:lf}
\end{figure}

In addition to the model theory and initiality result mentioned above, another
major advantage of using a logical framework is \emph{higher-order abstract
syntax}, in which variable binders are encoded using the LF's
$\lambda$-abstractions over representable types. As a result, unlike inference
rules, LF encodings need not explicitly represent the object-level contexts, and
in fact automatically ensure that all operations are stable under substitution.

\begin{figure}
  \begin{gather*}
  \begin{aligned}
    \DIM &: \star \\
    \COND,\LFTp &: \Box \\
    \PrfFam{-} &: \COND\to\star \\
    \LFTm &: \LFTp\to\star \\
    \Dim0, \Dim1 &: \DIM \\
    \prn{=} &: \DIM\times\DIM\to\COND \\
    \prn{\lor} &: \COND\times\COND\to\COND
  \end{aligned}
  \qquad\qquad
  \begin{aligned}
    \_ &: \Params{\phi} \Prod{p,q:\PrfFam{\phi}}{p =\Sub{\PrfFam{\phi}} q} \\
    \_ &: \Params{r,s} \prn{r=\Sub{\DIM}s}\cong\PrfFam{r=s} \\
    \_ &: \Params{\phi,\psi} \PrfFam{\phi}\to\PrfFam{\phi\lor\psi} \\
    \_ &: \Params{\phi,\psi} \PrfFam{\psi}\to\PrfFam{\phi\lor\psi} \\
    \Bdry{r} &= \prn{r=\Dim0 \lor r=\Dim1}
    \\& \\&
  \end{aligned}
  \\[8pt]
  \begin{aligned}
    \multispan{2}{\underline{\textbf{For each $\LFJdg\in\brc{\LFTp,\PrfFam{\phi},\LFTm\prn{A}}$:}\hfill}}\\
    \Ext{\LFJdg}{\phi}{x\Sub{\phi}} &= \Sum{x:\LFJdg} \Prod{p:\PrfFam{\phi}}
      x =\Sub{\LFJdg} x\Sub{\phi}\prn{p} \\
    \Kwd{abort}_\LFJdg &: \PrfFam{\Dim0=\Dim1} \to \LFJdg \\
    \_ &: \Prod{p:\PrfFam{\Dim0=\Dim1}} \Prod{x:\LFJdg} x =_\LFJdg \Kwd{abort}_\LFJdg\prn{p}
    \\&
  \end{aligned}
  \begin{aligned}
    &\\
    \Kwd{split}_\LFJdg &: \Params{\phi,\psi}
      \Prod{x\Sub{\phi}:\PrfFam{\phi}\to\LFJdg}
      \Prod{x\Sub{\psi}:\Ext{\LFJdg}{\phi}{x\Sub{\phi}}}
      \PrfFam{\phi\lor\psi}\to\LFJdg \\
    \_ &: \Params{\phi,\psi} \Prod{\_:\PrfFam{\phi}}
      \Kwd{split}_\LFJdg\prn{x\Sub{\phi},x\Sub{\psi}} =\Sub{\LFJdg} x\Sub{\phi} \\
    \_ &: \Params{\phi,\psi} \Prod{\_:\PrfFam{\psi}}
      \Kwd{split}_\LFJdg\prn{x\Sub{\phi},x\Sub{\psi}} =\Sub{\LFJdg} x\Sub{\psi} \\
    \_ &: \Params{\phi,\psi,x} \Prod{\_:\PrfFam{\phi\lor\psi}}
      x =\Sub{\LFJdg} \Kwd{split}_\LFJdg\prn{x,x}
  \end{aligned}
  \end{gather*}
  \\[8pt]
  \begin{align*}
    \TyPi,\TySg &: \prn{\Sum{A:\LFTp}{\prn{\LFTm\prn{A}\to\LFTp}}}\to\LFTp \\
    \TyPath &: \prn{\Sum{A:\DIM\to\LFTp}{\LFTm\prn{A\prn{\Dim0}}\times\LFTm\prn{A\prn{\Dim1}}}}\to\LFTp \\
    \TyBool,\Kwd{set} &: \LFTp \\
    \Kwd{pi/tm} &: \Params{A,B} \prn{\Prod{x:\LFTm\prn{A}}{\LFTm\prn{B\prn{x}}}}
      \cong \LFTm\prn{\TyPi\prn{A,B}} \\
    \Kwd{sg/tm} &: \Params{A,B} \prn{\Sum{x:\LFTm\prn{A}}{\LFTm\prn{B\prn{x}}}}
      \cong \LFTm\prn{\TySg\prn{A,B}} \\
    \Kwd{path/tm} &: \Params{A,a_0,a_1}
      \prn{\Prod{i:\DIM}{\Ext{\LFTm\prn{A\prn{i}}}{\Bdry{i}}{
        \Kwd{split}\Sub{\LFTm\prn{A\prn{i}}}\prn{a_0,a_1}}}}
      \cong \LFTm\prn{\TyPath\prn{A,a_0,a_1}} \\
    \TmTt,\TmFf &: \LFTm\prn{\TyBool} \\
    \Kwd{if} &: \Prod{C:\LFTm\prn{\TyBool}\to\LFTp} \Prod{b:\LFTm\prn{\TyBool}}
      \LFTm{\prn{C\prn{\TmTt}}}\to\LFTm{\prn{C\prn{\TmFf}}}\to\LFTm{\prn{C\prn{b}}} \\
    \_ &: \Params{C,c_{\TmTt},c_{\TmFf}}
      \Kwd{if}\prn{C,\TmTt,c_{\TmTt},c_{\TmFf}} =\Sub{\LFTm\prn{C\prn{\TmTt}}} c_{\TmTt} \\
    \_ &: \Params{C,c_{\TmTt},c_{\TmFf}}
      \Kwd{if}\prn{C,\TmFf,c_{\TmTt},c_{\TmFf}} =\Sub{\LFTm\prn{C\prn{\TmFf}}} c_{\TmFf} \\
    \Kwd{el} &: \LFTm\prn{\Kwd{set}}\to\LFTp \\
    \_ &:
      \Prod{\hat{A}:\LFTm\prn{\Kwd{set}}}
      \Prod{r : \DIM}
      \Prod{a,b : \LFTm\prn{\Kwd{el}\prn{\hat{A}}}}
      \prn{\PrfFam{\Bdry{r}}\to a =\Sub{\LFTm\prn{\Kwd{el}\prn{\hat{A}}}} b}\to
      a =\Sub{\LFTm\prn{\Kwd{el}\prn{\hat{A}}}} b \\
    \Kwd{coe} &:
      \Prod{\hat{A}:\DIM\to\LFTm\prn{\Kwd{set}}}
      \Prod{r,r':\DIM}
      \Prod{a:\LFTm\prn{\Kwd{el}\prn{\hat{A}\prn{r}}}}
      \Ext{\LFTm\prn{\Kwd{el}\prn{\hat{A}\prn{r'}}}}{r=r'}{a} \\
    \_ &:
      \Prod{\hat{A}:\LFTm\prn{\Kwd{set}}}
      \Prod{r,r':\DIM}
      \Prod{a:\LFTm\prn{\Kwd{el}\prn{\hat{A}}}}
      \Kwd{coe}\prn{\lambda\_.\hat{A},r,r',a} =\Sub{\LFTm\prn{\Kwd{el}\prn{\hat{A}}}} a \\
    \Kwd{com} &:
      \Params{\hat{A}}
      \Prod{s,r,r':\DIM}
      \Prod{a:\Prod{i:\DIM}\Prod{\_:\PrfFam{i=r\lor\Bdry{s}}}{\LFTm\prn{\Kwd{el}\prn{\hat{A}\prn{i}}}}}
      \Ext{\LFTm\prn{\Kwd{el}\prn{\hat{A}\prn{r'}}}}{r'=r\lor\Bdry{s}}{a\prn{r'}} \\
    \TmCodePi,\TmCodeSg &:
      \prn{\Sum{\hat{A}:\LFTm\prn{\Kwd{set}}}{\prn{\LFTm\prn{\Kwd{el}\prn{\hat{A}}}\to\LFTm\prn{\Kwd{set}}}}}\to\LFTm\prn{\Kwd{set}} \\
    \_ &: \Params{\hat{A},\hat{B}}
      \Kwd{el}\prn{\TmCodePi\prn{\hat{A},\hat{B}}} =\Sub{\LFTp}
      \TyPi\prn{\Kwd{el}\prn{\hat{A}},\lambda x.\Kwd{el}\prn{\hat{B}\prn{x}}} \\
    \_ &: \Params{\hat{A},\hat{B},r,r',f}
      \Kwd{coe}\prn{\lambda i.\TmCodePi\prn{\hat{A}\prn{i},\hat{B}\prn{i}},r,r',f}
      =\Sub{\LFTm\prn{\Kwd{el}\prn{\TmCodePi\prn{\hat{A}\prn{r'},\hat{B}\prn{r'}}}}}
      \lambda x.\Kwd{coe}\prn{\dots}
  \end{align*}
  \caption{The signature of XTT in Uemura's logical framework. For space
  reasons, we omit the remainder of the rules pertaining to the universe of
  Bishop sets.}%
  \label{fig:signature}
\end{figure}

\subsubsection{The signature of XTT}

In Figure~\ref{fig:signature}, we present XTT as a signature, or sequence of
constants, in Uemura's LF\@. In accordance with the judgments-as-types methodology~\cite{harper-honsell-plotkin:1993}, we render each judgment of XTT as a type
constant in its LF signature. We encode the judgment $\IsDim$ as a nullary
representable type with two elements; it is representable because XTT allows
context extension by dimension variables $\prn{\CxExtDim}$.
\begin{mathflow}
  \DIM : \star
  \split
  \Dim0 : \DIM
  \split
  \Dim1 : \DIM
\end{mathflow}

We encode the face formula judgment $\IsCond$ as a (non-representable) type
equipped with a representable decoding function $\PrfFam{-}$; contexts cannot be
extended by face formula variables, but we can extend contexts by an assumption
that $\phi$ holds, $\prn{\CxExtPrf}$.
\begin{mathflow}
  \COND : \Box \split \PrfFam{-} : \COND \to \star
\end{mathflow}
Note that $\PrfFam{-}$ is silent in our earlier notation; in both the inference
rules and in Figure~\ref{fig:signature}, we suppress the instantiation of
partial elements by (the unique) proofs of $\PrfFam{\phi}$.

Likewise, we encode $\IsTy$ as a (non-representable) type and $\IsTm{a}{A}$ as a
representable type family.
\begin{mathflow}
  \LFTp : \Box \split \LFTm : \LFTp \to \star
\end{mathflow}
We do not add constants for the equality judgments of XTT, as these are encoded
by LF equality at the corresponding type.

Finally, we encode the connectives and rules of XTT as constants, taking
advantage of the LF's type structure for brevity. For example, the formation
rule for dependent products corresponds to the following LF constant:
\[
  \IsTm[\CxExtEl[A : \LFTp]<B>{\Prod{x:\LFTm\prn{A}}{\LFTp}}]{\Kwd{pi}\prn{A,B}}{\LFTp}
\]
The force of the remaining rules for dependent product types is simply to assert
an isomorphism between $\LFTm\prn{\Kwd{pi}\prn{A,B}}$ and
$\Prod{x:\LFTm\prn{A}}{\LFTm\prn{B(a)}}$. To express this pattern concisely, we
introduce the notation $\IsTm{f}{A \cong B}$ for the following four constants:
\begin{mathflow}
  \IsTm[\CxExtEl{A}]{f^{\rightarrow}(x)}{B}
  \split
  \IsTm[\CxExtEl{B}]{f^{\leftarrow}(x)}{A}
  \break
  \IsTm[\CxExtEl{A}]{\_}{f^{\leftarrow}(f^{\rightarrow}(x)) =\Sub{A} x}
  \split
  \IsTm[\CxExtEl{B}]{\_}{f^{\rightarrow}(f^{\leftarrow}(x)) =\Sub{B} x}
\end{mathflow}
Using this notation, we encode the introduction, elimination, $\beta$, and
$\eta$ laws of dependent products in one stroke:
\[
  \IsTm[
    \CxExtEl[A : \LFTp]<B>{\Prod{x:\LFTm\prn{A}}{\LFTp}}
  ]{\Kwd{pi/tm}}{
    \Prod{x:\LFTm{\prn{A}}}{\LFTm\prn{B(x)}} \cong \LFTm\prn{\Kwd{pi}\prn{A,B}}
  }
\]

\NewDocumentCommand{\CanonicalCojoneDiagram}{D||{}mmmm}{%
  \DelimMin{1}
  \tikz[inline diagram,baseline = (2.base), #1]{
    \node (0) {$\Comma{#2}{\brc{#5}}$}; 
    \node (1) [right = .5cm of 0] {$#3$}; 
    \node (2) [right = .5cm of 1] {$#4$}; 
    \path[->] (0) edge node [above,yshift = .3ex] {$\MathSmall{\partial_0}$} (1);
    \path[->] (1) edge node [above,yshift = .3ex] {$\MathSmall{#2}$} (2);
  }
}

\section{Categorical preliminaries}

All the categorical machinery we assume can be found in standard introductory
textbooks and
references~\cite{maclane:1998,borceux:1994:vol1,borceux:1994:vol2,borceux:2010:vol3,awodey:2010,johnstone:2002};
in order to fix notations and render our presentation as self-contained as
possible, however, we have included a number of definitions. The beginning of
Section~\ref{sec:basic-categories} recalls basic categorical definitions and
fixes notation, while Section~\ref{sec:presheaves} covers more specific facts
related to presheaves and representables.
The remaining sections discuss functorial semantics and the categorical
machinery necessary to account for the semantics of type theory.

We recommend only skimming this section on a first read, and return to it
as needed.

\subsection{Basic categorical definitions}%
\label{sec:basic-categories}

\begin{nota}
  Given a category $\CCat$ and objects $C,D:\CCat$, we write
  $\Hom[\CCat]{C}{D}$ for the collection of arrows between $C$ and $D$. We will
  also write $\Hom{\CCat}{\DCat}$ for the \emph{category} of functors
  $\Mor{\CCat}{\DCat}$ and natural transformations between them.
\end{nota}

\begin{conv}
  Conventionally, we write $\SET$ and $\CAT$ for the categories of sets and
  categories respectively; of course, to be more precise, we should instead refer
  to $\SET_\alpha$ and $\CAT_\alpha$ for some strongly inaccessible cardinal
  $\alpha$, or equivalently a Grothendieck universe $\mathscr{U}$. We leave the
  resolution of these universes implicit, noting them explicitly in sensitive
  places.
\end{conv}

\begin{nota}
  We will write $\Simplex{n}$ for the $n$-simplex regarded as a category; in
  particular, $\Simplex{0}$ is the terminal category $\brc{*}$, and
  $\Simplex{1}$ is the category $\brc{\Mor{\bullet}{\circ}}$ of the walking
  arrow. Therefore $\ArrCat{\CCat}$ is the category of arrows and commutative
  squares in $\CCat$.
\end{nota}

\begin{defi}[Cartesian arrow category]
  We write $\Mor|embedding|{\CartArr{\CCat}}{\ArrCat{\CCat}}$ for the wide
  subcategory of arrows and \emph{cartesian} squares between them.
  Concretely, given $f,g : \CartArr{\CCat}$ a morphism between them exhibits $f$ as a
  pullback of $g$:
  \[
    \DiagramSquare{
      nw/style = pullback,
      nw = \partial_0f,
      sw = \partial_1f,
      ne = \partial_0f,
      se = \partial_1g,
      west = f,
      east = g
    }
    \qedhere
  \]
\end{defi}

\begin{defi}[Cartesian morphism]
  Given a functor $\Mor[p]{\ECat}{\BCat}$, a morphism $\Mor[f]{E_0}{E_1}$ is \emph{cartesian} if for
  every $\Mor[g]{E'}{E_1}$ and $\Mor[u]{pE'}{pE_0}$ such that $pg = pf \circ u$, there
  exists a unique $\Mor[h]{E'}{E_0}$ over $u$ such that $f \circ h = g$. Diagrammatically:
  \begin{equation*}
    \begin{tikzpicture}[diagram,baseline = (sw.base)]
      \SpliceDiagramSquare{
        west/style = |->,
        east/style = |->,
        nw = E_0,
        ne = E_1,
        sw = pE_0,
        se = pE_1,
        north/node/style = upright desc,
        north = f,
        south = pf,
        width = 2.5cm,
        height = 1.5cm,
      }

      \node (E') [left = of nw,yshift = 1.5cm,] {$E'$};
      \node (B') [below = 1.5cm of E'] {$pE'$};

      \path[->] (B') edge node [below,sloped] {$u$} (sw);
      \path[|->] (E') edge (B');
      \path[->,exists] (E') edge node [desc] {$h$} (nw);
      \path[->, bend left = 20] (E') edge node [above,sloped] {$g$} (ne);

      \node (tot) [right = of ne] {$\ECat$};
      \node (base) [right = of se] {$\BCat$};
      \path[->] (tot) edge (base);
    \end{tikzpicture}
    \qedhere
  \end{equation*}
\end{defi}

\begin{defi}[Opcartesian morphism]
  Dually, given a functor $\Mor[p]{\ECat}{\BCat}$, a morphism
  $\Mor[f]{E_0}{E_1}$ is \emph{opcartesian} if for every $\Mor[g]{E_0}{E'}$ and
  $\Mor[u]{pE_1}{pE'}$ such that $pg = u \circ pf$, there
  exists a unique factor $\Mor[h]{E_1}{E'}$ lying over $u$ such that $h\circ f
  = g$:

  \begin{equation*}
    \begin{tikzpicture}[diagram,baseline=(sw.base)]
      \SpliceDiagramSquare{
        west/style = |->,
        east/style = |->,
        nw = E_0,
        ne = E_1,
        sw = pE_0,
        se = pE_1,
        north/node/style = upright desc,
        north = f,
        south = pf,
        width = 2.5cm,
        height = 1.5cm,
      }
      \node (E') [right = of ne, yshift = 1.5cm,] {$E'$};
      \node (B') [below = 1.5cm of E'] {$pE'$};

      \path[->] (se) edge node [below,sloped] {$u$} (B');
      \path[|->] (E') edge (B');
      \path[->,bend left=20] (nw) edge node [above,sloped] {$g$} (E');
      \path[->,exists] (ne) edge node [desc] {$h$} (E');

      \node (tot) [right = 3cm of ne] {$\ECat$};
      \node (base) [right = 3cm of se ] {$\BCat$};
      \path[->] (tot) edge (base);
    \end{tikzpicture}
    \qedhere
  \end{equation*}
\end{defi}

\begin{defi}[Fibration]
  A \emph{fibration} is a functor $\Mor[p]{\ECat}{\BCat}$ such that for each
  morphism $\Mor[u]{B}{pE}$ there exists a cartesian morphism
  $\Mor[u^\dagger{E}]{u^*E}{E}$ lying over $u$:
  \begin{equation*}
    \begin{tikzpicture}[diagram,baseline = (sw.base)]
      \SpliceDiagramSquare{
        west/style = lies over,
        east/style = lies over,
        sw = B,
        se = pE,
        south = u,
        nw = u^*E,
        ne = E,
        north = u^\dagger{E},
        north/style = exists,
        height = 1.5cm,
      }
      \node (tot) [right = of ne] {$\ECat$};
      \node (base) [right = of se] {$\BCat$};
      \path[fibration] (tot) edge (base);
    \end{tikzpicture}
  \end{equation*}

  We emphasize the property of an arrow being a fibration by using an open
  triangular tip, \eg $\FibMor[p]{\ECat}{\BCat}$.
\end{defi}

\begin{defi}[Opfibration]
  Dually, an \emph{opfibration} is a functor $\Mor[p]{\ECat}{\BCat}$ such that
  for each morphism $\Mor[u]{pE}{B}$ there exists an opcartesian morphism
  $\Mor[u^\dagger{E}]{E}{u_!E}$ lying over $u$:
  \begin{equation*}
    \begin{tikzpicture}[diagram,baseline = (sw.base)]
      \SpliceDiagramSquare{
        west/style = op lies over,
        east/style = op lies over,
        se = B,
        sw = pE,
        south = u,
        ne = u_!E,
        nw = E,
        north = u^\dagger{E},
        north/style = exists,
        height = 1.5cm,
      }
      \node (tot) [right = of ne] {$\ECat$};
      \node (base) [right = of se] {$\BCat$};
      \path[opfibration] (tot) edge (base);
    \end{tikzpicture}
  \end{equation*}
  We emphasize the property of an arrow being an opfibration by using a filled triangular tip, \eg $\OpFibMor[p]{\ECat}{\BCat}$.
\end{defi}

\begin{fact}
  The codomain functor
  $\Mor[\partial_1]{\ArrCat{\CCat}}{\CCat}$ which sends $\Mor[f]{X}{Y}$ to $Y$
  is always an opfibration, with opcartesian lifts implemented by postcomposition
  (dependent sum).
\end{fact}

\begin{fact}
  In a category $\CCat$ with pullbacks, the codomain functor
  $\Mor[\partial_1]{\ArrCat{\CCat}}{\CCat}$ is a fibration.
  Cartesian lifts may are implemented by pullbacks in $\CCat$.
\end{fact}

\begin{defi}[Comma category]
  Given a pair of functors $\Mor[F]{\DCat}{\CCat}$, $\Mor[G]{\ECat}{\CCat}$,
  the \emph{comma category} $\Comma{F}{G}$ has as objects arrows
  $\Mor[X]{FD}{GE}$ and commutative squares of the following kind for
  arrows:
  \begin{equation*}
    \DiagramSquare{
      width = 2.5cm,
      nw = FD,
      sw = GE,
      ne = FD',
      se = GE',
      north = Fd,
      south = Ge,
      west = X,
      east = X',
    }
    \qedhere
  \end{equation*}
\end{defi}

The comma category $\Comma{F}{G}$ may be constructed more abstractly in terms
of the following (1-categorical) pullback in $\CAT$, the category of categories
and functors:

\begin{equation}
  \DiagramSquare{
    nw/style = pullback,
    width = 2.5cm,
    nw = \Comma{F}{G},
    ne = \ArrCat{\CCat},
    east = \prn{\partial_0,\partial_1},
    se = \CCat\times\CCat,
    south = \prn{F,G},
    sw = \DCat\times\ECat
  }
\end{equation}

\begin{nota}
  Let $X:\CCat$; we will write $\Mor[\brc{X}]{\Simplex{0}}{\CCat}$ for the constant functor $* \mapsto X$.
\end{nota}

\begin{nota}
  A common abuse of notation in the comma construction is that, when either $F$
  or $G$ is the identity functor $\Mor[\ArrId{\CCat}]{\CCat}{\CCat}$, they
  shall be written simply $\CCat$. For instance, $\Comma{\CCat}{G}$ is written for $\Comma{\ArrId{\CCat}}{G}$.
\end{nota}

An important instance of the comma construction is the \emph{slice} category.

\begin{defi}[Slice category]
  Given an object $X : \CCat$, the \emph{slice} or ``over-category'' of $\CCat$
  at $X$ is the comma category $\Sl{\CCat}{X} = \Comma{\CCat}{\brc{X}}$. The
  objects of $\Sl{\CCat}{X}$ can be seen to be arrows $\Mor{Y}{X}$; morphisms
  in the slice are commutative triangles.
\end{defi}

\begin{fact}
  In a category with pullbacks $\CCat$, a morphism $\Mor[f]{X}{Y}$ induces a functor
  $\Mor[f^*]{\Sl{\CCat}{Y}}{\Sl{\CCat}{X}}$ sending $\Mor[g]{Z}{Y}$ to
  $\Mor[f^*g]{Z \times_Y X}{X}$.
\end{fact}

\begin{nota}
  We will adopt a common abuse of notation and suppress the \emph{weakening} functor
  $\Mor[!_X^*]{\CCat}{\Sl{\CCat}{X}}$ when it is unambiguous.
\end{nota}

\begin{defi}
  Let $\kappa$ be a regular cardinal; we say that a category $\CCat$ is
  $\kappa$-(co)complete when $\CCat$ has all (co)limits of $\kappa$-small diagrams;
  a functor that preserves these (co)limits is called $\kappa$-(co)continuous.
  When $\kappa$ is omitted, a sufficiently large strongly inaccessible cardinal
  is assumed.
\end{defi}

\subsubsection{Presheaves, representability, and discrete fibrations}\label{sec:presheaves}

\begin{defi}
  A \emph{presheaf} on $\CCat$ is a functor $\Mor[F]{\OpCat{\CCat}}{\SET}$; the
  category of presheaves $\Hom{\OpCat{\CCat}}{\SET}$ is written $\Psh{\CCat}$.
\end{defi}

Presheaves capture the geometric intuition of probing a space or other object
by small figures: the role of contexts and substitutions in (strict) type
theory supplies type theorists and logicians with a useful concrete intuition
for presheaves.
A more structural perspective on presheaves is, however, essential: the
category $\Psh{\CCat}$ may be characterized universally as the
\emph{free cocompletion} of $\CCat$, equipping $\CCat$ with new colimits.
When $\CCat$ already has some colimits, it is important to note that the new
ones do \emph{not} coincide with the old ones.

\begin{con}[The Yoneda embedding]
  To be more precise, there is a universal functor
  $\Mor[\Yo[\CCat]]{\CCat}{\Psh{\CCat}}$, called the \emph{Yoneda embedding},
  taking each object $C:\CCat$ to a ``formal colimit'' $\Yo[\CCat]{C} =
  \Hom[\CCat]{\bullet}{C}$, such that every functor
  $\Mor{\CCat}{\ECat}$ with $\ECat$ cocomplete factors as $\Yo[\CCat]$ and a cocontinuous functor $\tilde{F}$ in
  an essentially unique way:
  \begin{equation*}
    \begin{tikzpicture}[diagram,baseline = (PrC.base)]
      \node (C) {$\CCat$};
      \node (E) [right = 3cm of C] {$\ECat$};
      \node (PrC) [between = C and E, yshift = -2.5cm] {$\Psh{\CCat}$};
      \path[->] (C) edge node [above] {$F$} (E);
      \path[->] (C) edge node [sloped,below] {$\Yo[\CCat]$} (PrC);
      \path[->,exists] (PrC) edge node [sloped,below] {$\tilde{F}$} (E);
    \end{tikzpicture}
    \qedhere
  \end{equation*}
\end{con}

\begin{lem}[Yoneda]\label{lem:yoneda}
  For each presheaf $X:\Psh{\CCat}$, we have the following isomorphism:
  \[
    \Hom[\Psh{\CCat}]{\Yo[\CCat]{C}}{X} \cong X\prn{C}
  \]
  As a consequence, the Yoneda embedding is full and faithful.
\end{lem}

\begin{defi}
  A presheaf $X : \Psh{\CCat}$ is called \emph{representable} when it lies in
  the essential image of $\Yo[\CCat]$, \ie $X$ is isomorphic to $\Yo[\CCat]{C}$
  for some $C:\CCat$.
\end{defi}

The notion of representable object is extended to maps in a canonical way, by
considering fibers over representable objects.

\begin{defi}[Representable natural transformation]\label{def:representable-nat-trans}
  A representable natural transformation is a map $\Mor[f]{Y}{X}:\Psh{\CCat}$ whose every fiber over
  a representable object is representable. In other words, the fiber product of $f$ with any
  $\Mor[x]{\Yo[\CCat]{C}}{X}$ is representable.
  \begin{equation*}
    \DiagramSquare{
      nw/style = pullback,
      nw = \Yo[\CCat]\prn{C.x},
      sw = \Yo[\CCat]{C},
      ne = Y,
      se = X,
      west = x^*f,
      south = x,
      east = f,
    }
  \end{equation*}
  In Section~\ref{sec:taichi}, we observe that the representing object of this fiber product plays a
  role analogous to context extension, so we denote it $C.x$.
\end{defi}

\NewDocumentCommand\CatEl{mm}{#1/#2}

Representable natural transformations are a prime example of Grothendieck's
``relative point of view'', extending a notion that is first defined on
objects to have sense on \emph{morphisms}.
It is useful to remark that the slice $\Sl{\Psh{\CCat}}{X}$ is itself the
category of presheaves $\Psh{\CatEl{\CCat}{X}}$ on the \emph{category of
elements} $\CatEl{\CCat}{X}$ of $X$, and that the representability of the map
$\Mor[f]{Y}{X}$ agrees with the representability of the object $f :
\Sl{\Psh{\CCat}}{X}$.

\begin{con}[Category of elements]
  The \emph{category of elements} $\CatEl{\CCat}{X}$ of a presheaf $X : \Psh{\CCat}$ has as
  objects pairs $C/x$ with $x \in X(C)$ and morphisms $\Mor[f^\dagger
  x]{D/f^*x}{C/x}$ for each $\Mor[f]{D}{C}$ and $x\in X(C)$.
\end{con}

In fact, the category of elements of a presheaf $X:\Psh{\CCat}$ is the total
category of a \emph{discrete fibration} over $\CCat$.

\begin{defi}
  A fibration $p : \ECat \to \BCat$ is discrete if $pf = \ArrId{\BCat}$
  implies that $f = \ArrId{\ECat}$; equivalently, if the fibers of $p$ are
  discrete categories.
  The collection of discrete fibrations over $\CCat$ forms a full subcategory
  $\DF{\CCat}\subseteq\mathbf{Fib}_{\CCat} \subseteq \Sl{\CAT}{\CCat}$, with
  morphisms given by commuting triangles.
\end{defi}

\begin{lem}
  Let $X:\Psh{\CCat}$ be a presheaf; the functor
  $\Mor[p_X]{\CatEl{\CCat}{X}}{\CCat}$ that takes each $C/x$ to $C$ is a
  discrete fibration. Conversely, letting $\Mor[F]{\ECat}{\CCat}$ be a discrete fibration, we
  may define a presheaf $F_\bullet:\Psh{\CCat}$ in which each $F_C$ is the pullback of $F$ along $\Mor[!_C]{\Simplex{0}}{\CCat}$.
\end{lem}

As might be expected, the assignment $X \mapsto p_X$ extends to a functor
$\Mor[p_\bullet]{\Psh{\CCat}}{\DF{\CCat}}$ which is full, faithful, and
essentially surjective (\ie an equivalence of categories). Therefore
$\DF{\CCat}$ may be used as an alternative to $\Psh{\CCat}$, and has its own
Yoneda embedding $\Mor|embedding|[\DFYo[\CCat]]{\CCat}{\DF{\CCat}}$.

In the context of discrete fibrations, there is an alternative characterization
of representable maps in addition to
the Grothendieck-style extension of the essential image of the Yoneda embedding
$\Mor|embedding|[\DFYo[\CCat]]{\CCat}{\DF{\CCat}}$ to maps via pullbacks.

\begin{lem}[Representability in discrete fibrations~\cite{awodey-butz-simpson-streicher:2014,awodey:2018:natural-models}]
  A map $\Mor[f]{G}{F} : \DF{\CCat}$ is representable iff the \emph{upstairs}
  functor $\Mor[\partial_0f]{\partial_0G}{\partial_0F}:\CAT$ has a
  (non-fibered) right adjoint, which we may write $\partial_0f \dashv
  \ReprAdj{f}$.
\end{lem}

In contrast with presheaves, it is natural to simultaneously work with discrete
fibrations over different base categories: we may write
$\DF{}\subseteq\mathbf{Fib}\subseteq\ArrCat{\CAT}$ for the category of fibrations
over arbitrary base, rendered a full subcategory of the category of arrows of
$\CAT$. Therefore, a morphism of discrete fibrations is a commuting square:
\[
  \DiagramSquare{
    nw = \ECat_0,
    sw = \BCat_0,
    ne = \ECat_1,
    se = \BCat_1,
    west = p_0,
    east = p_1,
    north = \alpha_1,
    south = \alpha_0,
  }
\]

\begin{lem}
  The map sending a discrete fibration to its base category $\FibMor[\partial_1]{\DF{}}{\CAT}$ is a
  fibration. We will write $\DF{\CCat}$ for the fiber $\partial_1[\CCat]$.
\end{lem}

\subsubsection{Density of the Yoneda embedding}

The Yoneda lemma (Lemma~\ref{lem:yoneda}) is indispensable, but the closely
related \emph{density theorem} more deeply exposes the character of
$\Psh{\CCat}$ as free cocompletion.

\begin{defi}[Density]
  A functor $\Mor[F]{\CCat}{\ECat}$ is called \emph{dense} if for each
  $E:\ECat$, the canonical cocone $\Mor[\delta_F^E]{\mathsf{D}_F^E}{\brc{E}}$
  is universal (\ie a colimit).
\end{defi}

\begin{lem}[Density]\label{lem:dual-yoneda}
  The Yoneda embedding $\Mor|embedding|[\Yo[\CCat]]{\CCat}{\Psh{\CCat}}$ is a
  dense functor.
\end{lem}

Of course, to understand Lemma~\ref{lem:dual-yoneda} we must first have an
understanding of density in the category theoretic sense. Every functor
$\Mor[F]{\CCat}{\ECat}$ generates a \emph{canonical cocone} over each object
$E:\ECat$ for the following diagram $\Mor[\mathsf{D}_F^E]{\Comma{F}{\brc{E}}}{\ECat}$:
\[
  \mathsf{D}_F^E \defeq
  \CanonicalCojoneDiagram{F}{\CCat}{\ECat}{E}
\]

The canonical cocone $\Mor[\delta_F^E]{\mathsf{D}_F^E}{\brc{E}}$ takes each
$\alpha_i : \Comma{F}{\brc{E}}$ to the underlying map
$\Mor[\alpha_i]{F\prn{E_i}}{E}$. Visually, the cocone might look something like
this:
\begin{equation}
  \begin{tikzpicture}[diagram,baseline=(E.base)]
    \node (E) {$E$};
    \node (F/Ei) [left = 5cm of E] {$F{E_i}$};
    \node (F/Ej) [above = of F/Ei] {$F{E_j}$};
    \node (F/Eh) [below = of F/Ei] {$F{E_h}$};
    \node (F/Ek) [between = F/Ei and F/Eh] {$F{E_k}$};
    \path[->] (F/Ei) edge node [desc] {$\alpha_i$} (E);
    \path[->] (F/Ej) edge node [desc] {$\alpha_j$} (E);
    \path[->] (F/Ek) edge node [desc] {$\alpha_k$} (E);
    \path[->] (F/Eh) edge node [desc] {$\alpha_h$} (E);
    \path[->] (F/Ej) edge node [left] {$F{f}$} (F/Ei);
  \end{tikzpicture}
\end{equation}

Therefore, Lemma~\ref{lem:dual-yoneda} says exactly that every presheaf is a
formal colimit of representable objects in a canonical way.

\subsubsection{Philo-logie and Diaconescu's theorem}

Categories of presheaves $\ECat = \Psh{\CCat}$ are complete, cocomplete,
locally cartesian closed, and exhibit certain non-trivial compatibilities
between certain limits and colimits which may be boiled down to the
existence of a classifying family for subobjects (monomorphisms). This
\emph{subobject classifier} is a monomorphism
$\Mor|>->|[\PropTrue[\ECat]]{\ObjTerm{\ECat}}{\Prop[\ECat]}$ that is universal
in the sense that all monomorphisms arise \emph{in a unique way} from
$\PropTrue[\ECat]$ by pullback:
\begin{equation}
  \DiagramSquare{
    nw/style = pullback,
    west/style = {>->},
    east/style = {>->},
    south/style = exists,
    nw = Y,
    sw = X,
    west = m,
    ne = \ObjTerm{\ECat},
    se = \Prop[\ECat],
    east = \PropTrue[\ECat],
    south = {\color{gray}\exists!} \widetilde{m}
  }
\end{equation}

A category with these properties may be referred to as a \emph{logos}
following the terminology of Anel and Joyal~\cite{anel-joyal:2021}, though we
defer the actual definition of logoi until Definition~\ref{def:logos}. Recalling our
characterization of $\Psh{\CCat}$ as the \emph{free} cocompletion of $\CCat$,
we argue that the correct way to understand categories of presheaves is as a
class of logoi that are generated in a specific way: namely, all colimits are
added freely without imposing any relations.

\NewDocumentCommand\Opens{m}{\mathscr{O}\prn{#1}}

We may also generate logoi in which the added colimits satisfy some relations;
historically, the most common way to do so is to augment the category $\CCat$
with the data of a \emph{coverage}. The canonical motivating example arises
when considering presheaves on the frame of open sets $\Opens{X}$ of a
topological space $X$. The logos of presheaves $\Psh{\Opens{X}}$ does
\emph{not} have the geometrically correct colimits corresponding to the gluing
together of components of an open cover $\brc{\Mor{U_i}{U}}$, but some
presheaves treat $\Yo[\Opens{X}]{U}$ ``as if'' it were the appropriate colimit of the covering diagram to varying degrees:
\begin{enumerate}
  \item A presheaf $X$ that has \emph{no more than} one section $x \in
    X{U}$ compatible with a family of sections $\brc{x_i\in X{U_i}}$
    defined on the cover is called \emph{separated}.
  \item A presheaf $X$ that has exactly one such section $x\in X{U}$ for
    each such compatible family of sections is called \emph{local}.\footnote{Typically local
      presheaves are referred to as sheaves, but we adopt this non-standard terminology for symmetry
    with separated presheaves.}
\end{enumerate}

\noindent
Separated presheaves will play an important role in the algebraic syntax and semantics of
XTT (Section~\ref{sec:functorial-semantics}), in which we wish to ensure that
there is at most one path $\Mor{\DIM}{A}$ compatible with a boundary
$\Mor{\symbf{1}+\symbf{1}}{A}$ in the following sense:
\begin{equation}
  \begin{tikzpicture}[diagram,baseline = (I.base)]
    \node (2) {$\symbf{1}+\symbf{1}$};
    \node (I) [below = of 2] {$\DIM$};
    \node (A) [right = of 2] {$A$};
    \path[->] (2) edge node [above] {$a$} (A);
    \path[->] (2) edge node [left] {$\brk{\Dim0\mid\Dim1}$} (I);
    \path[->,exists] (I) edge node [sloped,below] {\scriptsize\emph{unique if exists}} (A);
  \end{tikzpicture}
\end{equation}

Not all presheaves are local; one may correct this defect by
\emph{quotienting} or \emph{localizing} the logos, forcing certain maps
out of colimits to become isomorphisms. Although generally localizing a category may result
in a category which is almost entirely unrecognizable, in this case the resultant localization
has a straightforward characterization: it is precisely the full subcategory spanned
by local presheaves. More formally, the inclusion of this subcategory into $\Psh{\CCat}$ has a left
adjoint which preserves finite limits and this left adjoint sends certain maps out of
colimits to isomorphisms. This adjunction presents the subcategory of local presheaves
as a \emph{left exact localization} of $\Psh{\CCat}$.
We may therefore define a logos to be a left exact localization of the category
of presheaves on a small category $\CCat$.

\begin{defi}[Logos]%
  \label{def:logos}
  A logos is a left exact localization of the category of presheaves on a small
  category $\CCat$; a morphism $\Mor{\ECat}{\FCat}$ between logoi is a functor
  between the underlying categories that is both left exact and cocontinuous.
  Such a morphism is often called an \emph{algebraic morphism}. We will write
  $\LOGOS$ for the 2-category of logoi, with 2-cells given by natural
  transformations between the underlying functors.
\end{defi}

\begin{rem}[Direct image]
  Because logoi are left exact localizations of presheaves, they are locally
  presentable~\cite{adamek-rosicky:1994}. While we do not make extensive use of this fact, we do
  require the following consequences of local presentability: all logoi are complete and cocomplete
  and every algebraic morphism $\Mor[f^*]{\ECat}{\FCat}$ has a right adjoint
  $\Mor[f_*]{\FCat}{\ECat}$.
\end{rem}

Two different base categories $\CCat,\DCat$ may yet generate the same presheaf
logos; however, presheaf logoi enjoy a special relationship with their base
categories embodied in Diaconescu's Theorem below~\cite{borceux:1994:vol2}.

\begin{thm}[Diaconescu's Theorem]\label{thm:diaconescu}
  A presheaf logos $\Psh{\CCat}$ classifies \emph{flat functors} out of $\CCat$
  in the sense that the comparison map $\Hom[\LOGOS]{\Psh{\CCat}}{\ECat}\to
  \Hom{\CCat}{\ECat}_{\mathit{flat}}$ determined by precomposing with the
  Yoneda embedding is an equivalence of categories. In particular, each
  algebraic morphism $\Mor[f^*]{\Psh{\CCat}}{\ECat}$ corresponds to an
  essentially unique flat functor $\Mor[f^*\circ\Yo[\CCat]]{\CCat}{\ECat}$.
\end{thm}

A flat functor is a generalization of the notion of left exact functor which may be used in case the
domain category $\CCat$ is not finitely complete: in other words, a flat functor
$\Mor{\CCat}{\ECat}$ is one that preserves ``even the finite limits that don't exist''. While the
general definition is slightly complex (see~\cite[Definition 8.2.8]{borceux:1994:vol1}), for our
purposes it suffices to consider the special case where $\CCat$ is finitely complete and when a flat
functor is precisely one which preserves finite limits.

Flat functors play an essential role in the general gluing theorem for models of Martin-L\"of type
theory into logoi developed by Sterling and Angiuli~\cite{sterling-angiuli:2020} which we have
applied in this paper.

\subsubsection{Topos--logos duality} 

Following the philosophy of Anel and Joyal~\cite{anel-joyal:2021}, we have
(perhaps surprisingly to some readers) not referred to categories of presheaves
as ``topoi''. This is because we prefer to think of a topos as a
\emph{geometrical} object, whereas a logos is \emph{algebraic} in nature: for
instance, Diaconescu's Theorem (Theorem~\ref{thm:diaconescu}) shows that the
category of presheaves is an invariant form of the \emph{theory} of flat
functors in the style of Lawvere's functorial semantics~\cite{lawvere:thesis}.

It is instructive to start from the prototype of geometry--algebra duality 
embodied in the relationship between (sober) topological spaces and their
frames of open sets; in this case, the frame of opens is an algebraic object,
and the corresponding space is its geometric dual. By an analogy that may
be substantiated in a precise way, a logos is the algebraic object
corresponding to a topos, which is in contrast a kind of generalized space.

The duality between topoi and logoi is captured in a formal equivalence of
categories $\mathrm{Sh} : \Mor[\simeq]{\OpCat{\TOPOS}}{\LOGOS}$, taking a topos to its
``logos of sheaves''. While historically, many authors have thought of sheaves
as local presheaves for a specific generators-and-relations presentation of a
logos, Grothendieck insisted that this presentation is not at all the main
object of study~\cite{sga:4,grothendieck:rs}: sheaves should be thought of as
algebraic data varying over a (generalized) space, including both topological
spaces \emph{and} topoi. Indeed, it is enough to consider sheaves on topoi,
recalling that any sober space corresponds to an essentially unique
\emph{enveloping topos}~\cite{anel-joyal:2021}.

On the other hand, the algebraic perspective of \emph{logoi} and their
algebraic morphisms is the most germane to this paper, so we do not refer to
topoi in subsequent sections.

\subsection{Functorial semantics of dependent type theory \`a la Uemura}%
\label{sec:taichi}

Classically, the notion of an ``algebraic theory'' was understood in terms of
sets of operations and equations. For instance, the theory of monoids may be
written in terms of two operations: $\varepsilon$ of arity $[0]$ and $\odot$ of
arity $[2]$. Then, a set of equations is imposed on the set of trees generated
by $\brc{\varepsilon,\odot}$ to express the associativity and unit laws of a
monoid.
A model of a theory in this old-fashioned sense was then a \emph{structure} comprising the following data:
\begin{enumerate}
  \item a carrier set $A$,
  \item a map $\varepsilon_A : 1 \to A$,
  \item a map $\odot_A : A \times A \to A$,
  \item subject to the following equations:
    \begin{align*}
      \odot_A\prn{\varepsilon_A,x} &= x\\
      \odot_A\prn{x,\varepsilon_A} &= x\\
      \odot_A\prn{\odot_A\prn{x,y},z} &= \odot_A\prn{x,\odot_A\prn{y,z}}
    \end{align*}
\end{enumerate}

\noindent
Models in this sense arrange themselves into a category, with morphisms given
by functions between carrier sets that commute with the specified maps; of
course, this is nothing more than the category of monoids in $\SET$. Generally,
one considers models in categories other than $\SET$, a notion that makes
sense for any category $\CCat$ having the requisite structure (in this case,
finite products).

However, there are many other collections of operations and axioms that
equally well express the concept of \emph{monoid}, evidently exhibiting the
same categories of models. For instance, one may use the set of operations
$\brc{\mathrm{list}_n \mid n \in \mathbb{N}} \cup \brc{\odot}$ where each
operation $\mathrm{list}_n$ has arity $[n]$. Lawvere famously observed that
none of the important computations in universal algebra actually depend on
which specific operations and axioms are used to encode a theory, advocating a
perspective that regards the presentations above \emph{not} as the theories
themselves, but as structures lying over the theories.

A \emph{theory} for Lawvere is a category $\ThCat$ closed under certain structures
(\eg finite products, finite limits, \etc); a model of a theory in a category
$\CCat$ is a functor $\Mor{\ThCat}{\CCat}$ preserving (finite products,
finite limits, \etc). A collection of operations and axioms that generates a
theory is called a \emph{equational presentation} of that theory.

Lawvere's functorial perspective on theories and their models is called the
\emph{functorial semantics}~\cite{lawvere:thesis}.
Lawvere has observed that many of the fundamental operations by which new
theories are constructed from old theories are unnatural to describe in terms of
presentations, but are simple at the level of categories and functors.
One may continue to \emph{present} theories $\ThCat$ by sets of operations and
axioms as before, but the spirit of the functorial method is to freely adopt
whichever presentation is most useful in a specific context.

\subsubsection{Natural models, the judgmental essence of strict type theory}

We recall from Definition~\ref{def:representable-nat-trans} the notion of a
\emph{representable natural transformation} of presheaves: it is a family of
presheaves that, at every fiber over a representable object, is a
representable object. This notion, which first arose in the context of algebraic
geometry in the Grothendieck school~\cite{sga:4}, plays a fundamental role in
the semantics of dependent type theory as well as the categorical study of set
theory and
universes~\cite{awodey:2018:natural-models,awodey:2008,streicher:2005,streicher:2014}.

From the type theoretic perspective, the importance of representable maps is
easy to explain. In type theory, the basic objects under consideration are
contexts $\Gamma$, types in context $\Gamma\vdash A$, and typed terms in
context $\Gamma\vdash a : A$. The collection of types carries an action for
each substitution $\Mor[\gamma]{\Delta}{\Gamma}$ of contexts, and so does the
collection of elements: we have $\Delta\vdash \gamma^*A$ and $\Delta\vdash
\gamma^*a : \gamma^*A$. Moreover, while there is no context that represents the
collection of all types, for any specific type $\Gamma\vdash A$, we have a context
$\Gamma.A$ that represents the \emph{elements} of the type $A$.

This type theoretic situation can be captured mathematically in three steps:
\begin{enumerate}

  \item First of all, the collection of contexts may be organized into a
    category $\CCat$ with morphisms given by simultaneous substitutions.

  \item Next, the collection of types may be viewed as a presheaf
    $\Mod{\CCat}{\TY} : \Psh{\CCat}$: a section $A\in
    \Mod{\CCat}{\TY}\prn{\Gamma}$ is exactly a type $\Gamma\vdash A$, and the
    functorial action implements substitution on types. Likewise, the
    collection of typed elements is a presheaf indexed in the presheaf of
    types, \ie a family
    $\Mor[\Mod{\CCat}{\El}]{\Mod{\CCat}{\EL}}{\Mod{\CCat}{\TY}}$.

  \item From the perspective of the Yoneda lemma, one may think of contexts as
    representable presheaves $\Gamma : \Psh{\CCat}$, and a type in context
    $\Gamma$ is a morphism $\Mor[A]{\Gamma}{\Mod{\CCat}{\TY}}$. We therefore
    obtain representing contexts $\Gamma.A$ for each such type $A$ by requiring
    that the family
    $\Mor[\Mod{\CCat}{\El}]{\Mod{\CCat}{\EL}}{\Mod{\CCat}{\TY}}$ be a
    representable natural transformation:
    \begin{equation}\label{diag:cx-ext}
      \DiagramSquare{
        nw/style = pullback,
        nw = \Gamma.A,
        sw = \Gamma,
        se = \Mod{\CCat}{\TY},
        ne = \Mod{\CCat}{\EL},
        east = \Mod{\CCat}{\El},
        south = A,
        west = p_A,
        north = q_A
      }
    \end{equation}
\end{enumerate}

\noindent
The map $\Mor[p_A]{\Gamma.A}{\Gamma}$ is the \emph{weakening} substitution, and
the map $\Mor[q_A]{\Gamma.A}{\Mod{\CCat}{\EL}}$ is the \emph{variable} term. By
chasing Diagram~\ref{diag:cx-ext}, it is easy to see that the type of the term
$q_A$ is $p_A^*A$ as expected.
The structure defined above captures the decisive judgmental aspects of
dependent type theory, and has been referred to by Awodey as a \emph{natural
model}: natural in both the informal sense, and in the sense that it is defined
in terms of a representable \emph{natural} transformation.

\begin{defiC}[{\cite{awodey:2018:natural-models}}]\label{def:natural-model}
  A natural model is a category $\CCat$ with a terminal object, together with a representable natural transformation in $\Psh{\CCat}$.
\end{defiC}

\subsubsection{Representable map categories and the semantics of type theory}

A given type theory is more than just a natural model in the sense of
Definition~\ref{def:natural-model}: one must also specify other generators,
such as type connectives and their elements. Most type connectives can be
specified by writing down a family in $\Psh{\CCat}$ and then asking for a
cartesian square between that family and the natural model. This raises some questions:

\begin{enumerate}

  \item What kinds of structures can be added to the notion of a natural model
    and still give rise to a type theory?

  \item What is a morphism between models of such a type theory?

\end{enumerate}

\noindent
Uemura proposes to answer this question by developing a notion of ``general
type theory'' and associated functorial semantics~\cite{uemura:2019,uemura:2021:thesis}; while
Lawvere defines an algebraic theory to be (roughly) a category with finite
products, Uemura defines a \emph{type} theory $\ThCat$ to be a
\emph{representable map category}, which is a lex category $\CCat$ together
with a distinguished class $\mathcal{R}$ of ``representable maps'' which
captures the decisive aspects of the class of representable natural
transformations in presheaf logoi.

\begin{defi}
  A class of representable maps in a lex category $\CCat$ is a collection of maps
  $\mathcal{R}$ with the following closure conditions:

  \begin{enumerate}
    \item Each identity map is representable, and the composition of representable maps is representable.
    \item Every pullback of a representable map along an arbitrary map is representable.
    \item Pullback along a representable map has a right adjoint, pushforward.
      \qedhere
  \end{enumerate}
\end{defi}

\begin{nota}\label{notation:exponential}
  In particular, if the terminal map $\Mor{X}{\ObjTerm{\CCat}}$ is
  representable, then all exponentials (internal homs) out of $X$ may be
  computed by pushforward, which we may write either $X\Rightarrow Y$ or $\IHom{X}{Y}$.
\end{nota}

\begin{defi}
  A representable map category is a lex category $\CCat$ together with a class
  $\mathcal{R}$ of representable maps in $\CCat$; a \emph{representable map
  functor} between two representable map categories is a functor between the
  underlying categories that takes representable maps to representable maps.
\end{defi}

\begin{rem}
  Every lex category $\CCat$ can be thought of as a very weak kind of extensional type theory, in which dependent types are maps and substitutions are given by pullback. A representable map structure on $\CCat$ does nothing more than enrich this language as follows:

  \begin{enumerate}
    \item There is a class of \emph{small} types called ``representable'' types.
    \item Representable types are closed under dependent sums.
    \item Types are closed under dependent products with representable base.
  \end{enumerate}
  \noindent
  The ``type theory'' of a representable map category $\ThCat =
  \prn{\CCat,\mathcal{R}}$ can be thought of as a logical framework in the
  sense of~\cite{harper-honsell-plotkin:1993,nordstrom-peterson-smith:1990}, in
  which hypothetical judgments are represented by pushforward. \emph{From this
  perspective, it is most appropriate to refer to arbitrary objects and
  families in $\ThCat$ as judgments, following the ``judgments as types''
  philosophy of Harper, Honsell, and
  Plotkin~\cite{harper-honsell-plotkin:1993}.}
\end{rem}

\begin{exa}[The walking natural model]
  The type theory with judgments for types and terms is specified by the
  representable map category $\ThCat$ generated by a single representable map
  $\ElMor$.
\end{exa}

\begin{exa}
  $\Psh{\CCat}$ supports the structure of a representable map category with the class of
  representable maps given by representable natural transformations.
\end{exa}

The classic notion of a representable map in a category of presheaves is of
course an instance, giving rise to the ``canonical representable structures''
on $\Psh{\CCat}$ and $\DF{\CCat}$.

\begin{exa}
  $\DF{\CCat}$ supports the structure of a representable map category, with
  representable maps given by a functors between discrete fibrations that have
  a non-fibered right
  adjoint~\cite{awodey:2018:natural-models,awodey-butz-simpson-streicher:2014}.
  Explicitly, a morphism $\Mor[f]{X}{Y} : \DF{\CCat}$ is a representable
  morphism if there is a functor $\Mor[\ReprAdj{f}]{Y}{X}:\CAT$ such that $f
  \dashv \ReprAdj{f}$.

  Moreover, with this class of representable maps, the equivalence between
  $\DF{\CCat}$ and $\Psh{\CCat}$ induces an equivalence between representable
  map categories. In other words, viewing a representable natural
  transformation of presheaves as a functor between categories of elements, one
  has such a right adjoint, and vice versa.
\end{exa}

The functorial semantics of type theories \`a la Uemura is then given in terms
of these canonical representable map categories:

\begin{defi}[Models]\label{def:taichi-model}
  A model of a type theory $\ThCat = \prn{\CCat,\mathcal{R}}$ is a representable map functor
  $\Mor[\Mod{\CCat}]{\ThCat}{\DF{\CCat}}$ where $\CCat$ is a category with a terminal object.
\end{defi}

In Definition~\ref{def:taichi-model}, the base category $\CCat$ is the category
of contexts, and each $\Mod{\CCat}{J}$ is the interpretation of a
``judgment'' $J:\ThCat$ as a presheaf over $\CCat$.

\begin{nota}
  Given any object $J : \ThCat$, we write $\Mod{\CCat}{J}$ for
  $\Mod{\CCat}{J}$; likewise, for each $\Mor[f]{J}{I} : \ThCat$, we write
  $\Mor[\Mod{\CCat}{f}]{\Mod{\CCat}{J}}{\Mod{\CCat}{I}} : \DF{\CCat}$ for
  $\Mod{\CCat}{f}$.
\end{nota}

To a first approximation, a morphism $\Mod{F}$ between two $\ThCat$-models $\Mod{\CCat},\Mod{\DCat}$
should be a functor
$\Mor[F]{\CCat}{\DCat}$ together with a natural assignment of functors $X \mapsto \Mod{F}{X}$ for
each $X : \ThCat$:

\begin{equation}
  \begin{tikzpicture}[diagram,baseline = (sw.base)]
    \SpliceDiagramSquare{
      height = 1.5cm,
      nw = \Mod{\CCat}{X},
      sw = \CCat,
      ne = \Mod{\DCat}{X},
      se = \DCat,
      west/style = lies over,
      east/style = lies over,
      north = \Mod{F}{X},
      south = F,
    }
    \node (tot) [right = of ne] {$\DF{}$};
    \node (base) [right = of se] {$\CAT$};
    \path[fibration] (tot) edge (base);
  \end{tikzpicture}
  \tag{$\ast$}%
  \label{eq:morphism-of-models}
\end{equation}

In addition to these requirements, a morphism of models should also preserve
\emph{context extension} up to isomorphism; from a categorical perspective,
there is not much meaning in asking for context extension to be preserved ``on
the nose'', since contexts are \emph{objects} of a category and therefore considered
only up to isomorphism. Therefore, given a context $C:\CCat$ and a type
$\Mor[A]{\DFYo[\CCat]{C}}{\Mod{\CCat}{\TY}}$, we would expect that the context
extension $F{C}.\Mod{F}{\TY}{A}$ shall be isomorphic to $F\prn{C.A}$.

In Uemura's framework, as with natural models~\cite{awodey:2018}, context
extension is modeled by the representability of $\Mor[\El]{\EL}{\TY}$. In fact,
calculation shows that the (non-fibered) right adjoint
$\ReprAdj{\Mod{\CCat}{\El}}$ to $\Mod{\CCat}{\El}$ sends a type
$\Mor[A]{\DFYo[\CCat]{C}}{\Mod{\CCat}{\TY}}$ to the variable term
$\Mor{\DFYo[\CCat]\prn{C.A}}{\Mod{\CCat}{\EL}}$ in the extended context. We may
therefore phrase the preservation of context extensions, called
\emph{Beck-Chevalley} by Uemura, in terms of $\ReprAdj{\Mod{\CCat}{f}}$ and
$\ReprAdj{\Mod{\DCat}{f}}$ for \emph{each} representable map $f$, including
$\ElMor$.

First, observe that for each representable map $\Mor[f]{J}{I}$, we have a
canonical 2-cell in $\CAT$ with the following boundary:
\begin{equation}\label{diag:taichi-bc}
  \begin{tikzpicture}[diagram,baseline = (sw.base)]
    \SpliceDiagramSquare{
      nw = \Mod{\CCat}{I},
      sw = \Mod{\CCat}{J},
      west = \ReprAdj{\Mod{\CCat}{f}},
      east = \ReprAdj{\Mod{\DCat}{f}},
      ne = \Mod{\DCat}{I},
      se = \Mod{\DCat}{J},
      south = \Mod{F}{J},
      north = \Mod{F}{I},
    }
    \node [between = nw and se] {$\Rightarrow$};
  \end{tikzpicture}
  \quad
\end{equation}

The 2-cell of Diagram~\ref{diag:taichi-bc} may be computed as follows:
\begin{prooftree*}
  \infer0{
    \Mor{\Mod{\CCat}{f}\circ{\ReprAdj{\Mod{\CCat}{f}}}}{\ArrId{\Mod{\CCat}{I}}}
  }
  \infer1{
    \Mor{{\Mod{F}{I}}\circ\Mod{\CCat}{f}\circ{\ReprAdj{\Mod{\CCat}{f}}}}{\Mod{F}{I}}
  }
  \infer1{
    \strut
    \Mor{\Mod{\DCat}{f}\circ{\Mod{F}{J}\circ\ReprAdj{\Mod{\CCat}{f}}}}{\Mod{F}{I}}
  }
  \infer1{
    \strut
    \Mor{\Mod{F}{J}\circ\ReprAdj{\Mod{\CCat}{f}}}{\ReprAdj{\Mod{\DCat}{f}}\circ \Mod{F}{I}}
  }
\end{prooftree*}

The Beck-Chevalley condition for $\Mor[f]{J}{I}$ is, then, that the 2-cell in
Diagram~\ref{diag:taichi-bc} is invertible.

\begin{defi}[Morphism of Models]
  A morphism between $\Mor[\Mod{\CCat}]{\ThCat}{\DF{\CCat}}$ and
  $\Mor[\Mod{\DCat}]{\ThCat}{\DF{\DCat}}$ is a functor $\Mor[F]{\CCat}{\DCat}$
  preserving the terminal object, together with an assignment of functors
  $\Mor[\Mod{F}{X}]{\Mod{\CCat}{X}}{\Mod{\DCat}{X}}$ lying over $F$, natural in
  $X : \ThCat$, such that each representable map $\Mor[f]{X}{Y}$ satisfies the
  Beck-Chevalley condition.
\end{defi}

\begin{nota}
  Let $\Mor[\Mod{F}]{\Mod{\CCat}}{\Mod{\DCat}}$ be a morphism of
  $\ThCat$-models; given $X:\ThCat,x : \Mod{\CCat}{X}$, we will abusively write
  $F{x}$ for $\Mod{F}{X}{x}$.
\end{nota}

Because morphisms of models necessarily require preservation of context
extensions only up to isomorphism, a higher level of morphism naturally arises,
lending the collection of models $\ModCat{\ThCat}$ with the structure of a
2-category.

\begin{defi}[2-morphisms of models]
  Given a pair of morphisms $\Mor[\Mod{F},\Mod{G}]{\Mod{\CCat}}{\Mod{\DCat}}$, a 2-morphism
  $\Mor[\Mod{\alpha}]{\Mod{F}}{\Mod{G}}$ is a natural transformation $\Mor[\alpha]{F}{G}$ between the underlying functors such
  that for each $X : \ThCat$, there exists a \emph{necessarily unique} natural transformation
  $\Mor[\Mod{\alpha}{X}]{\Mod{F}{X}}{\Mod{G}{X}}$ lying over $\alpha$.
\end{defi}

\begin{rem}
  The uniqueness of the map $\Mor[\Mod{\alpha}{X}]{\Mod{F}{X}}{\Mod{G}{X}}$
  lying over $\alpha$ is ensured by the discreteness of $\Mod{\DCat}{X}$.
  Summarizing, the existence of $\Mod{\alpha}{X}$ is nothing more than the
  \emph{condition} that for each $x : \Mod{\CCat}{X}$ lying over $\Gamma :
  \CCat$, the following equation obtains:
  \[
    \alpha_{\Gamma}^*{Gx} = F{x}
    \qedhere
  \]
\end{rem}

\begin{thm}%
  \label{thm:bi-initial-model}
  The 2-category of models $\ModCat{\ThCat}$ has a \emph{bi-initial} object: an object $I$ such that
  for each $X : \ModCat{\ThCat}$, the category of morphisms $\Hom{I}{X}$ is contractible.
\end{thm}

This bi-initial object is the \emph{democratic heart} of the embedding
$\Mor[\DFYo]{\ThCat}{\DF{\ThCat}}$: the smallest full subcategory of $\ThCat$
containing the terminal object and closed under context extension.
The universal property of this bi-initial object ensures that there is at most one
morphism $\Mor{\Mod{\ICat}}{\Mod{\CCat}}$ for
each $\Mod{\CCat}$, up to a unique invertible 2-morphism.

\subsection{Left lifting structures, orthogonality, and separation}

In this section, let $\ECat$ be a category with finite limits. We call
an object $X : \ECat$ \emph{exponentiable} when every exponential
$\IHom{X}{Y}$ exists. We begin by recalling some basic definitions and facts
from~\cite{awodey:2018:natural-models,anel-biedermann-finster-joyal:2018,anel-biedermann-finster-joyal:2020}.

\begin{defi}[Cartesian gap map]
  Let the following square commute in $\ECat$:
  \begin{equation}\label{diag:gap-map:0}
    \DiagramSquare{
      nw = A,
      sw = B,
      ne = X,
      se = Y,
      west = f,
      east = g,
      north = k,
      south = h,
    }
  \end{equation}

  The universal property of the pullback of $g$ along $h$ states that the span
  $f,k$ induces a map $\gl{f,k}$, which we call the \emph{cartesian gap map} of
  Diagram~\ref{diag:gap-map:0}:
  \begin{equation*}
    \begin{tikzpicture}[diagram,baseline=(sw.base)]
      \SpliceDiagramSquare{
        nw/style = pullback,
        nw = X\times_Y B,
        sw = B,
        ne = X,
        se = Y,
        east = g,
        south = h,
      }
      \node (nww) [above left = 2.5cm of nw] {$A$};
      \path[->,bend right=30] (nww) edge node [left] {$f$}  (sw);
      \path[->,bend left=30] (nww) edge node [sloped,above] {$k$} (ne);
      \path[->,exists] (nww) edge node [desc] {$\gl{f,k}$} (nw);
    \end{tikzpicture}
    \qedhere
  \end{equation*}
\end{defi}

\begin{defi}[Internal pullback hom]
  Let $\Mor[f]{A}{B}$ and $\Mor[g]{X}{Y}$ be morphisms in $\ECat$ such that
  $A,B$ are both exponentiable. Then, the \emph{internal pullback hom} of $f$
  and $g$ is the cartesian gap map $\Mor[\ggl{f,g}]{\IHom{B}{X}}{\IHom{A}{X}\times_{\IHom{A}{Y}}\IHom{B}{Y}}$ of the following
  square:
  \begin{equation*}
    \DiagramSquare{
      width = 3cm,
      nw = \IHom{B}{X},
      sw = \IHom{B}{Y},
      se = \IHom{A}{Y},
      ne = \IHom{A}{X},
      west = \IHom{B}{g},
      east = \IHom{A}{g},
      north = \IHom{f}{X},
      south = \IHom{f}{Y}
    }
    \qedhere
  \end{equation*}
\end{defi}

\begin{defi}[External orthogonality]\label{def:external-orthogonality}
  A map $\Mor[f]{A}{B}$ is \emph{left orthogonal} to $\Mor[g]{X}{Y}$, written $\ExtOrth{f}{g}$, when there exists
  a unique lift for each square of the following kind:
  \begin{equation*}
    \begin{tikzpicture}[diagram,baseline=(sw.base)]
      \SpliceDiagramSquare{
        nw = A,
        sw = B,
        ne = X,
        se = Y,
        west = f,
        east = g,
        south = y,
        north = x,
      }
      \path[->,exists] (sw) edge node [desc] {$\exists!$} (ne);
    \end{tikzpicture}
    \qedhere
  \end{equation*}
\end{defi}

\begin{defi}[Internal orthogonality]\label{def:orthogonality-new}
  A map $\Mor[f]{A}{B}$ is \emph{internally left orthogonal} to
  $\Mor[g]{X}{Y}$, written $\IntOrth{f}{g}$, if we have $\ExtOrth{\prn{Z\times
  f}}{g}$ for every $Z:\ECat$.
\end{defi}

\begin{lemC}[{\cite[Definition~3.2.5]{anel-biedermann-finster-joyal:2020}}]
  Fix $\Mor[f]{A}{B}$ and $\Mor[g]{X}{Y}$ with $A,B$ exponentiable; then
  $\IntOrth{f}{g}$ iff the internal pullback hom $\ggl{f,g}$ is an isomorphism.
  \qed
\end{lemC}

\begin{defiC}[{\cite[Definition~18]{awodey:2018:natural-models}}]\label{def:left-lifting-structure}
  Fix $\Mor[f]{A}{B}$ and $\Mor[g]{X}{Y}$ with $A,B$ exponentiable; a \emph{left lifting structure}
  for $f$ against $g$ (written $f\pitchfork g$) is a section $j$ of the internal pullback hom $\ggl{f,g}$:
  \begin{equation*}
    \begin{tikzpicture}[diagram,baseline = (pb.base)]
      \node (bx) {$\IHom{B}{X}$};
      \node (pb) [below = of bx] {$\IHom{A}{X}\times_{\IHom{A}{Y}}\IHom{B}{Y}$};
      \path[->] (bx) edge node [upright desc] {$\ggl{f,g}$} (pb);
      \path[->,bend left = 70] (pb) edge node [left] {$j$} (bx);
    \end{tikzpicture}
    \qedhere
  \end{equation*}
\end{defiC}

\begin{exa}[Intensional identity types]
  In order to formulate the elimination of intensional identity types in a natural model,
  Awodey~\cite{awodey:2018} uses a left lifting structure. More generally,
  most ``pattern-matching'' style elimination rules in type theory can be formulated as a left
  lifting structure~\cite{gratzer-kavvos-nuyts-birkedal:2020}.
\end{exa}

\begin{lemC}[{\cite[Lemma~19]{awodey:2018:natural-models}}]%
  \label{lem:natural-lifting-structure}
  \pushQED{\qed}
  Given maps $\Mor[f]{A}{B}$ and $\Mor[g]{X}{Y}$, a left lifting structure $j : f \pitchfork g$ is
  equivalent to a choice of lifts $j_Z(y,x)$ natural in $Z$ for any diagram of the following kind:
  \begin{equation*}
    \begin{tikzpicture}[diagram,baseline = (sw.base)]
      \SpliceDiagramSquare{
        nw = Z \times A,
        sw = Z \times B,
        ne = X,
        se = Y,
        west = Z \times f,
        east = g,
        north = x,
        south = y,
      }
      \path[->,exists] (sw) edge node [desc] {$\MathSmall{j_Z(y,x)}$} (ne);
    \end{tikzpicture}
    \qedhere
  \end{equation*}
  \popQED
\end{lemC}

\begin{rem}
  A left lifting structure $j:f \pitchfork g$ exhibits $f$ as internally left orthogonal to $g$ when $j$ is simultaneously a retraction.
\end{rem}

Because both left lifting structures and orthogonality conditions may be
expressed in the language of finite limits as above, it is justified to freely
extend a representable map category $\ThCat$ by either $\IntOrth{f}{g}$ or $j :
f\pitchfork g$.
In many cases, however, an orthogonality condition or lifting structure will
need to be expressed in the free cocompletion $\Psh{\ThCat}$ because it may
involve colimits that don't exist in $\ThCat$; when defining a representable
map category by a sequence of clauses, it is not \emph{a priori} clear that
this move is legitimate.

We will therefore characterize a useful class of orthogonality and lifting
conditions on $\Psh{\CCat}$ which may be unravelled into a suitable condition
on a lex category $\CCat$, expressible using the language of finite limits.

\begin{lem}\label{lem:colimit-lift}
  Let $\ICat$ be a small category such that limits of $\ICat$-diagrams exist in $\CCat$;
  let $\Mor[\Phi_\bullet]{\ICat}{\CCat}$ be a diagram such that for
  each $i : \ICat$, the object $\Phi_i$ is exponentiable. We will write
  $\hat\Phi_\infty$ for the colimit of $\Phi_\bullet$ taken in $\Psh{\CCat}$,
  \ie $\hat\Phi_\infty = \Colim[\ICat]{\Yo[\CCat]{\Phi_\bullet}}$. Now let
  $\Mor[f]{\hat\Phi_\infty}{\Yo[\CCat]{B}}:\Psh{\CCat}$ with $B$ exponentiable, and
  let $\Mor[g]{X}{Y}$ be an arbitrary map in $\CCat$.
  Then, there exists a left lifting structure $j : f\pitchfork\Yo[\CCat]{g}$ in $\ThCat$ iff there
  exists a section of the cartesian gap map for the canonical
  Diagram~\ref{diag:colimit-lift:0} in $\CCat$ below:
  \begin{equation}\label{diag:colimit-lift:0}
    \DiagramSquare{
      nw = \IHom{B}{X},
      sw = \IHom{B}{Y},
      ne = \Lim[\ICat]*{\IHom{\Phi_\bullet}{X}},
      se = \Lim[\ICat]*{\IHom{\Phi_\bullet}{Y}},
      width = 2.5cm,
      height = 1.5cm,
    }
  \end{equation}

  Moreover, the left lifting structure $j$ exhibits $f$ as internally left orthogonal to
  $\Yo[\CCat]{g}$ iff the cartesian gap map of Diagram~\ref{diag:colimit-lift:0} is an
  isomorphism.
\end{lem}

Before proving Lemma~\ref{lem:colimit-lift}, we first clarify the construction
of Diagram~\ref{diag:colimit-lift:0}.  For each $i : \ICat$, we have the
following composite map:
\begin{equation}\label{diag:colimit-lift:1}
  \begin{tikzpicture}[diagram,baseline = (2.base)]
    \node (0) {$\Yo[\CCat]{\Phi_i}$};
    \node (1) [right = of 0] {$\hat\Phi_\infty$};
    \node (2) [below = of 1] {$\Yo[\CCat]{B}$};
    \path[->] (0) edge node [above] {$\Kwd{in}_i$} (1);
    \path[->] (1) edge node [right] {$f$} (2);
    \path[->,exists] (0) edge node [sloped,below] {$\Yo[\CCat]{f_i}$} (2);
  \end{tikzpicture}
\end{equation}

Therefore, letting $Z$ range over $X,Y$, we may construct a map into
$\Lim[\ICat]*{\IHom{\Phi_\bullet}{Z}}$ from a cone defined as follows:
\begin{equation}\label{diag:colimit-lift:2}
  \begin{tikzpicture}[diagram,baseline = (2.base)]
    \node (0) {$\IHom{B}{Z}$};
    \node (1) [right = 3.5cm of 0] {$\Lim[\ICat]*{\IHom{\Phi_\bullet}{Z}}$};
    \node (2) [below = of 1] {$\IHom{\Phi_i}{Z}$};
    \path[->,exists] (0) edge node [above] {$\brk{i.\IHom{f_i}{Z}}$} (1);
    \path[->] (1) edge node [right] {$\pi_i$} (2);
    \path[->] (0) edge node [sloped,below] {$\IHom{f_i}{Z}$} (2);
  \end{tikzpicture}
\end{equation}

The force of Lemma~\ref{lem:colimit-lift} is therefore to assert that the
existence of a left lifting structure $j : f\pitchfork\Yo[\CCat]{g}$ in
$\Psh{\CCat}$ is equivalent to the existence of a section to the cartesian gap
map of Diagram~\ref{diag:colimit-lift:3} below:
\begin{equation}\label{diag:colimit-lift:3}
  \DiagramSquare{
    west = \IHom{B}{g},
    east = \Lim[\ICat]*{\IHom{\Phi_\bullet}{g}},
    nw = \IHom{B}{X},
    sw = \IHom{B}{Y},
    ne = \Lim[\ICat]*{\IHom{\Phi_\bullet}{X}},
    se = \Lim[\ICat]*{\IHom{\Phi_\bullet}{Y}},
    north = \brk{i.\IHom{f_i}{X}},
    south = \brk{i.\IHom{f_i}{Y}},
    width = 3.5cm,
  }
\end{equation}

\begin{proof}[Proof of Lemma~\ref{lem:colimit-lift}]
  The condition of the internal pullback hom $\ggl{f,\Yo[\CCat]{g}}$ may be portrayed
  as follows:
  \begin{equation}
    \begin{tikzpicture}[diagram,baseline = (sw.base)]\label{diag:colimit-lift:4}
      \SpliceDiagramSquare{
        width = 4cm,
        height = 4cm,
        ne = \IHom{\hat\Phi_\infty}{\Yo[\CCat]{X}},
        se = \IHom{\hat\Phi_\infty}{\Yo[\CCat]{Y}},
        sw = \Yo[\CCat]{\IHom{B}{Y}},
        nw = \Yo[\CCat]{\IHom{B}{X}},
        west/style = {bend right = 10},
        north/style = {bend left = 10},
      }
      \node (pb) [pullback,below right = 2.75cm of nw] {$P$};
      \path[->] (pb) edge (ne);
      \path[->] (pb) edge (sw);
      \path[->,bend left = 30] (nw) edge node [desc] {$\MathSmall{\ggl{f,\Yo[\CCat]{g}}}$} (pb);
      \path[->,exists,bend left = 30] (pb) edge node [desc] {$s$} (nw);
    \end{tikzpicture}
  \end{equation}

  But for each $Z:\CCat$, the presheaf $\IHom{\hat\Phi_\infty}{\Yo[\CCat]{Z}}$ is
  canonically represented by the object $\Lim[\ICat]*{\IHom{\Phi_\bullet}{Z}} :
  \CCat$, using the universality of colimits in $\Psh{\CCat}$ and
  the fact that the Yoneda embedding commutes with limits and exponentials.
  Therefore, Diagram~\ref{diag:colimit-lift:5} below faithfully translates the
  existence of the dotted map in Diagram~\ref{diag:colimit-lift:4} into the
  language of $\CCat$:
  \begin{equation}\label{diag:colimit-lift:5}
    \begin{tikzpicture}[diagram,baseline = (sw.base)]
      \SpliceDiagramSquare{
        width = 4cm,
        height = 4cm,
        ne = \Lim[\ICat]*{\IHom{\Phi_\bullet}{X}},
        se = \Lim[\ICat]*{\IHom{\Phi_\bullet}{Y}},
        sw = \IHom{B}{Y},
        nw = \IHom{B}{X},
        west/style = {bend right = 10},
        north/style = {bend left = 10},
      }
      \node (pb) [pullback,below right = 2.75cm of nw] {$P$};
      \path[->] (pb) edge (ne);
      \path[->] (pb) edge (sw);
      \path[->,bend left = 30] (nw) edge node [desc] {$\MathSmall{\ggl{f,\Yo[\CCat]{g}}}$} (pb);
      \path[->,exists,bend left = 30] (pb) edge node [desc] {$s$} (nw);
    \end{tikzpicture}
  \end{equation}

  But the above is exactly the existence of a section to the cartesian gap map
  for Diagram~\ref{diag:colimit-lift:3}. It is likewise easy to see that one
  section is a retraction iff the other is.
\end{proof}

We will have need of an orthogonality-like notion in which lifts may not exist,
but when they do, they are unique. It is most appropriate to call this
condition \emph{separation}, by analogy with coverages and Grothendieck
topologies.

\begin{defi}\label{def:separation}
  Let $\Mor[f]{A}{B}$ and $\Mor[g]{X}{Y}$ be maps in a category $\CCat$; then
  we say that $g$ is \emph{separated} with respect to $f$ when, for every
  object $Z : \CCat$, there is at most one lift for any square of the following shape:
  \begin{equation*}
    \begin{tikzpicture}[diagram,baseline=(sw.base)]
      \SpliceDiagramSquare{
        nw = Z\times A,
        sw = Z\times B,
        west = Z\times f,
        east = g,
        ne = X,
        se = Y,
      }
      \path[->,exists] (sw) edge (ne);
    \end{tikzpicture}
    \qedhere
  \end{equation*}
\end{defi}

\subsection{Using Artin gluing to prove canonicity}

In this section, we give a gentle introduction to the use of Artin gluing for
proving ordinary canonicity results in order to set the stage for our more
sophisticated application. Let $\CCat$ be the free cartesian closed category
generated by a base type $\Kwd{ans}:\CCat$ and two constants
$\Mor[\Kwd{yes},\Kwd{no}]{\ObjTerm{\CCat}}{\Kwd{ans}}$. We wish to prove the
following \emph{closed canonicity} theorem:

\begin{prop}[Closed canonicity]\label{prop:lambda-canonicity}
  Let $\Mor[a]{\ObjTerm{\CCat}}{\Kwd{ans}}$ be a morphism in $\CCat$; then
  either $a = \Kwd{yes}$ or $a = \Kwd{no}$.
\end{prop}

This can be proved by means of \emph{logical predicates} / the \emph{Tait's method of
  computability}~\cite{tait:1967}. Peter Freyd observed that the structural aspects of Tait's method
can be captured by the well-known categorical construct of Artin gluing.
We define a new category $\mathscr{P}$ that will serve as a domain of
interpretation for a logical predicates model of $\CCat$. An object of
$\mathscr{P}$ is defined to be a pair $\prn{C,\tilde{C}}$ with $C : \CCat$ and
$\tilde{C}\subseteq\Hom[\CCat]{\ObjTerm{\CCat}}{C}$, and a morphism
$\Mor{\prn{C,\tilde{C}}}{\prn{D,\tilde{D}}}$ is given by a morphism
$\Mor[f]{C}{D} : \CCat$ that preserves the predicates in the sense that the
following commutative diagram exists:
\[
  \DiagramSquare{
    nw = \tilde{C},
    ne = \tilde{D},
    sw = \Hom[\CCat]{\ObjTerm{\CCat}}{C},
    se = \Hom[\CCat]{\ObjTerm{\CCat}}{D},
    south = \Hom[\CCat]{\ObjTerm{\CCat}}{f},
    west = \subseteq,
    east = \subseteq,
    north/style = exists,
    west/style = embedding,
    east/style = embedding,
    width = 3.5cm,
  }
\]

The diagram above states the traditional condition for extending a logical predicate
to open terms: for any closed $c \in \tilde{C}$, the evaluation $f(c)$ is in
$\tilde{D}$.

We have a \emph{fibration} $\FibMor[p]{\mathscr{P}}{\CCat}$ sending
$\prn{C,\tilde{C}}$ to $C$; moreover, we can show that $\mathscr{P}$ is
cartesian closed and the fibration $p$ preserves the cartesian closed
structure. The latter amounts to constructing products and exponentials in $\mathscr{P}$
whose $\CCat$-parts are the corresponding products and exponentials from the
syntactic category, which is always built into logical predicates arguments.

In functorial semantics, a logical predicates model of $\CCat$ is the same as
a functor $\Mor{\CCat}{\mathscr{P}}$ that preserves the cartesian closed
structure for which the composite with the fibration
$\FibMor[p]{\mathscr{P}}{\CCat}$ is the identity, \ie a section of $p$. Because $\CCat$ is the
\emph{free} cartesian closed category generated by
$\brc{\Kwd{ans},\Kwd{yes},\Kwd{no}}$, such a functor is \emph{uniquely} determined by
a choice of an object $\bbrk{\Kwd{ans}} : \mathscr{P}$ such that
$p\bbrk{\Kwd{ans}} = \Kwd{ans}$, and a choice of constants
$\bbrk{\Kwd{yes}},\bbrk{\Kwd{no}}$ configured as follows:
\[
  \begin{tikzpicture}[diagram,baseline = (sw.base)]
    \SpliceDiagramSquare<l/>{
      west/style = lies over,
      east/style = lies over,
      sw = \ObjTerm{\CCat},
      se = \Kwd{ans},
      south = \Kwd{yes},
      nw = \ObjTerm{\mathscr{P}},
      ne = \bbrk{\Kwd{ans}},
      north = \bbrk{\Kwd{yes}},
      north/style = exists,
      height = 1.5cm,
    }
    \SpliceDiagramSquare<r/>{
      west/style = lies over,
      east/style = lies over,
      nw/style = {right = of l/ne},
      sw = \ObjTerm{\CCat},
      se = \Kwd{ans},
      south = \Kwd{no},
      nw = \ObjTerm{\mathscr{P}},
      ne = \bbrk{\Kwd{ans}},
      north = \bbrk{\Kwd{no}},
      north/style = exists,
      height = 1.5cm,
    }
    \node (tot) [right = of r/ne] {$\mathscr{P}$};
    \node (base) [right = of r/se] {$\CCat$};
    \path[fibration] (tot) edge (base);
  \end{tikzpicture}
\]

We then choose $\bbrk{\Kwd{ans}}$ to be the pair
$\prn{\Kwd{ans},\brc{\Kwd{yes},\Kwd{no}}\subseteq\Hom[\CCat]{\ObjTerm{\CCat}}{\Kwd{ans}}}$.
Then the elements $\bbrk{\Kwd{yes}},\bbrk{\Kwd{no}}$ can be defined to be
$\Kwd{yes},\Kwd{no}$ respectively. The canonicity theorem follows.

\begin{proof}[Proof of Proposition~\ref{prop:lambda-canonicity}]
  Let $\Mor[a]{\ObjTerm{\CCat}}{\Kwd{ans}}$ be a morphism of $\CCat$; by means
  of the cartesian closed section $\Mor[\bbrk{-}]{\CCat}{\mathscr{P}}$ of
  $\FibMor[p]{\mathscr{P}}{\CCat}$, there exists a morphism
  $\Mor[\bbrk{a}]{\ObjTerm{\mathscr{P}}}{\bbrk{\Kwd{ans}}}$ such that
  $p\bbrk{a} = a$. Unfolding the definition of morphisms in $\mathscr{P}$ and
  the definition of $\bbrk{\Kwd{ans}}$, this is the same as saying that
  $a\in \brc{\Kwd{yes},\Kwd{no}}$.
\end{proof}

In this paper, we prove a much more sophisticated canonicity theorem and
therefore must change the construction in a few important ways.

\begin{enumerate}

  \item We replace the proof-irrelevant logical predicate interpretation (based
    on subsets) with a proof-relevant interpretation (based on families).

  \item Rather than considering closed terms (\ie elements of the hom sets
    $\Hom[\CCat]{\ObjTerm{\CCat}}{C}$), we consider terms relative to purely
    cubical contexts $\Psi = \brk{i:\DIM,\ldots}$, \ie elements of any hom sets
    $\Hom[\CCat]{\Psi}{C}$. We require our families (proof-relevant predicates)
    to carry a contravariant functorial action for cubical substitutions $\Mor{\Psi'}{\Psi}$.

\end{enumerate}
\noindent
The modifications described above are actually parameters to a much more
general construction called Artin gluing. When $\CCat$ is a category and
$\Mor[F]{\CCat}{\ECat}$ is a functor, the Artin gluing of $\CCat$ along $F$ is
defined be the comma category $\GlCat \equiv \Comma{\brc{\ECat}}{F}$. It is instructive to
view the gluing as a restriction of the codomain fibration in the case that
$\ECat$ has finite limits:
\[
  \DiagramSquare{
    nw/style = pullback,
    east/style = fibration,
    west/style = fibration,
    nw = \GlCat,
    ne = \ArrCat{\ECat},
    se = \ECat,
    east = \Cod,
    south = F,
    west = p,
    sw = \CCat,
  }
\]

Artin gluing theorems express the relationship between various properties of
$\CCat$ and $F$ and $\ECat$ to corresponding properties on $\GlCat$. For
instance, if $F$ is a binary product preserving functor between cartesian
closed categories, then $\GlCat$ is cartesian closed and
$\FibMor[p]{\GlCat}{\CCat}$ preserves the cartesian closed structure.  Similar
results hold for topoi, which we make use of in this paper.

\section{Functorial semantics of XTT}\label{sec:functorial-semantics}

In Section~\ref{sec:xtt:lf} we defined XTT as a signature in Uemura's logical framework. Following
Uemura~\cite{uemura:2019}, for the purposes of our canonicity proof we systematically elaborate this
signature to a series of structures on an arbitrary representable map category
$\ThCat$~\cite[Theorem 5.17]{uemura:2019}. $\ThCat$ is a model of XTT precisely
when it admits all of these structures.
We write $\EMP:\ThCat$ for the terminal object of $\ThCat$.

\begin{spec}[Judgmental structure]
  The basic judgmental structure of XTT is specified below.
  \begin{enumerate}

    \item A representable map $\ElMor$ which encodes the collection
      of typed elements lying over typed terms; the representability of $\El$
      allows the (abstract) context to be extended by an element $x:A$ of a type
      $A$.

    \item A representable map $\DimMor$, implementing the interval
      and its context extension.

    \item A representable map $\CondMor$, implementing the face
      formula judgment and its context extension. $\Prf$ is the ``true'' face
      formula. As a map out of the terminal object, $\Prf$ is automatically monomorphic: this means that two proofs of the
      same face formula are identical.
      \qedhere
  \end{enumerate}
\end{spec}

\begin{spec}[The interval]
  We require only minimal structure on the interval, the two endpoints
  $\Mor[\Dim0,\Dim1]{\EMP}{\DIM}$.
\end{spec}

\begin{spec}[Dimension equality]\label{spec:dim-eq}
  We require that there is a code $\Mor[\prn{=}]{\DIM^2}{\COND}$ for the
  diagonal map $\Mor[\delta]{\DIM}{\DIM\times\DIM}$ in the following sense:
  \begin{equation*}
    \CondSquare{
      south/style = exists,
      nw = \DIM,
      sw = \DIM^2,
      west = \delta,
      south = \prn{=}
    }
    \qedhere
  \end{equation*}
\end{spec}

We must be careful when specifying disjunction of face formulas; the ``true''
disjunction of $\Prf[\phi]$ and $\Prf[\psi]$ ought to be a pushout, but we
don't expect to have pushouts in $\ThCat$, and moreover, we do not wish to
require the ability to ``split'' on a disjunction in all the syntactic sorts of
our type theory. Instead, we make $\ElMor$ \emph{think} that this pushout
exists, in a certain sense.

We begin by formulating the ``true'' disjunction in the free cocompletion
$\Psh{\ThCat}$. First, we have a characteristic map that decodes a
face condition $\phi : \Yo[\ThCat]{\COND}$ to a proposition $\DecodeCond{\phi} :
\Prop$.
\begin{equation}
  \DiagramSquare{
    nw/style = pullback,
    ne = \ObjTerm,
    se = \Prop,
    east = \PropTrue,
    east/style = >->,
    west/style = >->,
    south = \DecodeCond,
    sw = \Yo[\ThCat]{\COND},
    nw = \Yo[\ThCat]{\EMP},
    north = \cong,
    west = \Yo[\ThCat]{\Prf},
    south/style = exists,
  }
\end{equation}

As a first cut toward disjunction, we may define the \emph{non-representable}
subobject $\Pb{\PropTrue}{\Rst{\COND}{\lor}}$ of true
disjunctions of face conditions:
\begin{equation}
  \begin{tikzpicture}[diagram,baseline = (l/sw.base),node distance=2cm]
    \SpliceDiagramSquare<l/>{
      west/style = lies over,
      east/style = lies over,
      north/style = exists,
      nw = \Pb{\PropTrue}{\Rst{\COND}{\lor}},
      ne = \Pb{\PropTrue}{\lor},
      sw = \Yo[\ThCat]{\COND^2},
      se = \Prop^2,
      south = \MathSmall{\DecodeCond^2},
      south/node/style = above,
      height = 1.5cm,
      nw/style = dotted pullback,
      ne/style = dotted pullback,
    }
    \SpliceDiagramSquare<r/>{
      glue = west,
      glue target = l/,
      east/style = lies over,
      north/style = exists,
      ne = \PropTrue,
      se = \Prop,
      south = \lor,
      south/node/style = above,
      height = 1.5cm,
    }
    \node (tot) [right = of r/ne] {$\SUB{\Psh{\ThCat}}$};
    \node (base) [right = of r/se] {$\Psh{\ThCat}$};
    \path[fibration] (tot) edge (base);
    \path[->,bend right=30] (l/sw) edge node [below] {$\Rst{\COND}{\lor}$} (r/se);
  \end{tikzpicture}
\end{equation}

We will then define the disjunction of face conditions to be a representable
approximation of $\Pb{\PropTrue}{\Rst{\COND}{\lor}}$ that is respected by certain
judgments of XTT\@.

\begin{spec}[Disjunction]\label{spec:disjunction}
  We require a face condition formation map $\Mor[\lor]{\COND^2}{\COND} :
  \ThCat$ satisfying some conditions which we will describe forthwith.
  We then require an ``introduction'' rule
  $\Mor|>->|[\Kwd{in}_{\lor}]{\Pb{\PropTrue}{\Rst{\COND}{\lor}}}{\Yo[\ThCat]{\Pb{\Prf}{\lor}}}$
  in the slice $\Sl{\Psh{\ThCat}}{\Yo[\ThCat]{\COND^2}}$.
  The ``elimination'' rules for the disjunction are then
  expressed as a pair of internal orthogonality conditions in $\Sl{\Psh{\ThCat}}{\Yo[\ThCat]{\COND^2}}$:
  \begin{enumerate}
    \item We require $\IntOrth{\Kwd{in}_{\lor}}{\Yo[\ThCat]{\COND^2}^*\Yo[\ThCat]{\Prf}}$ in $\Sl{\Psh{\ThCat}}{\Yo[\ThCat]{\COND^2}}$, ensuring that the truth of a face condition may be established by eliminating a disjunction.
    \item We must have $\IntOrth{\Kwd{in}_{\lor}}{\Yo[\ThCat]{\COND^2}^*\Yo[\ThCat]{\El}}$ in $\Sl{\Psh{\ThCat}}{\Yo[\ThCat]{\COND^2}}$, ensuring that a ``matching family'' for a term on two disjuncts shall be a term under a disjunction.
      \qedhere
  \end{enumerate}
\end{spec}

\NewDocumentCommand\SPAN{}{\mathbf{Span}}

\begin{spec}[Falsehood]\label{spec:falsehood}
  By Specification~\ref{spec:dim-eq}, we have a map
  $\Mor[\bot\defeq\prn{=}\circ\gl{\Dim0,\Dim1}]{\EMP}{\COND}$. We may
  therefore test the truth of this ``false'' equation:
  \begin{equation}
    \begin{tikzpicture}[diagram,baseline=(sw.base)]
      \SpliceDiagramSquare{
        west/style = lies over,
        east/style = lies over,
        ne = \Prf,
        se = \COND,
        sw = \EMP,
        south = \bot,
        nw = \Pb{\Prf}{\bot},
        north/style = {>->,exists},
        south/style = >->,
        height = 1.5cm,
        nw/style = dotted pullback,
      }
      \node (tot) [right = of ne] {$\SUB{\ThCat}$};
      \node (base) [right = of se] {$\ThCat$};
      \path[fibration] (tot) edge (base);
    \end{tikzpicture}
  \end{equation}

  There is a universal comparison map
  $\Mor|>->|[\Kwd{in}_{\bot}]{\ObjInit}{\Yo[\ThCat]{\Pb{\Prf}{\bot}}}$ in
  $\Psh{\ThCat}$ given by the universal property of the initial object; to
  support the ``elimination rule'' of the inconsistent face condition, then, we
  require orthogonalities $\IntOrth{\Kwd{in}_{\bot}}{\Yo[\ThCat]{\Prf}}$ and
  $\IntOrth{\Kwd{in}_{\bot}}{\Yo[\ThCat]{\El}}$ in $\Psh{\ThCat}$.
\end{spec}

\begin{rem}
  The orthogonality conditions from Specification~\ref{spec:disjunction} may be
  restated in the language of $\Sl{\ThCat}{\COND^2}$.
  Let $\SPAN$ be the
  walking span, and let $\Mor[\Phi_\bullet]{\SPAN}{\Sl{\ThCat}{\COND^2}}$ be
  the following diagram in $\Sl{\ThCat}{\COND^2}$:
  \begin{equation*}
    \begin{tikzpicture}[diagram,baseline = (sw.base)]
      \node (nw) {$\Prf\times\Prf$};
      \node (sw) [left = 2cm of nw] {$\Pb{\Prf}{\pi_1}$};
      \node (ne) [right = 2cm of nw] {$\Pb{\Prf}{\pi_2}$};
      \path[->] (nw) edge (sw);
      \path[->] (nw) edge (ne);
    \end{tikzpicture}
  \end{equation*}

  The following canonical
  squares must be cartesian:
  \begin{gather*}
    \DiagramSquare{
      nw/style = pullback,
      width = 3.75cm,
      height = 1.5cm,
      nw = \IHom{\Pb{\Prf}{\lor}}{\Pb{\EMP}{\prn{\COND^2}}},
      sw = \IHom{\Pb{\Prf}{\lor}}{\Pb{\COND}{\prn{\COND^2}}},
      ne = \Lim[\SPAN]*{\IHom{\Phi_\bullet}{\Pb{\EMP}{\prn{\COND^2}}}},
      se = \Lim[\SPAN]*{\IHom{\Phi_\bullet}{\Pb{\EMP}{\prn{\COND^2}}}},
    }
    \quad
    \DiagramSquare{
      nw/style = pullback,
      width = 3.75cm,
      height = 1.5cm,
      nw = \IHom{\Pb{\Prf}{\lor}}{\Pb{\EL}{\prn{\COND^2}}},
      sw = \IHom{\Pb{\Prf}{\lor}}{\Pb{\TY}{\prn{\COND^2}}},
      ne = \Lim[\SPAN]*{\IHom{\Phi_\bullet}{\Pb{\EL}{\prn{\COND^2}}}},
      se = \Lim[\SPAN]*{\IHom{\Phi_\bullet}{\Pb{\TY}{\prn{\COND^2}}}},
    }
    \qedhere
  \end{gather*}
\end{rem}

\begin{rem}
  The orthogonality condition of Specification~\ref{spec:falsehood} may be
  restated in the language of $\ThCat$ as the requirement that the following
  canonical squares are cartesian:
  \begin{equation*}
    \DiagramSquare{
      nw/style = pullback,
      width = 2.5cm,
      height = 1.5cm,
      nw = \IHom{\Pb{\Prf}{\bot}}{\EMP},
      sw = \IHom{\Pb{\Prf}{\bot}}{\COND},
      ne = \EMP,
      se = \EMP,
    }
    \qquad
    \DiagramSquare{
      nw/style = pullback,
      width = 2.5cm,
      height = 1.5cm,
      nw = \IHom{\Pb{\Prf}{\bot}}{\EL},
      sw = \IHom{\Pb{\Prf}{\bot}}{\TY},
      ne = \EMP,
      se = \EMP,
    }
    \qedhere
  \end{equation*}
\end{rem}

The following notation will often be used in the internal language of $\ThCat$.

\begin{nota}[Restriction to a partial element]
  Let $\phi : \COND$ be a face condition, $A : \TY$ a type, and $a :
  \Prf\brk{\phi}\Rightarrow \El\brk{A}$ a partial element; we write
  $\El\brk{A\mid \phi \to a}$ for the collection of elements of $A$ which
  restrict to $a$ on $\phi$, \ie the elements $a':\El\brk{A}$ where
  $\lambda\_.a' = a : \Prf\brk{\phi}\Rightarrow A$.
\end{nota}

\begin{nota}[The boundary of a dimension]
  We will write $\Bdry{r} : \COND$ for the \emph{boundary} of a dimension
  $r:\DIM$, defined as the disjunction $\Bdry{r} \defeq (r=\Dim0)\lor(r=\Dim1)$.
\end{nota}

The boundary of a dimension will be used in specifying the closure under
\emph{path types}.

\begin{spec}[Closure under connectives]\label{spec:closure-under-connectives}
  We will require that typing judgment $\ElMor$ is closed under dependent
  sum, dependent product, and dependent \emph{path} types.  We first express the
  generic map underlying each connective, and then force it to be representable.

  \medskip
  \begin{mathflow}
    \Mor[\SG]{\CartArr{\ThCat}}{\CartArr{\ThCat}}
    \split
    \Mor[\PI]{\CartArr{\ThCat}}{\CartArr{\ThCat}}
    \split
    \Mor[\PATH]{\CartArr{\ThCat}}{\CartArr{\ThCat}}
  \end{mathflow}
  \medskip

  Let $\Mor[f]{Y}{X}$, \ie $f$ an object of $\CartArr{\ThCat}$; we first
  define the bases $X^Q$ of the functorial action $f^Q$ where $Q\in \brc{\Pi,\Sigma,\mathrm{P}}$:
  \begin{align*}
    \SG{X} &= \brk{A : X; B : f\brk{A}\Rightarrow X}\\
    \PI{X} &= \brk{A : X; B : f\brk{A}\Rightarrow X}\\
    \PATH{X} &= \brk{A : \DIM\Rightarrow X; a : \prn{i : \DIM,\_ :\Bdry{i}}\Rightarrow f\brk{A\prn{i}}}
  \end{align*}

  Then, we define the rest of the action in the internal language of $\Sl{\ECat}{\prn{X^Q}}$:
  \begin{align*}
    \SG{f}\brk{A;B} &= \brk{a : f\brk{A}; b : f\brk{B\prn{a}}}\\
    \PI{f}\brk{A;B} &= \prn{a : f\brk{A}}\Rightarrow f\brk{B\prn{a}}\\
    \PATH{f}\brk{A;a} &= \brc{p : \prn{i : \DIM}\Rightarrow f\brk{A\prn{i}}\mid \lambda i.\lambda\_.p\prn{i} = a}
  \end{align*}

  The closure of the type theory under these connectives is then accomplished by
  requiring in $\CartArr{\ThCat}$ algebras $\Mor{\SG{\El}}{\El}$,
  $\Mor{\PI{\El}}{\El}$, and $\Mor{\PATH{\El}}{\El}$.
\end{spec}

\begin{rem}
  Unfolding Specification~\ref{spec:closure-under-connectives} into the language
  of $\ThCat$, this means that we have the following cartesian
  squares:
  \begin{gather*}
    \ElSquare{
      nw = \SG{\EL},
      west = \SG{\El},
      sw = \SG{\TY},
      north = \TmPair,
      south = \TySg,
      north/style = exists,
      south/style = exists,
    }
    \quad
    \ElSquare{
      nw = \PI{\EL},
      west = \PI{\El},
      sw = \PI{\TY},
      north = \TmLam,
      south = \TyPi,
      north/style = exists,
      south/style = exists,
    }
    \quad
    \ElSquare{
      nw = \PATH{\EL},
      west = \PATH{\El},
      sw = \PATH{\TY},
      north = \TmAbs,
      south = \TyPath,
      north/style = exists,
      south/style = exists,
    }
    \qedhere
  \end{gather*}
\end{rem}

We will add a type of \emph{booleans}; as usual in type theory, the boolean
type is \emph{not} the coproduct of the point with itself, but a weak version
thereof.  The simplest way to specify a weak coproduct is by means of a left
lifting structure (Definition~\ref{def:left-lifting-structure}) internal to the
\emph{free cocompletion} $\Psh{\ThCat}$ of $\ThCat$.

\begin{spec}[Booleans]\label{spec:booleans}
  To specify the type of booleans, we then require a type
  $\Mor[\TyBool]{\EMP}{\TY}$ together with the following objects in
  $\Psh{\ThCat}$:

  \begin{enumerate}

    \item A morphism $\Mor[i_\TyBool]{\Two{\Psh{\ThCat}}}{\Yo[\ThCat]\prn{\Pb{\El}{\TyBool}}} : \Psh{\ThCat}$,
      where $\Two{\Psh{\ThCat}}$ is the coproduct $\ObjTerm{\Psh{\CCat}} + \ObjTerm{\Psh{\CCat}}$
      in the presheaf category.

    \item A left lifting structure for $i_\TyBool$ with respect to
      $\Yo[\ThCat]{\ElMor}$, \ie $\Kwd{ind}_{\TyBool} : i_{\TyBool}\pitchfork
      \Yo[\ThCat]{\El}$.
      \qedhere

  \end{enumerate}
\end{spec}

\begin{rem}
  We may unravel Specification~\ref{spec:booleans} into the language of
  $\ThCat$ using Lemma~\ref{lem:colimit-lift}. First, the morphism $i_\TyBool$
  corresponds to exactly two introduction forms
  $\Mor[\TmTt,\TmFf]{\EMP}{\Pb{\El}{\TyBool}}$ in $\ThCat$; the left lifting
  structure $\Kwd{ind}_{\TyBool}$ amounts to a choice of section for the
  cartesian gap map of the following canonical square:
  \begin{equation}
    \DiagramSquare{
      width=3cm,
      height = 1.5cm,
      nw = \IHom{\Pb{\El}{\TyBool}}{\EL},
      sw = \IHom{\Pb{\El}{\TyBool}}{\TY},
      se = \TY\times\TY,
      ne = \EL\times\EL,
    }
  \end{equation}

  Translated into the type theoretic internal language of $\ThCat$, this
  corresponds to a dependent eliminator of the following form:
  \begin{gather*}
    \Kwd{ind}_{\TyBool} :
    \prn{
      C : \El[\TyBool]\Rightarrow \TY,
      c_0 : \El[C\prn{\TmTt}],
      c_1 : \El[C\prn{\TmFf}],
      b : \El[\TyBool]
    }
    \longrightarrow
    \El[C\prn{b}]
    \\
    C,c_0,c_1,b\mid b = \TmTt \vdash
    \Kwd{ind}_{\TyBool}\prn{C,c_0,c_1,b} = c_0
    \\
    C,c_0,c_1,b\mid b = \TmFf \vdash
    \Kwd{ind}_{\TyBool}\prn{C,c_0,c_1,b} = c_1
    \qedhere
  \end{gather*}
\end{rem}

\subsection{Universe \`a la Tarski}

In this section, we will specify the closure of $\ThCat$ under a universe \`a
la Tarski of Bishop sets. First, we must define what it means semantically to
be a Bishop set, in the style of Coquand~\cite{coquand:2017:bish}; following
previous work, we refer to this condition as \emph{boundary
separation}~\cite{sterling-angiuli-gratzer:2019}.

\begin{defi}
  A family $\Mor[f]{X}{Y}:\ThCat$ is said to be \emph{boundary separated} when
  $\Pb{f}{\DIM}$ is separated with respect to the boundary inclusion
  $\Mor|>->|{\brk{i : \DIM\vdash \Bdry{i}}}{\brk{i:\DIM\vdash\EMP}}$ in the
  slice $\Sl{\ThCat}{\DIM}$, in the sense of Definition~\ref{def:separation}.
  Unfolding, this means that for any map $\Mor[r]{Z}{\DIM}$, any square of the
  following shape in $\ThCat$ has \emph{at most} one lift:
  \begin{equation}
    \begin{tikzpicture}[diagram,baseline=(sw.base)]
      \SpliceDiagramSquare{
        nw = \brk{Z \vdash \Bdry{r}},
        sw = \brk{Z \vdash \EMP},
        east = f,
        ne = X,
        se = Y,
        height = 1.75cm,
      }
      \path[->,exists] (sw) edge (ne);
    \end{tikzpicture}
  \end{equation}

  The concept of boundary separation can be expressed in the language of $\ThCat$ as follows:
  \[
    i : \DIM; x : X; a,b : f[x] \mid (\lambda \alpha.a) =_{\prn{\Bdry{i}\Rightarrow{}f[x]}} (\lambda \alpha.b) \vdash a =_{f[x]} b
    \qedhere
  \]
\end{defi}

\begin{spec}[Universe \`a la Tarski]
  To specify a universe \`a la Tarski, we require a type $\Mor[\TskU]{\EMP}{\TY}$
  together with a decoding map $\Mor[\TskDec]{\Pb{\El}{\TskU}}{\TY}$. We will write
  $\TskUTY$ for the fiber $\Pb{\El}{\TskU}:\ThCat$.  Pulling back $\ElMor$ along the
  decoding map, we have a new representable map
  $\Mor[\TskUEl]{\TskUEL}{\TskUTY}$:
  \begin{equation}
    \ElSquare{
      west/style = exists,
      nw = \TskUEL,
      sw = \TskUTY,
      west = \TskUEl,
      south = \TskDec,
      height = 1.75cm,
    }
  \end{equation}

  We then require that the family $\TskUMor$ be boundary
  separated.
\end{spec}

\begin{spec}[Codes for connectives]
  The closure of $\TskUTY$ under various connectives of type theory is accomplished
  as follows:
  \begin{gather*}
    \DiagramSquare{
      nw = \SG{\TskUTY},
      sw = \SG{\TY},
      ne = \TskUTY,
      se = \TY,
      north = \TmCodeSg,
      south = \TySg,
      east = \TskDec,
      west = \SG{\TskDec},
      north/style = exists,
    }
    \quad
    \DiagramSquare{
      nw = \PI{\TskUTY},
      sw = \PI{\TY},
      ne = \TskUTY,
      se = \TY,
      north = \TmCodePi,
      south = \TyPi,
      east = \TskDec,
      west = \PI{\TskDec},
      north/style = exists,
    }
    \\
    \DiagramSquare{
      nw = \PATH{\TskUTY},
      sw = \PATH{\TY},
      ne = \TskUTY,
      se = \TY,
      north = \TmCodePath,
      south = \TyPath,
      east = \TskDec,
      west = \PATH{\TskDec},
      north/style = exists,
    }
    \quad
    \begin{tikzpicture}[diagram,node distance=2cm,baseline = (T.base)]
      \node (EMP) {$\EMP$};
      \node (U) [right = of EMP] {$\TskUTY$};
      \node (T) [below = of U] {$\TY$};
      \path[->,exists] (EMP) edge node [above] {$\TmCodeBool$} (U);
      \path[->] (U) edge node [right] {$\TskDec$} (T);
      \path[->] (EMP) edge node [sloped,below] {$\TyBool$} (T);
    \end{tikzpicture}
    \qedhere
  \end{gather*}
\end{spec}

\begin{spec}[Type-case]\label{spec:type-case}
  The type-case construct of XTT's closed universe is implemented by a left
  lifting structure. First, we define the following coproduct in the free
  cocompletion $\Psh{\ThCat}$:
  \begin{equation*}
    \mathfrak{F}_\TskUTY \defeq \Yo[\ThCat]{\SG{\TskUTY}} + \Yo[\ThCat]{\PI{\TskUTY}} + \Yo[\ThCat]{\PATH{\TskUTY}} + \Yo[\ThCat]{\EMP}
  \end{equation*}

  The codes for each type constructor can be arranged into an algebra
  $\Mor[\alpha_\TskUTY]{\mathfrak{F}_\TskUTY}{\Yo[\ThCat]{\TskUTY}}$:
  \begin{equation*}
    \Mor[
      \brk{\Yo[\ThCat]{\TmCodeSg}\mid\Yo[\ThCat]{\TmCodePi}\mid\Yo[\ThCat]{\TmCodePath}\mid\Yo[\ThCat]{\TmCodeBool}}
    ]{\mathfrak{F}_\TskUTY}{\Yo[\ThCat]{\TskUTY}}
  \end{equation*}

  We then require a left lifting structure $\Kwd{case}_{\TskUTY} : \alpha_{\TskUTY} \pitchfork
  \Yo[\ThCat]{\El}$, which provides solutions to lifting problems of the following shape:
  \begin{equation*}
    \begin{tikzpicture}[diagram,baseline = (sw.base)]
      \SpliceDiagramSquare{
        width = 3cm,
        height = 2.5cm,
        nw = \Yo[\ThCat]{X}\times\mathfrak{F}_\TskUTY,
        sw = \Yo[\ThCat]\prn{X\times\TskUTY},
        ne = \Yo[\ThCat]{\EL},
        se = \Yo[\ThCat]{\TY},
        west = \Yo[\ThCat]{X}\times\alpha_{\TskUTY},
        east = \Yo[\ThCat]{\El},
        north = c,
        south =  C
      }
      \path[->,exists] (sw) edge node [desc] {$\MathSmall{\Kwd{case}_{\TskUTY}^X\prn{C,c}}$} (ne);
    \end{tikzpicture}
    \qedhere
  \end{equation*}
\end{spec}

\begin{rem}
  In the language of $\ThCat$, the lifting structure from
  Specification~\ref{spec:type-case} amounts to a section of the cartesian gap
  map for the following canonical square:
  \begin{equation*}
    \DiagramSquare{
      width = 5.5cm,
      height = 1.5cm,
      nw = \IHom{\TskUTY}{\EL},
      sw = \IHom{\TskUTY}{\TY},
      ne =
      \MathSmall{
        \IHom{\SG{\TskUTY}}{\EL}\times
        \IHom{\PI{\TskUTY}}{\EL}\times
        \IHom{\PATH{\TskUTY}}{\EL}\times
        \EL
      },
      se =
      \MathSmall{
        \IHom{\SG{\TskUTY}}{\TY}\times
        \IHom{\PI{\TskUTY}}{\TY}\times
        \IHom{\PATH{\TskUTY}}{\TY}\times
        \TY
      },
    }
    \qedhere
  \end{equation*}
\end{rem}

The types encoded by XTT's closed universe must be ``fibrant'' in the sense
that they support transport along paths and composition of paths. We will
express these as two separate left lifting structures, with suitable
compatibility laws.

\begin{spec}\label{spec:coercion}
  \emph{Coercion} will be specified as a left lifting structure in the slice
  $\Sl{\ThCat}{\DIM}$. First, observe that we have a diagonal map
  $\Mor|>->|[\brk{r:\DIM\vdash r}]{\brk{r : \DIM\vdash \EMP}}{\brk{r:\DIM\vdash\DIM}} :
  \Sl{\ThCat}{\DIM}$.
  We then require a left lifting structure $\Kwd{coe} : \brk{r:\DIM\vdash r}\pitchfork
  \brk{r : \DIM\vdash \TskUEl}$ in $\Sl{\ThCat}{\DIM}$.
\end{spec}

\begin{spec}\label{spec:composition}
  \emph{Composition} is specified by a left lifting structure in the slice
  $\Sl{\ThCat}{\DIM^2}$: the first dimension $\brk{r,s:\DIM\vdash r}$ plays the
  same role as the generic dimension in Specification~\ref{spec:coercion}, and
  the second dimension $\brk{r,s:\DIM\vdash s}$ generates the boundary
  cofibration $\brk{r,s:\DIM\vdash\Bdry{s}}$ along which we are extending a
  partial line.

  Consider the map $\Mor|>->|[\iota]{\brk{r,s:\DIM\vdash \brc{i:\DIM\mid
  \Prf\brk{i=r\lor \Bdry{s}}}}}{\brk{r,s:\DIM\vdash \DIM}} :
  \Sl{\ThCat}{\DIM^2}$; we then require a left lifting structure $\Kwd{comp} :
  \iota \pitchfork \brk{r,s:\DIM\vdash \TskUEl}$ in $\Sl{\ThCat}{\DIM^2}$.
\end{spec}

\begin{rem}\label{rem:com-coe}
  Unfolding Specifications~\ref{spec:coercion}~and~\ref{spec:composition} into the language of $\ThCat$, we have
  the following constants:
  \begin{mathflow}
    \inferrule{
      r:\DIM & \hat{A} : \DIM\Rightarrow\TskUTY &
      a : \TskUEl\brk{\hat{A}\prn{r}}
    }{
      \Coe{r}{\bullet}{\hat{A}}{a} : \prn{i:\DIM} \Rightarrow \TskUEl\brk{\hat{A}\prn{i} \mid i=r\to a}
    }
    \split
    \inferrule{
      r,s:\DIM & \hat{A} : \DIM\Rightarrow\TskUTY &
      a : \prn{i:\DIM,\alpha:\Prf[i=r\lor\Bdry{s}]}\Rightarrow\TskUEl\brk{\hat{A}\prn{i}}
    }{
      \Comp{r}{\bullet}{s}{\hat{A}}{a} : \prn{i:\DIM} \Rightarrow \TskUEl\brk{\hat{A}\prn{i}\mid i=r\lor\Bdry{s}\to a\prn{i}}
    }
  \end{mathflow}

  \emph{Homogeneous composition} is the special case of composition where the
  type code is constant:
  \[
    \Hcomp{r}{\bullet}{s}{\hat{A}}{a} \defeq
    \Comp{r}{\bullet}{s}{\lambda\_.\hat{A}}{a}
    \qedhere
  \]
\end{rem}

\begin{spec}[Regularity]\label{spec:regularity}
  We will require the following regularity laws:
  \begin{align*}
    \Coe{r}{r'}{\lambda\_.\hat{A}}{a} &= a
    \\
    \Hcomp{r}{r'}{s}{\hat{A}}{\prn{\lambda\_~\_.a}} &= a \qedhere
  \end{align*}
\end{spec}

\begin{lem}\label{lem:hcom-com-compatibility}
  As a consequence of boundary separation, the following compatibility law between homogeneous and heterogeneous equality holds:
  \[
    \Comp{r}{r'}{s}{\hat{A}}{a}
    =
    \Coe{r}{r'}{\hat{A}}{
      \prn{
        \Hcomp{r}{r'}{s}{\hat{A}r}{
          \prn{
            \lambda i,\alpha.
            \Coe{i}{r}{\hat{A}}{a\prn{i,\alpha}}
          }
        }
      }
    }
    \qedhere
  \]
\end{lem}

\begin{spec}[Coercion at connectives]\label{spec:coe-at-connectives}
  In order to satisfy the \emph{canonicity} property, our theory must further
  constrain the behavior of the coercion operation at each type.
  \begin{align*}
    \pi_1\prn{
      \Coe{r}{r'}{\lambda i.\TmCodeSg\prn{\hat{A}i,\hat{B}i}}{p}
    }
    &=
    \Coe{r}{r'}{\hat{A}}{\pi_1\prn{p}}
    \\
    \pi_2\prn{
      \Coe{r}{r'}{\lambda i.\TmCodeSg\prn{\hat{A}i,\hat{B}i}}{p}
    }
    &=
    \Coe{r}{r'}{
      \lambda i.
      \hat{B}i\prn{
        \DelimProtect{
          \Coe{r}{i}{
            \hat{A}
          }{
            \pi_1\prn{p}
          }
        }
      }
    }{\pi_2\prn{p}}
    \\
    \prn{
      \Coe{r}{r'}{\lambda i.\TmCodePi\prn{\hat{A}i,\hat{B}i}}{p}
    }\prn{a}
    &=
    \Coe{r}{r'}{\lambda i.
    \hat{B}i\prn{
      \DelimProtect{
        \Coe{r'}{i}{\hat{A}}{a}
      }
    }
    }{p\prn{\Coe{r'}{r}{\hat{A}}{a}}}
    \\
    \prn{
      \Coe{r}{r'}{\lambda i.\TmCodePath\prn{\hat{A}i,ai}}{p}
    }\prn{s}
    &=
    \Comp{r}{r'}{s}{\hat{A}s}{
      \prn{
        \lambda \_.
        p\prn{s}
      }
    }
    \qedhere
  \end{align*}
\end{spec}

\begin{lem}[Composition at connectives]\label{lem:hcom-at-connectives}
  The behavior of the homogeneous composition operations at each
  connective is completely determined by
  Specifications~\ref{spec:regularity}~and~\ref{spec:coe-at-connectives}.
  In particular, we have:
  \begin{align*}
    \pi_1\prn{
      \Hcomp{r}{r'}{s}{\TmCodeSg\prn{\hat{A},\hat{B}}}{p}
    } &=
    \Hcomp{r}{r'}{s}{\hat{A}}{\prn{\lambda i,\alpha.\pi_1\prn{p\prn{i,\alpha}}}}
    \\
    \pi_2\prn{
      \Hcomp{r}{r'}{s}{\TmCodeSg\prn{\hat{A},\hat{B}}}{p}
    } &=
    \Comp{r}{r'}{s}{
      \lambda i.\hat{B}\prn{
        \DelimProtect{
          \Hcomp{r}{i}{s}{\hat{A}}{\prn{\lambda i,\alpha.\pi_1\prn{p\prn{i,\alpha}}}}
        }
      }
    }{\prn{\lambda i,\alpha.\pi_2\prn{p\prn{i,\alpha}}}}
    \\
    \prn{
      \Hcomp{r}{r'}{s}{\TmCodePi\prn{\hat{A},\hat{B}}}{p}
    }\prn{a}
    &=
    \Hcomp{r}{r'}{s}{
      \hat{B}\prn{a}
    }{\prn{\lambda i,\alpha.p\prn{i,\alpha,a}}}
    \\
    \prn{
      \Hcomp{r}{r'}{s}{\TmCodePath\prn{\hat{A},a}}{p}
    }\prn{s'}
    &=
    \Hcomp{r}{r'}{s}{\hat{A}s'}{
      \prn{
        \lambda i,\alpha.
        p(i,\alpha,s')
      }
    }
    \\
    \Hcomp{r}{r'}{s}{\TmCodeBool}{p} &= p\prn{r,*}
  \end{align*}
\end{lem}

\begin{proof}
  By boundary separation, pivoting on $s:\DIM$.
\end{proof}

\subsection{Summary}

We have expressed the necessary and sufficient conditions on a representable
map category to be an algebra for the XTT language. From now on, we will take
$\ThCat$ to be the \emph{smallest} representable map category satisfying the
present specification, hence $\ThCat$ is the syntactic category of XTT\@. The
existence of such a representable map category follows from Proposition 5.16 of~\cite{uemura:2019}. A model of XTT in a representable map category $\CCat$ is
then precisely a functor $\Mor[\Mod{\CCat}]{\ThCat}{\DF{\CCat}}$.

\section{Cubical computability structures}%
\label{sec:computability}

In this section we develop the building blocks of the canonicity proof for XTT\@. While
eventually these shall be applied to the initial model, we work with an arbitrary
model $\Mor[\Mod{\CCat}]{\ThCat}{\DF{\CCat}}$.

\begin{defi}\label{def:cube}
  The category $\CUBE$ of \emph{Cartesian cubes} is the free category
  with finite products generated by a bi-pointed object, \ie the Lawvere
  category of the theory of bi-pointed objects~\cite{awodey-notes:2015}. We
  will write $\cSET$ for the category $\Psh{\CUBE}$ of \emph{Cartesian cubical
  sets}.
\end{defi}

A finite product preserving
functor $\Mor[\DimIncl]{\CUBE}{\CCat}$ is determined by an \emph{interval
object} (\ie a bi-pointed object) in $\CCat$, which we construct in
Construction~\ref{con:cx-interval} below.

\begin{con}[An interval object in $\CCat$]\label{con:cx-interval}
  Because the Yoneda embedding reflects limits, the terminal discrete fibration
  $\Mod{\CCat}{\EMP} : \DF{\CCat}$ is represented by the terminal object of
  $\CCat$, which we may write $\Repr{\Mod{\CCat}{\EMP}}$.
  We have required that the terminal map
  $\Mod{\CCat}{\prn{\Mor{\DIM}{\EMP}}}:\DF{\CCat}$ is a representable family;
  therefore, the discrete fibration $\Mod{\CCat}{\DIM}:\DF{\CCat}$ is represented
  by some object $\Repr{\Mod{\CCat}{\DIM}} : \CCat$. Since the Yoneda embedding
  is fully faithful, the constants
  $\Mod{\CCat}{\prn{\Mor[\Dim0,\Dim1]{\EMP}{\DIM}}}$ are represented by two maps
  $\Mor[\Repr{\Dim0},\Repr{\Dim1}]{\Repr{\Mod{\CCat}{\EMP}}}{\Repr{\Mod{\CCat}{\DIM}}} : \CCat$.
\end{con}

From Construction~\ref{con:cx-interval} we obtain a suitable functor
$\Mor[\DimIncl]{\CUBE}{\CCat}$, which can also be seen to be fully faithful; in
other words, the category of cubes $\CUBE$ is the full subcategory of $\CCat$ spanned
by cubical objects.  By change of base, we obtain a reindexing functor
$\Mor[\DimIncl^*]{\Psh{\CCat}}{\cSET}$ which has both left and right adjoints
by Kan extension. Composing with the equivalence $\DF{\CCat}\simeq\Psh{\CCat}$,
we define a functor $\F$ to glue along below:
\begin{equation}
  \begin{tikzpicture}[diagram,baseline = (0.base)]
    \node (0) {$\DF{\CCat}$};
    \node (1) [right = of 0] {$\Psh{\CCat}$};
    \node (2) [right = of 1] {$\cSET$};
    \path[->] (0) edge node [above] {$\simeq$} (1);
    \path[->] (1) edge node [above] {$\DimIncl^*$} (2);
    \path[->, bend right = 30,exists] (0) edge node [below] {$\F$} (2);
  \end{tikzpicture}
\end{equation}

Thus the functor $\Mor[\F]{\DF{\CCat}}{\Psh{\CUBE}}$ is in fact an algebraic
morphism of logoi, and therefore a good candidate for type theoretic
gluing~\cite{sterling-angiuli:2020}.

\begin{con}[A cubical nerve]\label{con:cubical-nerve}
  By composing the change of base $\Mor[\F]{\DF{\CCat}}{\cSET}$ with
  the Yoneda embedding, we obtain a \emph{nerve} functor from the category of
  contexts of $\Mod{\CCat}$ into cubical sets:
  \begin{equation*}
    \begin{tikzpicture}[diagram,baseline=(cSET.base)]
      \node (C) {$\CCat$};
      \node (DF/C) [right = of C] {$\DF{\CCat}$};
      \node (cSET) [below = of DF/C] {$\cSET$};
      \path[->,embedding] (C) edge node [above] {$\DFYo[\CCat]$} (DF/C);
      \path[->] (DF/C) edge node [right] {$\F$} (cSET);
      \path[->,exists] (C) edge node [sloped,below] {$\DimInclNv$} (cSET);
    \end{tikzpicture}
    \qedhere
  \end{equation*}
\end{con}

In 2015, Awodey suggested the idea of gluing along a \emph{cubical nerve} like
Construction~\ref{con:cubical-nerve} to develop the metatheory of cubical type
theory; to many researchers, it seemed as though Huber's operational proof of
canonicity for cubical type theory~\cite{huber:2018} could be reconstructed in
a mathematical way. Since then, Awodey and Fiore have used this cubical gluing
technique to study a version of intensional type theory with an interval in
unpublished joint work; in 2019, the present authors applied cubical gluing to
prove canonicity for an earlier version of
XTT~\cite{sterling-angiuli-gratzer:2019}.

\begin{lem}%
  \label{lem:cubical-nerve-is-flat}
  The cubical nerve is \emph{flat}.
\end{lem}

\begin{proof}
  By Diaconescu's theorem~\cite{borceux:2010:vol3}, it suffices to show that
  $\Mor[\operatorname{Lan}_{\DFYo[\CCat]}\DimInclNv]{\DF{\CCat}}{\cSET}$ is an
  algebraic morphism of logoi; but this Kan extension is just $\F\simeq
  \DimIncl^*$, which has both left and right adjoints.
\end{proof}

\subsubsection*{The classical gluing construction}

We may obtain a category of \emph{glued contexts} as the comma category $\Comma{\cSET}{\DimInclNv}$,
whose objects are ``cubically computable contexts'' and whose morphisms are natural transformations
of cubical sets that are tracked by substitutions from the model $\Mod{\CCat}$. Intuitively, a glued
context is a pair of a syntactic context with a (proof-relevant) Kripke predicate indexed in
dimension variable contexts. By analogy with~\cite{tait:1967}, we regard the image of these
predicates as the \emph{computable} closing substitutions of contexts. Because glued substitutions
are commuting squares, they automatically send computable elements to computable elements; the top
map of such a square sends computability witnesses to computability witnesses.

By defining a model of XTT in this category (\ie collections of glued types and elements closed
under the rules), we obtain a categorical version of a logical relations model of XTT\@. Our
categorical perspective immediately offers some advantages over ``free-style'' logical relations for
dependent type theory---much of the indexing can be moved behind categorical abstractions, and
certain results become automatic, such as computable substitutions being closed under composition.
By choosing certain predicates appropriately, we can moreover ensure that the computable elements of
boolean types are precisely $\TmTt$ and $\TmFf$, \etc.

This is the perspective pursued in previous work on gluing for strict type
theory~\cite{sterling-angiuli-gratzer:2019,kaposi-huber-sattler:2019,coquand-huber-sattler:2019,gratzer-kavvos-nuyts-birkedal:2020};
in recent joint work, the first two authors of this paper have argued that it
is considerably simpler to first glue over $\DF{\CCat}$ rather than $\CCat$,
and then restrict further to $\CCat$ by pulling back along
$\EmbMor[\DFYo[\CCat]]{\CCat}{\DF{\CCat}}$~\cite{sterling-angiuli:2020}.

Following~\cite{sterling-angiuli:2020}, we prefer to first glue along
$\Mor[\F]{\DF{\CCat}}{\cSET}$ rather than
$\Mor[\DimInclNv]{\CCat}{\cSET}$, because the comma category
$\GlCat = \Comma{\cSET}{\F}$ has more regular properties than
$\CptGlCat = \Comma{\cSET}{\DimInclNv}$, being a Grothendieck logos~\cite{sga:4}.
The resulting two-step gluing process amounts to enlarging the collection of
computability structures to ones that lie over arbitrary discrete fibrations
$X:\DF{\CCat}$, in addition to the familiar ones that lie directly over
contexts $\Gamma:\CCat$.
Computability structures lying over a context $\Gamma:\CCat$ will be referred
to as ``compact''; the fundamental example of a computability structure which is
not compact is the computability structure of computable types, lying over the
non-representable discrete fibration $\Mod{\CCat}{\TY}:\DF{\CCat}$.

The general computability structures are connected to a \emph{model} of type
theory (which must take place in discrete fibrations over \emph{compact}
computability structures) via a nerve--realization adjunction. The nerve 
functor $\Mor{\GlCat}{\DF{\CptGlCat}}$ is fully faithful and even locally
cartesian closed, and may therefore be used to transform general computability
families into constituents of a model of type theory over $\CptGlCat$ which are
otherwise vastly more difficult to compute.

\subsection{General and compact computability structures}

The picture painted at the end of the previous section is depicted in
Diagram~\ref{diag:gluing} below:
\begin{equation}\label{diag:gluing}
  \begin{tikzpicture}[diagram,baseline=(l/sw.base)]
    \SpliceDiagramSquare<l/>{
      nw/style = pullback, ne/style = pullback,
      width = 2.5cm,
      height = 1.75cm,
      nw = \CptGlCat,
      ne = \GlCat,
      sw = \CCat,
      se = \DF{\CCat},
      south = \DFYo[\CCat],
      west/style = fibration,
      east/style = fibration,
      east/node/style = upright desc,
      north/style = embedding,
      south/style = embedding,
      north = \CptIncl,
      east = \GlProj,
      west = \CptGlProj,
      south/node/style = upright desc,
    }
    \SpliceDiagramSquare<r/>{
      glue = west,
      glue target = l/,
      width = 2.5cm,
      height = 1.75cm,
      east/style = fibration,
      south = \F,
      ne = \ArrCat{\cSET},
      se = \cSET,
      east = \partial_1,
      south/node/style = upright desc,
    }
    \path[->,bend right=30] (l/sw) edge node [below] {$\DimInclNv$} (r/se);
  \end{tikzpicture}
\end{equation}

We call $\GlCat$ the category of (general) computability structures, whereas
the full subcategory $\CptGlCat$ is the category of \emph{compact}
computability families. Here the fibration $\GlProj$ projects a discrete
fibration from a general computability structure, and the restriction
$\CptGlProj$ projects a context from a compact computability structure. The
``compact'' terminology is justified by Theorem~\ref{thm:cpt-incl-dense} below.

\begin{thmC}[{\cite[Theorem~4.8]{sterling-angiuli:2020}}]\label{thm:cpt-incl-dense}
  The subcategory inclusion $\EmbMor[\CptIncl]{\CptGlCat}{\GlCat}$ is
  \textbf{dense} in the sense that every object $X:\GlCat$ is canonically the
  colimit of the following diagram of compact objects:
  \begin{equation}\label{diag:cpt-incl-density}
    \begin{tikzpicture}[diagram]
      \node (0) {$\Comma{\CptIncl}{\brc{X}}$};
      \node (1) [right = of 0] {$\CptGlCat$};
      \node (2) [right = of 1] {$\GlCat$};
      \path[->] (0) edge node [above] {$\pi_{\CptIncl}$} (1);
      \path[->,embedding] (1) edge node [above] {$\CptIncl$} (2);
    \end{tikzpicture}
  \end{equation}
\end{thmC}

We reproduce a variant of the proof given by Sterling and Angiuli~\cite{sterling-angiuli:2020}.

\begin{proof}
  Using the universality of colimits in a Grothendieck logos and the fact that
  $\F$ is cocontinuous, we will show that $X$ is the colimit of a
  particular canonical diagram in $\GlCat$ which is final for
  Diagram~\ref{diag:cpt-incl-density}.
  First, we use the dual Yoneda lemma (that $\CCat$ is dense in $\DF{\CCat}$)
  to observe that $\GlProj{X}$ is the colimit of the following diagram:
  \begin{equation}\label{diag:cpt-incl-density:1}
    \begin{tikzpicture}[diagram,baseline = (2.base)]
      \node (0) {$\Comma{\DFYo[\CCat]}{\brc{\GlProj{X}}}$};
      \node (1) [right = of 0] {$\CCat$};
      \node (2) [right = of 1] {$\DF{\CCat}$};
      \path[->] (0) edge node [above] {$\pi_{\CCat}$} (1);
      \path[->,embedding] (1) edge node [above] {$\DFYo[\CCat]$} (2);
    \end{tikzpicture}
  \end{equation}

  Each leg $\Mor[\alpha_i]{\DFYo[\CCat]{C_i}}{\GlProj{X}}$ of the colimiting
  cone for Diagram~\ref{diag:cpt-incl-density:1} induces a cartesian lift at
  $X$ in the gluing fibration:
  \begin{equation}\label{diag:cpt-incl-density:2}
    \begin{tikzpicture}[diagram,baseline = (sw.base)]
      \SpliceDiagramSquare{
        height = 1.5cm,
        west/style = lies over,
        east/style = lies over,
        sw = \DFYo[\CCat]{C_i},
        se = \GlProj{X},
        south = \alpha_i,
        ne = X,
        nw = \alpha_i^*X,
        north = \alpha_i^\dagger{X}
      }
    \end{tikzpicture}
  \end{equation}

  We will see that the resulting cocone
  $\brc{\Mor[\alpha_i^\dagger{X}]{\alpha_i^*X}{X}}$ in $\GlCat$ is
  universal for $X$. Because colimits in the comma category may be computed
  pointwise, we may reason as follows.  Cartesian lifts in the gluing fibration
  are computed as pullbacks in $\cSET$; because colimits in the presheaf logos
  $\cSET$ are universal, it suffices to check that the cone
  $\brc{\Mor[\F{\alpha_i}]{\F{\DFYo[\CCat]{C_i}}}{\F{\GlProj{X}}}}$
  is universal in $\cSET$. But $\Mor[\F]{\DF{\CCat}}{\cSET}$ is
  cocontinuous, so it is enough that
  $\brc{\Mor[\alpha_i]{\DFYo[\CCat]{C_i}}{\GlProj{X}}}$ is universal.
  Finally, the universality of Cartesian lifts ensures that the collection
  $\brc{{\alpha_i^*X}}$ is final for Diagram~\ref{diag:cpt-incl-density}.
\end{proof}

\begin{defi}
  An object $X:\GlCat$ is \emph{compact} when it lies in the essential image of $\CptIncl$;
  equivalently, when $\GlProj\prn{X}:\DF{\CCat}$ is representable.
\end{defi}

We may impose the structure of a representable map category on $\GlCat$, based
on generalizing the notion of compactness from objects to families.

\begin{defi}
  A family $\Mor[p]{Y}{X}:\GlCat$ is called \emph{compact} when every change of
  base to a compact object $\CptIncl{\Gamma}$ is compact in the sense of the
  following diagram:
  \begin{equation*}
    \DiagramSquare{
      nw/style = pullback,
      ne = Y,
      east = p,
      se = X,
      south = x,
      sw = \CptIncl{Z},
      nw = \CptIncl{Y_x}
    }
    \qedhere
  \end{equation*}
\end{defi}

Then, we say that the representable families in $\GlCat$ are exactly the
compact ones. It is simple to verify that this class of maps satisfies the
axioms of a representable map category: they are clearly closed under change of
base, and since $\GlCat$ is locally cartesian closed, pushforwards along
representable maps always exist.

\begin{lem}\label{lem:gl-proj-adjoints}
  The gluing fibrations $\FibMor[\GlProj]{\GlCat}{\DF{\CCat}}$ and
  $\FibMor[\CptGlProj]{\CptGlCat}{\CCat}$ have both left and right adjoints,
  and are therefore both continuous and cocontinuous; moreover, both adjoints are sections.
  \begin{mathflow}
    \begin{tikzpicture}[diagram,baseline=(DF/C.base)]
      \node (G) {$\GlCat$};
      \node (DF/C) [below = of G] {$\DF{\CCat}$};
      \path[fibration] (G) edge node [upright desc] {$\GlProj$} (DF/C);
      \path[->,bend left = 60] (DF/C) edge node [left] {$\GlUnProjL$} (G);
      \path[->,bend right = 60] (DF/C) edge node [right] {$\GlUnProjR$} (G);
      \node [between = G and DF/C, xshift = .4cm]  {$\MathSmall{\dashv}$};
      \node [between = G and DF/C, xshift = -.4cm] {$\MathSmall{\dashv}$};
    \end{tikzpicture}
    \split
    \begin{tikzpicture}[diagram,baseline=(C.base)]
      \node (G) {$\CptGlCat$};
      \node (C) [below = of G] {$\CCat$};
      \path[fibration] (G) edge node [upright desc] {$\CptGlProj$} (C);
      \path[->,bend left = 65] (C) edge node [left] {$\CptGlUnProjL$} (G);
      \path[->,bend right = 65] (C) edge node [right] {$\CptGlUnProjR$} (G);
      \node [between = G and C, xshift = .5cm]  {$\MathSmall{\dashv}$};
      \node [between = G and C, xshift = -.5cm] {$\MathSmall{\dashv}$};
    \end{tikzpicture}
    \split
    \begin{tikzpicture}[diagram,baseline=(l/sw.base)]
      \SpliceDiagramSquare<l/>{
        ne = \CCat,
        se = \DF{\CCat},
        nw = \CptGlCat,
        sw = \GlCat,
        east = \DFYo[\CCat],
        west = \CptIncl,
        west/style = embedding,
        east/style = embedding,
        south = \GlUnProjL,
        north = \CptGlUnProjL,
        south/style = {<-},
        north/style = {<-},
        east/node/style = upright desc,
      }
      \SpliceDiagramSquare<r/>{
        glue = west,
        glue target = l/,
        ne = \CptGlCat,
        se = \GlCat,
        east = \CptIncl,
        east/style = embedding,
        south = \GlUnProjR,
        north = \CptGlUnProjR,
      }
      \node [between = l/nw and l/se] {$\cong$};
      \node [between = l/ne and r/se] {$\cong$};
    \end{tikzpicture}
  \end{mathflow}
\end{lem}

\begin{proof}
  From the perspective of the left and right adjoints as sections of the
  fibration, it is particularly simple to explain their behavior: the left
  adjoint takes a discrete fibration $X$ to the \emph{initial} object of the
  fiber category $\GlProj[X]$, and the right adjoint takes $X$ to the
  \emph{terminal} object of the fiber category $\GlProj[X]$. Considering that
  $\GlProj[X]$ is just the slice ${\cSET}/{\F{X}}$, we may compute the families as follows:
  \begin{align*}
    \GlUnProjL{X} &=
    \prn{\Mor[!_{\F{X}}]{\emptyset_{\cSET}}{\F{X}}}\\
    \GlUnProjR{X} &=
    \prn{\Mor[\ArrId{\F{X}}]{\F{X}}{\F{X}}}
  \end{align*}

  In fact, this characterization already describes the left and right adjoints
  to $\CptGlProj$, considering that the fibers $\CptGlProj[C] =
  \cSET/\DimInclNv{C}$ are equivalent to $\GlProj[\DFYo[\CCat]{C}] =
  \cSET/\F{\DFYo[\CCat]{C}}$.
\end{proof}

\begin{lem}\label{lem:gl-proj-repr-map-fun}
  The gluing fibration $\FibMor[\GlProj]{\GlCat}{\DF{\CCat}}$ is a representable map functor.
\end{lem}

\begin{proof}
  $\GlProj$ is a logical morphism, preserving in particular finite limits and
  all pushforwards. Therefore, it remains to check that it takes representable
  maps to representable maps. We fix a compact family $\Mor[p]{Y}{X}$ to check
  that the map $\GlProj{\Mor[p]{Y}{X}}$ is representable. In this case, it will
  be simplest to use the Grothendieck-style characterization of representable
  maps in $\DF{\CCat}$ in terms of change of base to a representable object.
  Fixing $C:\CCat$ and a generalized element
  $\Mor[x]{\DFYo[\CCat]{C}}{\GlProj{X}}$, we must check that the fiber product
  $\GlProj{Y}\times_{\GlProj{X}}\DFYo[\CCat]{C}$ is representable:
  \begin{equation}\label{diag:gl-proj-repr-map-fun:0}
    \DiagramSquare{
      width = 2.5cm,
      nw/style = pullback,
      ne = \GlProj{Y},
      east = \GlProj{p},
      se = \GlProj{X},
      south = x,
      sw = \DFYo[\CCat]{C},
      nw = {\GlProj{Y}\times_{\GlProj{X}}\DFYo[\CCat]{C}}
    }
  \end{equation}

  By transposing along the adjunction $\GlUnProjL\dashv
  \GlProj$, we have a map
  $\Mor[\tilde{x}]{\GlUnProjL{\DFYo[\CCat]{C}}}{X}:\GlCat$;
  noting that $\GlUnProjL{\DFYo[\CCat]{C}}\cong\CptIncl{\CptGlUnProjL{C}}$, we
  therefore have the following cartesian square in $\GlCat$ using the
  compactness of $p$:
  \begin{equation}\label{diag:gl-proj-repr-map-fun:1}
    \DiagramSquare{
      width = 2.5cm,
      nw/style = pullback,
      ne = Y,
      east = p,
      se = X,
      south = \tilde{x},
      sw = \CptIncl{\CptGlUnProjL{C}},
      nw = \CptIncl{Y_{\tilde{x}}}
    }
  \end{equation}

  The image of Diagram~\ref{diag:gl-proj-repr-map-fun:1} under
  $\FibMor[\GlProj]{\GlCat}{\DF{\CCat}}$ has the same cospan as
  Diagram~\ref{diag:gl-proj-repr-map-fun:0}; but $\GlProj$ preserves pullbacks
  and $\CptIncl{Y_{\tilde{x}}}$ must lie over a representable object, so we are
  finished.
\end{proof}

\begin{lem}\label{lem:gl-proj-reflects-repr}
  The gluing fibration $\FibMor[\GlProj]{\GlCat}{\DF{\CCat}}$ reflects representable maps.
\end{lem}

\begin{proof}
  Fixing a family $\Mor[p]{Y}{X}:\GlCat$ such that $\GlProj{\Mor[p]{Y}{X}}$ is
  representable, we must check that $p$ is a compact family; in other words,
  fixing $\Mor[x]{\CptIncl{\Gamma}}{X}$, we must check that the fiber product
  $Y\times_{X}\CptIncl{\Gamma}$ below is compact:
  \begin{equation}\label{diag:gl-proj-reflects-repr-maps:0}
    \DiagramSquare{
      nw/style = pullback,
      nw = Y\times_{X}\CptIncl{\Gamma},
      ne = Y,
      se = X,
      east = p,
      sw = \CptIncl{\Gamma},
      south = x
    }
  \end{equation}

  It suffices to check that $\GlProj{Y\times_{X}\CptIncl{\Gamma}}$ is
  representable; because $\GlProj{p}$ is representable,
  Diagram~\ref{diag:gl-proj-reflects-repr-maps:0} lies over
  the square below:
  \begin{equation*}
    \DiagramSquare{
      width = 2.5cm,
      nw/style = pullback,
      nw = \DFYo[\CCat]{\GlProj{Y}_{\GlProj{x}}},
      ne = \GlProj{Y},
      se = \GlProj{X},
      east = \GlProj{p},
      sw = \DFYo[\CCat]{\CptGlProj{\Gamma}},
      south = \GlProj{x}
    }
    \qedhere
  \end{equation*}

\end{proof}

\subsection{Nerve and realization}%
\label{sec:nerve}

A computability model of XTT is given by a representable map functor
$\Mor[\Mod{\CptGlCat}]{\ThCat}{\DF{\CptGlCat}}$; in particular, we must
construct a natural model in $\DF{\CptGlCat}$ which is closed under all the
connectives of XTT and its universe of Bishop sets. These objects are, however,
particularly difficult to construct from the perspective of discrete
fibrations or presheaves; in previous work, it has accordingly been necessary to construct
the constitutents of the computability model at the level of
sets~\cite{sterling-angiuli-gratzer:2019,coquand:2019,kaposi-huber-sattler:2019},
manually quantifying over computable contexts and computable substitutions.

In recent work, Sterling and Angiuli have shown that it is simpler to construct
these objects \emph{internally} to the logos $\GlCat$ of general computability
structures, and then transfer them in a single motion to $\DF{\CptGlCat}$. This
is accomplished by means of a \emph{nerve} functor
$\Mor[\GlNv]{\GlCat}{\DF{\CptGlCat}}$ which, by virtue of the density of
$\EmbMor[\CptIncl]{\CptGlCat}{\GlCat}$ (Theorem~\ref{thm:cpt-incl-dense}), is
not only fully faithful but also locally cartesian closed. Crucially, $\GlNv$
will also turn out to be a representable map functor.

\begin{con}[Nerve]\label{con:gl-nv}
  Let $X : \GlCat$ be a general computability structure; we may define a discrete fibration $\GlNv{X}
  : \DF{\CptGlCat}$ whose fiber at a compact computability structure
  $\Gamma:\CptGlCat$ is the hom set $\Hom[\GlCat]{\CptIncl{\Gamma}}{X}$. This
  assignment extends to a functor $\Mor[\GlNv]{\GlCat}{\DF{\CptGlCat}}$, which
  may be viewed either as a restriction of the Yoneda embedding of $\GlCat$, or a left Kan extension
  of the Yoneda embedding of $\CptGlCat$:
  \begin{equation*}
    \begin{tikzpicture}[diagram,baseline=(DF/Gk.base)]
      \node (G) {$\GlCat$};
      \node (DF/G) [right = of G] {$\DF{\GlCat}$};
      \node (DF/Gk) [below = of DF/G] {$\DF{\CptGlCat}$};
      \path[embedding] (G) edge node [above] {$\DFYo[\GlCat]$} (DF/G);
      \path[->,exists] (G) edge node [sloped,below] {$\GlNv$} (DF/Gk);
      \path[->] (DF/G) edge node [right] {$\CptIncl^*$} (DF/Gk);
    \end{tikzpicture}
    \qquad
    \begin{tikzpicture}[diagram,baseline=(G.base)]
      \node (Gk) {$\CptGlCat$};
      \node (DF/Gk) [right = of Gk] {$\DF{\CptGlCat}$};
      \node (G) [below = of Gk] {$\GlCat$};
      \path[embedding] (Gk) edge node [left] {$\CptIncl$} (G);
      \path[embedding] (Gk) edge node [above] {$\DFYo[\CptGlCat]$} (DF/Gk);
      \path[->,exists] (G) edge node [sloped,below] {$\GlNv$} (DF/Gk);
      \node (sw) [below = of DF/Gk] {};
      \node [between = Gk and sw, xshift = -.3cm,yshift = .3cm] {$\Downarrow$};
    \end{tikzpicture}
    \qedhere
  \end{equation*}
\end{con}

\begin{lem}[Realization]\label{lem:gl-rl}
  The nerve functor has a left adjoint $\Mor[\GlRl]{\DF{\CptGlCat}}{\GlCat}$,
  the \emph{realization} of a discrete fibration on compact computability
  structures; consequently, $\GlNv$ preserves small limits.
\end{lem}

\begin{proof}
  The realization functor is obtained by left Kan extension:
  \begin{equation*}
    \begin{tikzpicture}[diagram,baseline = (sw.base)]
      \node (Gk) {$\CptGlCat$};
      \node (DF/Gk) [below = of Gk] {$\DF{\CptGlCat}$};
      \node (G) [right = of Gk] {$\GlCat$};
      \path[embedding] (Gk) edge node [above] {$\CptIncl$} (G);
      \path[embedding] (Gk) edge node [left] {$\DFYo[\CptGlCat]$} (DF/Gk);
      \path[->,exists] (DF/Gk) edge node [sloped,below] {$\GlRl$} (G);
      \node (sw) [below = of G] {};
      \node [between = Gk and sw, xshift = -.3cm,yshift = .3cm] {$\Downarrow$};
    \end{tikzpicture}
    \qedhere
  \end{equation*}
\end{proof}

The realization of a specific discrete fibration may be computed as a coend,
using the general formula for pointwise Kan extensions. Letting $F :
\DF{\CptGlCat}$, we calculate:
\begin{align*}
  \GlRl{F} &\cong \prn{\Lan{\DFYo[\CptGlCat]}{\CptIncl}}\prn{F}\\
  &\cong
  \Coend{\Gamma : \CptGlCat}{
    \Hom[\DF{\CptGlCat}]{\DFYo[\CptGlCat]{\Gamma}}{F} \cdot \CptIncl{\Gamma}
  }
  \\
  &\cong
  \Coend{\Gamma:\CptGlCat}{
    F_{\Gamma} \cdot \CptIncl{\Gamma}
  }
\end{align*}

Writing $\Mor[F_\bullet]{\OpCat{\CptGlCat}}{\SET}$ for the presheaf
corresponding to the discrete fibration $F$, we may package the computation of
$F$'s realization in terms of the tensor calculus of functors:
\begin{equation}
  \GlRl{F} \cong F_\bullet \otimes_{\CptGlCat} \CptIncl
\end{equation}

From Lemma~\ref{lem:gl-nv-ff} below, we can see that the nerve exhibits a true
``Yoneda embedding'' of general computability structures into discrete
fibrations.

\begin{lem}\label{lem:gl-nv-ff}
  The nerve functor $\Mor[\GlNv]{\GlCat}{\DF{\CptGlCat}}$ is fully faithful.
\end{lem}

\begin{proof}
  This is equivalent to the density of the inclusion $\CptIncl$.
\end{proof}

\begin{lem}[Equivalent subcategories]\label{lem:gl-nv-equiv}
  The nerve--realization adjunction $\GlRl\dashv \GlNv$ restricts to an equivalence (of categories) 
  between the subcategories of generating objects (compact computability
  structures and representable presheaves respectively).
\end{lem}

\begin{proof}
  We first check that the nerve of any compact computability structure $X \cong
  \CptIncl{\Gamma}$ is a representable discrete fibration. It suffices to
  compute in terms of the corresponding presheaves,
  $\GlNv{\CptIncl{\Gamma}}_{\bullet}
    \cong \Hom[\GlCat]{\CptIncl{\bullet}}{\CptIncl{\Gamma}}
    \cong \Hom[\CptGlCat]{\bullet}{\Gamma}
    \cong \Yo[\CptGlCat]{\Gamma}
    \cong \prn{\DFYo[\CptGlCat]{\Gamma}}_\bullet
  $.
  Therefore, we have $\GlNv{\CptIncl{\Gamma}} \cong \DFYo[\CptGlCat]{\Gamma}$.
  Next, we check that the realization of any representable discrete fibration
  $F \cong \DFYo[\CptGlCat]{\Gamma}$ is compact; by the above, we have $\GlRl{\DFYo[\CptGlCat]{\Gamma}} \cong
  \GlRl{\GlNv{\CptIncl{\Gamma}}}$; because the nerve is fully faithful
  (Lemma~\ref{lem:gl-nv-ff}), the counit to the adjunction $\GlRl\dashv\GlNv$
  is an isomorphism, so we further have $\GlRl{\DFYo[\CptGlCat]{\Gamma}} \cong
  \GlRl{\GlNv{\CptIncl{\Gamma}}} \cong \CptIncl{\Gamma}$.

  We have shown that the nerve--realization adjunction restricts to a pair of 
  functors between subcategories. It is easy to see that this is in fact an
  equivalence, since we have shown that $\GlRl{\GlNv{\CptIncl{\Gamma}}} \cong
  \CptIncl{\Gamma}$ and $\GlNv{\GlRl{\DFYo[\CptGlCat]{\Gamma}}} \cong \GlNv{\CptIncl{\Gamma}}
  \cong \DFYo[\CptGlCat]{\Gamma}$.
\end{proof}

To establish the behavior of the nerve on pushforwards, it will be useful to
choose a good dense subcategory of each slice $\Sl*{\DF{\CptGlCat}}{F}$.

\begin{rem}\label{rem:df-sl-generators}
  First we recall that in presheaves, each slice
  $\Sl{\Psh{\CptGlCat}}{F_\bullet}$ may be reconstructed equivalently as the
  category of presheaves $\Psh{\DelimMin{1}\int\!F_\bullet} = \Psh{F}$ on the
  total category of $F$; therefore, by the Yoneda lemma, each slice
  $\Sl*{\DF{\CptGlCat}}{F}$ is densely generated by the functor
  $\Mor{F}{\Sl*{\DF{\CptGlCat}}{F}}$ which sends every element $x \in F_\Gamma$
  to the corresponding map $\Mor[x]{\DFYo[\CptGlCat]{\Gamma}}{F}$.

  Furthermore, by Lemma~\ref{lem:gl-nv-equiv} each generating object
  $\Mor{\DFYo[\CptGlCat]{\Gamma}}{F}$ may be written equivalently as
  $\Mor{\GlNv{\CptIncl{\Gamma}}}{F}$. If $F = \GlNv{X}$, we may observe that
  the slice $\Sl*{\DF{\CptGlCat}}{F}$ is in fact densely generated by the comma
  category $\Comma{\CptIncl}{\brc{X}}$ under the functor which sends each
  $\Mor[x]{\CptIncl{\Gamma}}{X}$ to $\GlNv{x} : \Sl*{\DF{\CptGlCat}}{F}$, a
  direct consequence of the fully faithfulness of the nerve
  (Lemma~\ref{lem:gl-nv-ff}).
\end{rem}

\begin{lemC}[{\cite[Lemma~4.2]{sterling-angiuli:2020}}]\label{lem:nv-preserves-pushforward}
  The nerve functor $\Mor[\GlNv]{\GlCat}{\DF{\CptGlCat}}$  preserves all pushforwards $f_* :
  \Mor{\Sl{\GlCat}{X}}{\Sl{\GlCat}{Y}}$ for $\Mor[f]{X}{Y} : \GlCat$.
\end{lemC}

\begin{proof}
  Letting $g : \Sl{\GlCat}{X}$, we intend to check that $\GlNv\prn{f_*g}$ is the
  pushforward $\prn{\GlNv{f}_*}{\GlNv{g}}$.  By Remark~\ref{rem:df-sl-generators}, it
  suffices to check the universal property at the generators
  $\Mor[\GlNv{x}]{\GlNv{\CptIncl{\Gamma}}}{\GlNv{Y}}$ of the slice
  $\Sl*{\DF{\CptGlCat}}{\GlNv{Y}}$.
  \begin{align*}
    &\Hom{\GlNv{x}}{\prn{\GlNv{f}}_*\GlNv{g}}
    \\
    &\quad\cong
    \Hom{{\GlNv{f}}^*\GlNv{x}}{\GlNv{g}}
    \\
    &\quad\cong
    \Hom{\GlNv{f^*x}}{\GlNv{g}}
    \\
    &\quad\cong
    \Hom{f^*x}{g}
    \\
    &\quad\cong
    \Hom{x}{f_*g}
    \\
    &\quad\cong
    \Hom{\GlNv{x}}{\GlNv\prn{f_*g}}
    \qedhere
  \end{align*}
\end{proof}

\begin{cor}\label{cor:gl-nv-lcc}
  The nerve functor $\Mor[\GlNv]{\GlCat}{\DF{\CptGlCat}}$ is locally cartesian closed.
\end{cor}

\begin{lem}\label{lem:gl-nv-repr-maps}
  The nerve functor  $\Mor[\GlNv]{\GlCat}{\DF{\CptGlCat}}$  preserves representable maps.
\end{lem}

\begin{proof}
  We fix a representable map $\Mor[f]{Y}{X}:\GlCat$, to check that its nerve
  $\Mor[\GlNv{f}]{\GlNv{Y}}{\GlNv{X}} : \DF{\CptGlCat}$ is representable; it
  will be simplest to check this condition formulated in the classical
  Grothendieck-style, considering fiber products with representable objects:
  \begin{equation}\label{diag:gl-nv-repr-maps:0}
    \DiagramSquare{
      nw/style = pullback,
      width = 2.8cm,
      nw = {\GlNv{Y}\times_{\GlNv{X}}\DFYo[\CptGlCat]{\Gamma}},
      ne = \GlNv{Y},
      se = \GlNv{X},
      east = \GlNv{f},
      sw = \DFYo[\CptGlCat]{\Gamma},
      south = x,
    }
  \end{equation}

  First of all, we may replace $\DFYo[\CptGlCat]{\Gamma}$ with the isomorphic
  $\GlNv{\CptIncl{\Gamma}}$; since $\GlNv$ is fully faithful and left exact, the entire square lies in the image of the nerve:
  \begin{equation}\label{diag:gl-nv-repr-maps:1}
    \DiagramSquare{
      width = 2.75cm,
      nw/style = pullback,
      ne = \GlNv{Y},
      se = \GlNv{X},
      east = \GlNv{f},
      sw = \GlNv{\CptIncl{\Gamma}},
      nw = \GlNv{Y\times_{X}\CptIncl{\Gamma}},
      south = \GlNv{\ulcorner x\urcorner},
    }
    \qquad
    \DiagramSquare{
      width = 2.5cm,
      nw/style = pullback,
      ne = {Y},
      se = {X},
      east = {f},
      sw = {\CptIncl{\Gamma}},
      nw = {Y\times_{X}\CptIncl{\Gamma}},
      south = {\ulcorner x\urcorner},
    }
  \end{equation}

  Therefore, it suffices to check that the fiber product
  $Y\times_{X}\CptIncl{\Gamma}$ is taken by the nerve to a representable
  object; by Lemma~\ref{lem:gl-nv-equiv}, this follows from the compactness of
  $Y\times_{X}\CptIncl{\Gamma}$ by virtue of the representability of $f$.
\end{proof}

\begin{cor}
  The nerve functor  $\Mor[\GlNv]{\GlCat}{\DF{\CptGlCat}}$  is a representable map functor.
\end{cor}

\begin{proof}
  By Corollary~\ref{cor:gl-nv-lcc} and Lemma~\ref{lem:gl-nv-repr-maps}.
\end{proof}

\subsection{Universes in the gluing fibration}

As a technical matter, we will require suitable type theoretic universes in
$\GlCat$ that are sent to appropriate universes in $\DF{\CCat}$. While there
are a variety of ways to construct these by combining the existing
universes of $\cSET$ and $\DF{\CCat}$ (see Uemura~\cite{uemura:2017}), there is a
considerably simpler alternative that becomes available once we observe that
$\GlCat$ is (equivalent to) a presheaf logos.

\begin{lem}
  There is a small category $\DCat$ such that $\GlCat\simeq\Psh{\DCat}$.
\end{lem}

\begin{proof}
  This follows from the result of Carboni and
  Johnstone~\cite{carboni-johnstone:1995}, itself an explication and
  generalization of an exercise posed in SGA~4~\cite[Tome 1, Expos\'e iv,
  Exercise 9.5.10]{sga:4}. In particular, it suffices to observe that the
  gluing functor $\Mor[\F]{\DF{\CCat}}{\cSET}$ is the inverse image part of an
  \emph{essential} morphism between presheaf topoi and is hence continuous.
\end{proof}

Therefore the Hofmann--Streicher lifting of Grothendieck 
universes~\cite{hofmann-streicher:1997} from $\SET$ into presheaf logoi
applies, yielding a cumulative hiearchy of universes $\GlUElMor$ in $\GlCat$.
These universes can be seen to lie over a corresponding universe hierarchy
$\CUElMor$ in $\DF{\CCat}$, defined simply by $\GlProj{\GlUElMor}$. A
\emph{small map} is then defined to be one that arises by pullback from a given
universe.

\begin{defi}
  Given a universe $\CUElMor$, we say that two codes $\Mor[A,B]{X}{\CUTY}$ are
  \emph{isomorphic} when they have the same extensions, \ie the pullbacks
  $\Pb{\CUEl}{A},\Pb{\CUEl}{B}$ are isomorphic. In other words, the two codes are characteristic for the same family.
\end{defi}

Our chosen universes satisfy a strict \emph{realignment} principle that will
play an important role in our development.

\begin{lem}[Realignment]\label{lem:fibered-universe-realignment}
  Let $\Mor[\frk{A}]{X}{\GlUTY}:\GlCat$ be a code for a small computability
  family in the following configuration:
  \begin{equation}
    \DiagramSquare{
      nw = X, ne = \GlUTY,
      height = 1.5cm,
      west/style = lies over, east/style = lies over,
      north = \frk{A},
      south = A,
      sw = \GlProj{X},
      se = \CUTY
    }
  \end{equation}

  Let $\Mor[B]{\GlProj{X}}{\CUTY}$ be a code isomorphic to
  $\Mor[A]{\GlProj{X}}{\CUTY}$; then, we have a code
  $\Mor[\frk{A}_B]{X}{\GlUTY}$ lying strictly over $B$ which is isomorphic to
  $\frk{A}$.
\end{lem}

\begin{proof}
  This follows from the \emph{strict gluing} principle for topos-theoretic
  universes~\cite{awodey:2021:qms,orton-pitts:2016,bbcgsv:2016,shulman:2015:elegant,streicher:2014:simplicial,kapulkin-lumsdaine:2021},
  which is known to hold for Hofmann--Streicher universes. In the internal 
  language of $\GlCat$, there exists a distinguished subterminal object $\P :=
  \GlUnProjL{\ObjTerm{\DF{\CCat}}}$ with the property that
  $\Sl{\GlCat}{\P}\simeq \DF{\CCat}$ and under this identification, the gluing
  fibration $\GlProj$ is the pullback functor
  $\Mor[\P^*]{\GlCat}{\Sl{\GlCat}{\P}}$.  In the internal language of $\GlCat$,
  we may speak of $\CUTY$ as $\GlUnProjR\GlProj{\GlUTY} \cong \P_*\P^*\GlUTY
  \cong \IHom{\P}{\GlUTY}$.
  Therefore, if $\GlUTY$ satisfies the strict gluing axiom, we may internally
  realign any total element $\frk{A}:\GlUTY$ to agree strictly with any partial
  element $A:\IHom{\P}{\GlUTY}$ equipped with a partial isomorphism
  $\frk{A}\prn{z}\cong A$ under $z:\P$.
\end{proof}

\section{XTT in computability structures}

The essence of the canonicity argument for XTT lies in constructing a
representable map functor $\Mor[\GlFun]{\ThCat}{\GlCat}$ together with a
natural isomorphism $\Mor[\GlCell]{\Mod{\CCat}}{\GlProj\circ\GlFun}$ in the
sense of Diagram~\ref{diag:gl-fun-cell} below:
\begin{equation}\label{diag:gl-fun-cell}
  \begin{tikzpicture}[diagram,baseline = (DF/C.base)]
    \node (T) {$\ThCat$};
    \node (G) [right = 2.5cm of T] {$\GlCat$};
    \node (DF/C) [between = T and G, yshift = -2cm] {$\DF{\CCat}$};
    \path[->] (T) edge node [sloped,below] {$\Mod{\CCat}$} (DF/C);
    \path[->] (T) edge node [above] {$\GlFun$} (G);
    \path[fibration] (G) edge node [sloped,below] {$\GlProj$} (DF/C);
    \node [above = 1.25cm of DF/C] {$\GlCell$};
  \end{tikzpicture}
\end{equation}

In keeping with the previous section, $\GlFun$ is constructed over an arbitrary model
$\Mor[\Mod{\CCat}]{\ThCat}{\DF{\CCat}}$. Eventually, we shall specialize to the initial model of XTT
and use this to derive canonicity.

Recall that the specification of $\ThCat$ in Section~\ref{sec:functorial-semantics}
is derived from the signature for XTT in Uemura's logical framework (see Section~\ref{sec:xtt:lf}).
Therefore, to see that the canonical natural isomorphism $\GlCell$ exists, it will suffice to choose
suitable isomorphisms $\GlCell{J}$ for just the \emph{generating} objects
$J:\ThCat$, the components of the signature; in all cases, the isomorphism $\GlCell{J}$ will be
canonical (or even the identity), because we will always define $\GlFun{\Mor{I}{J}}$ to lie
essentially over $\Mod{\CCat}{\prn{\Mor{I}{J}}}$.

\subsection{Glued cubical structure}

\begin{con}[The interval]
  We must construct a computability structure $\GlFun{\DIM}$ lying over
  $\Mod{\CCat}{\DIM}$; at the level of cubical sets, we will use the ``generic
  dimension'' $\Mor[\GlFun{\DIM}]{\Yo[\CUBE]{[1]}}{\F{\Mod{\CCat}{\DIM}}}$
  determined essentially by the functorial action of $\EmbMor[\DimIncl]{\CUBE}{\CCat}$.
  Fixing an element $\Mor[r]{[m]}{[1]}$ of $\Yo[\CUBE]{[1]}$, we have by
  functoriality an element $\Mor[\DimIncl{r}]{\DimIncl{m}}{\DimIncl{1}}$ of
  $\F{\Mod{\CCat}{\DIM}}$.
  We likewise have appropriate endpoints as follows:
  \begin{equation}\label{diag:glued-endpoints}
    \DiagramSquare{
      width = 3.75cm,
      height = 1.75cm,
      nw = \ObjTerm{\cSET},
      ne = \Yo[\CUBE]{[1]},
      sw = \F{\Mod{\CCat}{\EMP}},
      se = \F{\Mod{\CCat}{\DIM}},
      north = {\Dim0,\Dim1},
      west = \cong,
      east = \GlFun{\DIM},
      south = {\F{\Mod{\CCat}{\Dim0}},\F{\Mod{\CCat}{\Dim1}}}
    }
  \end{equation}

  Diagram~\ref{diag:glued-endpoints} can be seen to commute, considering the
  definition of $\DimIncl$ as the finite product preserving functor
  corresponding to the interval-algebra in $\CCat$.
\end{con}

\begin{lem}\label{lem:f-dim-repr}
  The syntactic interval $\Mod{\CCat}{\DIM}$ is taken by the base change
  $\Mor[\F]{\DF{\CCat}}{\cSET}$ to the actual interval $\Yo[\CUBE]{\brk{1}}$ of
  $\cSET$.
\end{lem}

\begin{proof}
  This follows by computation, recalling that $\Mod{\CCat}{\DIM}$ is
  represented in $\CCat$ by an interval context $\Repr{\Mod{\CCat}{\DIM}} \cong
  \DimIncl{\brk{1}}$, and using the fact that $\EmbMor[\DimIncl]{\CUBE}{\CCat}$
  is fully faithful:
  \[
    \F{\Mod{\CCat}{\DIM}}
    =
    \Mod{\CCat}{\DIM}\prn{\DimIncl}
    \cong
    \Hom[\CCat]{\DimIncl}{\DimIncl{\brk{1}}}
    \cong
    \Hom[\CUBE]{-}{\brk{1}}
    =
    \Yo[\CUBE]{\brk{1}}
    \qedhere
  \]
\end{proof}

\begin{lem}[Tininess of the interval]\label{lem:tininess}
  The interval object $\GlFun{\DIM}$ is \emph{tiny} in the sense that the
  exponential functor $\IHom{-}{\GlFun{\DIM}}$ has a further right adjoint.
\end{lem}

\begin{proof}
  $\GlCat$ is the Artin gluing the inverse image part of a morphism of topoi;
  under these circumstances, the universality of colimits ensures that it is
  enough for the restrictions of $\GlFun{\DIM}$ to both $\DF{\CCat}$ and
  $\cSET$ to be tiny (see~\cite[Lemma~32]{sterling-angiuli:2021} for a more
  detailed argument). These restrictions happen to be representable
  (Lemma~\ref{lem:f-dim-repr}), so tininess follows from the fact that both
  $\CCat$ and $\CUBE$ have finite products.
\end{proof}

\begin{con}[The face formula classifier]

  We recall that the gluing category $\GlCat$ is a logos~\cite{sga:4}, and
  moreover, the gluing fibration $\FibMor[\GlProj]{\GlCat}{\DF{\CCat}}$ is a
  \emph{logical morphism} and therefore preserves the subobject classifier and
  its first-order logic~\cite{johnstone:2002}. Therefore, we may obtain a glued
  face formula classifier in a conceptual way.

  Because $\Prop[\DF{\CCat}]:\DF{\CCat}$ classifies monomorphisms, we obtain a unique
  cartesian classifying square in $\DF{\CCat}$ for the generic face formula of
  $\Mod{\CCat}$:
  \begin{equation}\label{diag:cof-glue:0}
    \DiagramSquare{
      nw/style = pullback,
      east/style = {>->},
      west/style = {>->},
      nw = \Mod{\CCat}{\EMP},
      sw = \Mod{\CCat}{\COND},
      ne = \ObjTerm{\DF{\CCat}},
      se = \Prop[\DF{\CCat}],
      north = \cong,
      east = \PropTrue[\DF{\CCat}],
      south/style = exists,
      west = \Mod{\CCat}{\Prf},
      south = \DecodeCond
    }
  \end{equation}

  Therefore, we may construct a suitable glued face formula classifier by taking a
  cartesian lift in the gluing fibration of the subobject classifier
  $\Prop[\GlCat]$ along $\DecodeCond$ from
  Diagram~\ref{diag:cof-glue:0}, considering that $\Prop[\GlCat]$ lies over
  $\Prop[\DF{\CCat}]$:
  \begin{equation}\label{diag:cof-glue:1}
    \begin{tikzpicture}[diagram,baseline=(sw.base)]
      \SpliceDiagramSquare{
        west/style = lies over,
        east/style = lies over,
        width = 2.5cm,
        height = 1.5cm,
        sw = \Mod{\CCat}{\COND},
        se = \Prop[\DF{\CCat}],
        ne = \Prop[\GlCat],
        nw = \GlFun{\COND},
        north = \DecodeCond,
        south = \DecodeCond,
        north/style = exists
      }
      \node (tot) [right = of ne] {$\GlCat$};
      \node (base) [right = of se] {$\DF{\CCat}$};
      \path[fibration] (tot) edge (base);
    \end{tikzpicture}
  \end{equation}

  The generic face formula $\GlFun{\Prf}$ is then obtained by pullback in $\GlCat$:
  \begin{equation}\label{diag:cof-glue:2}
    \DiagramSquare{
      nw/style = pullback,
      nw = \GlFun{\EMP},
      ne = \ObjTerm{\GlCat},
      sw = \GlFun{\COND},
      se = \Prop[\GlCat],
      east = \PropTrue[\GlCat],
      west = \GlFun{\Prf},
      south = \DecodeCond,
      north/style = exists,
      east/style = {>->},
      west/style = {>->,exists},
    }
  \end{equation}

  Because $\GlProj$ preserves finite limits and pullback cones are unique up to
  unique isomorphism, we see that $\GlFun{\Prf}$ lies over
  $\Mod{\CCat}{\Prf}$ as required.
\end{con}

We may prove a Beck-Chevalley lemma for dimension equality in $\DF{\CCat}$:

\begin{lem}[Beck-Chevalley]\label{lem:bc-eq}
  The following diagram commutes in $\DF{\CCat}$.
  \begin{equation*}
    \begin{tikzpicture}[diagram,baseline = (W.base)]
      \node (I2) {$\Mod{\CCat}{\DIM}^2$};
      \node (F) [below = of I2] {$\Mod{\CCat}{\COND}$};
      \node (W) [right = of F] {$\Prop[\DF{\CCat}]$};
      \path[->] (I2) edge node [left] {$\Mod{\CCat}{\prn{=}}$} (F);
      \path[->] (F) edge node [below] {$\DecodeCond$} (W);
      \path[->] (I2) edge node [above,sloped] {$\prn{=}$} (W);
    \end{tikzpicture}
  \end{equation*}
\end{lem}

\begin{proof}
  Recalling the diagram from
  Specification~\ref{spec:dim-eq}, we observe that both maps are
  characteristic of the same subobject, and thence equal.
\end{proof}

\begin{con}[Glued dimension equality]\label{con:glued-equality}
  Dimension equality is lifted from $\Mod{\CCat}$ to the gluing category by the
  universal property of the cartesian lift below, using
  the Beck-Chevalley triangle of Lemma~\ref{lem:bc-eq} and the fact
  that $\GlProj$ is a logical morphism.
  \begin{equation}\label{diag:glued-eq:0}
    \begin{tikzpicture}[diagram,baseline = (sw.base)]
      \SpliceDiagramSquare{
        west/style = lies over,
        east/style = lies over,
        width = 3cm,
        height = 1.5cm,
        nw = \GlFun{\COND},
        ne = \Prop[\GlCat],
        sw = \Mod{\CCat}{\COND},
        se = \Prop[\DF{\CCat}],
        north/node/style = upright desc,
        north = \DecodeCond,
        south = \DecodeCond
      }

      \node (G/Dim2) [above left = of nw] {$\GlFun{\DIM}^2$};
      \node (C/Dim2) [below = 1.5cm of G/Dim2] {$\Mod{\CCat}{\DIM}^2$};

      \path[->] (C/Dim2) edge node [below,sloped] {$\Mod{\CCat}{\prn{=}}$} (sw);
      \path[lies over] (G/Dim2) edge (C/Dim2);
      \path[->,exists] (G/Dim2) edge node [desc] {$\MathSmall{\GlFun{=}}$} (nw);
      \path[->,bend left = 30] (G/Dim2) edge node [above,sloped] {$\prn{=}$} (ne);

      \node (tot) [right = of ne] {$\GlCat$};
      \node (base) [right = of se] {$\DF{\CCat}$};
      \path[fibration] (tot) edge (base);
    \end{tikzpicture}
  \end{equation}

  We must check that the square below from
  Specification~\ref{spec:dim-eq} is cartesian:
  \begin{equation}\label{diag:glued-eq:1}
    \DiagramSquare{
      nw/style = dotted pullback,
      nw = \GlFun{\DIM},
      sw = \GlFun{\DIM}^2,
      west = \delta,
      south = \GlFun{=},
      ne = \GlFun{\EMP},
      se = \GlFun{\COND},
      east = \GlFun{\Prf}
    }
  \end{equation}

  By the pullback lemma, it would suffice to check that the outer
  square below is cartesian.
  \begin{equation}\label{diag:glued-eq:2}
    \begin{tikzpicture}[diagram,baseline = (l/sw.base)]
      \SpliceDiagramSquare<l/>{
        nw/style = dotted pullback,
        ne/style = pullback,
        nw = \GlFun{\DIM},
        sw = \GlFun{\DIM}^2,
        west = \delta,
        south = \GlFun{=},
        ne = \GlFun{\EMP},
        se = \GlFun{\COND},
        east = \GlFun{\Prf},
        east/node/style = upright desc,
        east/style = >->
      }
      \SpliceDiagramSquare<r/>{
        glue = west,
        glue target = l/,
        ne = \ObjTerm{\GlCat},
        se = \Prop[\GlCat],
        south = \DecodeCond,
        east = \PropTrue[\GlCat],
        east/style = >->
      }
    \end{tikzpicture}
  \end{equation}

  Using the upstairs triangle of Diagram~\ref{diag:glued-eq:0}, it
  suffices to observe that the following classification square is cartesian:
  \begin{equation*}
    \DiagramSquare{
      nw/style = pullback,
      nw = \GlFun{\DIM},
      sw = \GlFun{\DIM}^2,
      west = \delta,
      ne = \ObjTerm{\GlCat},
      se = \Prop[\GlCat],
      east = \PropTrue[\GlCat],
      south = \prn{=},
      east/style = >->,
      west/style = >->,
    }
    \qedhere
  \end{equation*}
\end{con}

We do not expect a Beck-Chevalley lemma for disjunction analogous to
Lemma~\ref{lem:bc-eq}, since $\phi\lor\psi$ is not (and cannot be) the
``true'' disjunction of $\DF{\CCat}$: instead, we imposed orthogonality
conditions in Specification~\ref{spec:disjunction} to ensure that certain
judgments of XTT (typehood, typing, and formula satisfaction) treat
$\phi\lor\psi$ as if it were a disjunction.

\begin{con}[Glued disjunction]
  We may test a pair of glued face conditions for truth of disjunction as follows:
  \begin{equation}\label{diag:glued-disj:0}
    \begin{tikzpicture}[diagram,baseline = (l/sw.base)]
      \SpliceDiagramSquare<l/>{
        west/style = lies over,
        east/style = lies over,
        width = 2.5cm,
        height = 1.5cm,
        nw = \GlFun{\COND}^2,
        sw = \Mod{\CCat}{\COND}^2,
        ne = \Prop[\GlCat]^2,
        se = \Prop[\DF{\CCat}]^2,
        north = \MathSmall{\DecodeCond^2},
        south = \MathSmall{\DecodeCond^2},
        south/node/style = upright desc,
        north/node/style = upright desc,
      }
      \SpliceDiagramSquare<r/>{
        glue = west,
        glue target = l/,
        width = 2.5cm,
        height = 1.5cm,
        east/style = lies over,
        ne = \Prop[\GlCat],
        se = \Prop[\DF{\CCat}],
        south = \MathSmall{\lor},
        north = \MathSmall{\lor},
        south/node/style = upright desc,
        north/node/style = upright desc,
      }
      \path[->,bend left=30] (l/nw) edge node [above] {$\Rst{\GlFun{\COND}}{\lor}$} (r/ne);
      \path[->,bend right=30] (l/sw) edge node [below] {$\Rst{\Mod{\CCat}{\COND}}{\lor}$} (r/se);
      \node (tot) [right = of r/ne] {$\GlCat$};
      \node (base) [right = of r/se] {$\DF{\CCat}$};
      \path[fibration] (tot) edge (base);
    \end{tikzpicture}
  \end{equation}

  Unfortunately, the subobject $\Pb{\PropTrue}{\Rst{\Mod{\CCat}{\COND}}{\lor}}\defeq
  \brc{\phi,\psi \mid
  \Mod{\CCat}{\Prf}\brk{\phi}\lor\Mod{\CCat}{\Prf}\brk{\psi}}$ corresponding to
  the downstairs map of Diagram~\ref{diag:glued-disj:0} is \emph{not}
  classified by $\Mod{\CCat}{\COND}$! This is because such a subobject must have
  representable fibers, and but the real disjunction is a colimit in
  $\DF{\CCat}$ and therefore not representable. We need something that lies
  instead over the following pullback:
  \begin{equation}\label{diag:glued-disj:1}
    \DiagramSquare{
      nw/style = pullback,
      west/style = >->,
      east/style = >->,
      nw = \Mod{\CCat}{\prn{\Pb{\Prf}{\lor}}},
      sw = \Mod{\CCat}{\COND}^2,
      se = \Mod{\CCat}{\COND},
      south = \Mod{\CCat}{\lor},
      east = \Mod{\CCat}{\Prf},
      ne = \Mod{\CCat}{\EMP},
    }
  \end{equation}

  Because the ``ideal'' disjunction $\Pb{\PropTrue}{\Rst{\Mod{\CCat}{\COND}}{\lor}}$ is more
  universal than the disjunction of $\Mod{\CCat}{\COND}$, we obtain a unique map
  $\Mor|>->|[i]{\Pb{\PropTrue}{\Rst{\Mod{\CCat}{\COND}}{\lor}}}{\Mod{\CCat}{\prn{\Pb{\Prf}{\lor}}}}$;
  taking an opcartesian lift, we may shift $\Pb{\PropTrue}{\Rst{\GlFun{\COND}}{\lor}}$ (the
  subobject corresponding to the upstairs map of
  Diagram~\ref{diag:glued-disj:0}) to lie over
  $\Mod{\CCat}{\prn{\Pb{\Prf}{\lor}}}$:
  \begin{equation}\label{diag:glued-disj:2}
    \DiagramSquare{
      west/style = op lies over,
      east/style = op lies over,
      north/style = {exists, >->},
      south/style = >->,
      width = 3cm,
      height = 1.5cm,
      nw = \Pb{\PropTrue}{\Rst{\GlFun{\COND}}{\lor}},
      sw = \Pb{\PropTrue}{\Rst{\Mod{\CCat}{\COND}}{\lor}},
      se = \Mod{\CCat}{\prn{\Pb{\Prf}{\lor}}},
      south = i,
      ne = i_!\Pb{\PropTrue}{\Rst{\GlFun{\COND}}{\lor}},
    }
  \end{equation}

  This lift can be seen to be a subobject of $\GlFun{\COND}^2$ using the
  universal property of the opcartesian lift:
  \begin{equation}\label{diag:glued-disj:3}
    \begin{tikzpicture}[diagram,baseline = (sw.base)]
      \SpliceDiagramSquare{
        west/style = op lies over,
        east/style = op lies over,
        north/style = >->,
        south/style = >->,
        width = 3cm,
        height = 1.5cm,
        nw = \Pb{\PropTrue}{\Rst{\GlFun{\COND}}{\lor}},
        sw = \Pb{\PropTrue}{\Rst{\Mod{\CCat}{\COND}}{\lor}},
        se = \Mod{\CCat}*{\Pb{\Prf}{\lor}},
        south = i,
        ne = i_!\Pb{\PropTrue}{\Rst{\GlFun{\COND}}{\lor}},
      }
      \node (nee) [right = 2.5cm of ne, yshift = 1.5cm] {$\GlFun{\COND}^2$};
      \node (see) [right = 2.5cm of se, yshift = 1.5cm] {$\Mod{\CCat}{\COND}^2$};
      \path[op lies over] (nee) edge (see);
      \path[>->] (se) edge (see);
      \path[>->,exists,] (ne) edge (nee);
      \path[>->,bend left=30] (nw) edge (nee);
    \end{tikzpicture}
  \end{equation}

  Consider the characteristic map of the dotted monomorphism from
  Diagram~\ref{diag:glued-disj:3}:
  \begin{equation}\label{diag:glued-disj:4}
    \DiagramSquare{
      nw/style = pullback,
      se = \Prop[\GlCat],
      sw = \GlFun{\COND}^2,
      west/style = >->,
      east/style = >->,
      ne = \ObjTerm{\GlCat},
      nw = i_!\Pb{\PropTrue}{\Rst{\GlFun{\COND}}{\lor}},
      width = 3cm,
      height = 1.5cm,
      south = \chi,
      south/style = exists,
    }
  \end{equation}

  Because $\GlProj$ is a logical functor and $\Prop[\DF{\CCat}]$ classifies
  subobjects strictly, Diagram~\ref{diag:glued-disj:4} must lie over the
  following square:
  \begin{equation}\label{diag:glued-disj:5}
    \begin{tikzpicture}[diagram,node distance=2cm,baseline = (sw.base)]
      \node [pullback] (nw) {$\Mod{\CCat}*{\Pb{\Prf}{\lor}}$};
      \node (sw) [below = 1.5cm of nw] {$\Mod{\CCat}{\COND}^2$};
      \node (se) [right = of sw] {$\Mod{\CCat}{\COND}$};
      \node (see) [right = of se] {$\Prop[\DF{\CCat}]$};
      \node (ne) [above = 1.5cm of see] {$\ObjTerm{\DF{\CCat}}$};
      \path[>->] (nw) edge (sw);
      \path[>->] (ne) edge (see);
      \path[->] (nw) edge (ne);
      \path[->] (sw) edge node [below] {$\Mod{\CCat}{\lor}$} (se);
      \path[->] (se) edge node [below] {$\DecodeCond$} (see);
    \end{tikzpicture}
  \end{equation}

  Therefore, we may use the universal property of the cartesian lift to obtain
  a code for disjunction of glued face conditions:
  \begin{equation*}
    \begin{tikzpicture}[diagram,baseline = (sw.base)]
      \SpliceDiagramSquare{
        west/style = lies over,
        east/style = lies over,
        width = 3cm,
        height = 1.5cm,
        nw = \GlFun{\COND},
        ne = \Prop[\GlCat],
        sw = \Mod{\CCat}{\COND},
        se = \Prop[\DF{\CCat}],
        south = \DecodeCond,
        north = \DecodeCond,
        north/node/style = upright desc,
      }
      \node (nww) [left = 2.5cm of nw, yshift = 1.5cm] {$\GlFun{\COND}^2$};
      \node (sww) [left = 2.5cm of sw, yshift = 1.5cm] {$\Mod{\CCat}{\COND}^2$};
      \path[lies over] (nww) edge (sww);
      \path[->,bend left = 30] (nww) edge node [sloped,above] {$\chi$} (ne);
      \path[->] (sww) edge node [sloped,below] {$\Mod{\CCat}{\lor}$} (sw);
      \path[->,exists] (nww) edge node [desc] {$\GlFun{\lor}$} (nw);
    \end{tikzpicture}
    \qedhere
  \end{equation*}
\end{con}

\begin{lem}[Disjunction elimination / truth]\label{lem:disj-elim-truth}
  The glued disjunction satisfies the internal orthogonality condition with
  respect to $\GlFun{\Prf}$ written in Specification~\ref{spec:disjunction}.
\end{lem}
\begin{proof}
  Fixing $\Mor[\prn{\phi,\psi}]{X}{\GlFun{\COND}^2}$, we must find a unique lift for
  the following square, lying over the corresponding unique lift in $\DF{\CCat}$:
  \begin{equation}\label{diag:disj-elim-truth:0}
    \begin{tikzpicture}[diagram,baseline=(sw.base)]
      \SpliceDiagramSquare{
        west/style = >->,
        east/style = >->,
        width = 3cm,
        nw = \prn{\phi,\psi}^*\Pb{\PropTrue}{\Rst{\GlFun{\COND}}{\lor}},
        sw = \prn{\phi,\psi}^*\GlFun{\Pb{\Prf}{\lor}},
        ne = \GlFun{\EMP},
        se = \GlFun{\COND},
      }
      \path[->,exists] (sw) edge (ne);
    \end{tikzpicture}
  \end{equation}

  Because the lift in $\DF{\CCat}$ is assumed to exist and is unique, it
  suffices to find a unique lift for the image of
  Diagram~\ref{diag:disj-elim-truth:0} under $\Mor[\GlEl]{\GlCat}{\cSET}$; but
  $\GlFun{\Pb{\Prf}{\lor}}$ is an opcartesian lift of
  $\Pb{\PropTrue}{\Rst{\GlFun{\COND}}{\lor}}$, so the left-hand map becomes an identity in
  $\cSET$.
\end{proof}

\subsection{Glued type structure}

We will now show the sense in which the semantic constructions of Uemura
summarized above suffice to develop the type structure of XTT in the gluing
fibration.

\begin{con}[Universe of glued types]\label{con:universe-of-glued-types}
  We will define a computability structure $\GlFun{\TY} : \GlCat$ lying over
  $\Mod{\CCat}{\TY}$ in the gluing fibration
  $\FibMor[\GlProj]{\GlCat}{\DF{\CCat}}$. Because we have assumed
  $\Mod{\CCat}{\prn{\ElMor}}$ is small for $\CUElMor$, we have a characteristic
  map:
  \begin{equation*}
    \DiagramSquare{
      nw/style = pullback,
      south/style = exists,
      north/style = exists,
      nw = \Mod{\CCat}{\EL},
      sw = \Mod{\CCat}{\TY},
      west = \Mod{\CCat}{\El},
      se = \CUTY,
      ne = \CUEL,
      east = \CUEl,
      south = \CodeOf{\Mod{\CCat}{\El}}
    }
  \end{equation*}

  We therefore obtain the base of a universe by cartesian lift:
  \begin{equation*}
    \begin{tikzpicture}[diagram]
      \SpliceDiagramSquare{
        west/style = lies over,
        east/style = lies over,
        width = 2.5cm,
        height = 1.5cm,
        sw = \Mod{\CCat}{\TY},
        se = \CUTY,
        south = \CodeOf{\Mod{\CCat}{\El}},
        ne = \GlUTY,
        nw = \GlFun{\TY},
        north/style = exists,
        north = \CodeOf{\GlFun{\El}},
        nw/style = dotted pullback,
      }
      \node (tot) [right = of ne] {$\GlCat$};
      \node (base) [right = of se] {$\DF{\CCat}$};
      \path[fibration] (tot) edge (base);
    \end{tikzpicture}
  \end{equation*}

  Then, the rest of the universe $\GlFun{\ElMor}$ is obtained by pullback:
  \begin{equation*}
    \DiagramSquare{
      nw/style = pullback,
      south = \CodeOf{\GlFun{\El}},
      ne = \GlUEL,
      east = \GlUEl,
      se = \GlUTY,
      sw = \GlFun{\TY},
      nw = \GlFun{\EL},
      west = \GlFun{\El},
      west/style = exists,
      north/style = exists,
    }
    \qedhere
  \end{equation*}
\end{con}

\begin{lem}[Disjunction elimination / elements]\label{lem:disj-elim-el}
  The glued disjunction satisfies the internal orthogonality condition with
  respect to $\GlFun{\El}$ written in Specification~\ref{spec:disjunction}.
\end{lem}
\begin{proof}
  The proof is identical to that of Lemma~\ref{lem:disj-elim-truth}.
\end{proof}

\begin{con}[Closure under dependent product]
  We must show that $\GlFun{\El}$ has a code for dependent products lying over the
  corresponding algebra $\Mor{\Mod{\CCat}{\PI{\El}}}{\Mod{\CCat}{\El}}$.
  First of all, we have a potential code in $\GlCat$ for the dependent product of
  $\GlFun{\El}$-families in $\GlUTY$, defined using functoriality and the closure
  of $\GlUElMor$ under dependent products, and the fact that $\GlProj$
  preserves dependent products:
  \begin{equation*}
    \begin{tikzpicture}[diagram,baseline = (l/sw.base)]
      \SpliceDiagramSquare<l/>{
        west/style = |->,
        east/style = |->,
        width = 2.5cm,
        height = 1.5cm,
        nw = \PI{\GlFun{\El}},
        ne = \PI{\GlUEl},
        north = \PI{\CodeOf{\GlFun{\El}}},
        sw = \PI{\Mod{\CCat}{\El}},
        se = \PI{\CUEl},
        south = \PI{\CodeOf{\Mod{\CCat}{\El}}},
      }
      \SpliceDiagramSquare<r/>{
        glue = west,
        glue target = l/,
        east/style = |->,
        height = 1.5cm,
        se = \CUEl,
        south = \Kwd{pi}_{\CUEl},
        north = \Kwd{pi}_{\GlUEl},
        ne = \GlUEl,
      }
      \node (tot) [right = of r/ne] {$\CartArr{\GlCat}$};
      \node (base) [right = of r/se] {$\CartArr{\DF{\CCat}}$};
      \path[->] (tot) edge (base);
    \end{tikzpicture}
  \end{equation*}

  By realignment (Lemma~\ref{lem:fibered-universe-realignment}), using the fact
  $\Kwd{pi}_{\CUTY}\circ\PI{\CodeOf{\Mod{\CCat}{\El}}}$ and
  $\CodeOf{\Mod{\CCat}{\El}}\circ\Mod{\CCat}{\TyPi}$ are (different) characteristic maps for
  the same family, we obtain a new code in $\GlUTY$ for the same family in the following
  configuration:
  \begin{equation*}
    \begin{tikzpicture}[diagram,baseline=(sw.base)]
      \SpliceDiagramSquare{
        west/style = lies over,
        east/style = lies over,
        nw = \PI{\GlFun{\TY}},
        sw = \PI{\Mod{\CCat}{\TY}},
        se = \CUTY,
        ne = \GlUTY,
        south = \CodeOf{\Mod{\CCat}{\El}}\circ \Mod{\CCat}{\TyPi},
        north/style = exists,
        width = 2.5cm,
        height = 1.5cm,
      }
      \node (tot) [right = of ne] {$\GlCat$};
      \node (base) [right = of se] {$\DF{\CCat}$};
      \path[fibration] (tot) edge (base);
    \end{tikzpicture}
  \end{equation*}

  Therefore, we are in a position to define the type code using the universal property of the cartesian lift:
  \begin{equation*}
    \begin{tikzpicture}[diagram,baseline = (sw.base)]
      \SpliceDiagramSquare{
        west/style = lies over,
        east/style = lies over,
        sw = \Mod{\CCat}{\TY},
        se = \CUTY,
        nw = \GlFun{\TY},
        ne = \GlUTY,
        north = \CodeOf{\GlFun{\El}},
        south = \CodeOf{\Mod{\CCat}{\El}},
        north/node/style = upright desc,
        height = 1.5cm,
        width = 3cm,
      }
      \node (Gl/T/Pi) [above left = 2.5cm of nw] {$\PI{\GlFun{\TY}}$};
      \node (C/T/Pi) [below = 1.5cm of Gl/T/Pi] {$\PI{\Mod{\CCat}{\TY}}$};
      \path[lies over] (Gl/T/Pi) edge (C/T/Pi);
      \path[->] (C/T/Pi) edge node [sloped,below] {$\Mod{\CCat}{\TyPi}$} (sw);
      \path[->,bend left=30] (Gl/T/Pi) edge (ne);
      \path[->,exists] (Gl/T/Pi) edge node [desc] {$\GlFun{\TyPi}$} (nw);
      \node (tot) [right = of ne] {$\GlCat$};
      \node (base) [right = of se] {$\DF{\CCat}$};
      \path[fibration] (tot) edge (base);
    \end{tikzpicture}
  \end{equation*}

  We have got the downstairs map aligned properly; to complete the algebra
  $\Mor{\Mod{\CCat}{\PI{\El}}}{\Mod{\CCat}{\El}}$ with a properly aligned
  upstairs map, we may use the universal property of the pullback and the fact
  that $\GlProj$ is lex.
\end{con}

The closure under dependent sum works identically. We will, however, illustrate the closure under path types.

\begin{con}[Closure under path types]
  First of all, the universe $\GlUTY$ is closed under path types because path
  types may be constructed (up to isomorphism) using the interval, dependent
  products, and subobject comprehension (all of which are small). Therefore, we
  have a cartesian map $\Mor{\PATH{\GlUEl}}{\GlUEl} : \CartArr{\GlCat}$. Using
  functoriality of $\PATH$ and the fact that path types are preserved by
  $\GlProj$, we have:
  \begin{equation*}
    \begin{tikzpicture}[diagram,baseline = (l/sw.base)]
      \SpliceDiagramSquare<l/>{
        west/style = |->,
        east/style = |->,
        nw = \PATH{\GlFun{\El}},
        ne = \PATH{\GlUEl},
        sw = \PATH{\Mod{\CCat}{\El}},
        se = \PATH{\CUEl},
        height = 1.5cm,
      }
      \SpliceDiagramSquare<r/>{
        glue = west,
        glue target = l/,
        east/style = |->,
        se = \CUEl,
        ne = \GlUEl,
        height = 1.5cm,
      }
      \node (tot) [right = of r/ne] {$\CartArr{\GlCat}$};
      \node (base) [right = of r/se] {$\CartArr{\DF{\CCat}}$};
      \path[->] (tot) edge (base);
    \end{tikzpicture}
  \end{equation*}

  By Lemma~\ref{lem:fibered-universe-realignment}, we may realign the upstairs
  map to lie over $\Mod{\CCat}{\TyPath}\circ\CodeOf{\Mod{\CCat}{\El}}$; therefore,
  we obtain a code for the glued path type lying over the original code by the
  universal property of the cartesian lift:
  \begin{equation*}
    \begin{tikzpicture}[diagram,baseline = (sw.base)]
      \SpliceDiagramSquare{
        west/style = lies over,
        east/style = lies over,
        width = 3cm,
        height = 1.5cm,
        sw = \Mod{\CCat}{\TY},
        se = \CUTY,
        nw = \GlFun{\TY},
        ne = \GlUTY,
        north = \CodeOf{\GlFun{\El}},
        south = \CodeOf{\Mod{\CCat}{\El}},
        north/node/style = upright desc,
      }
      \node (Gl/T/Path) [above left = 2.5cm of nw] {$\PATH{\GlFun{\TY}}$};
      \node (C/T/Path) [below = 1.5cm of Gl/T/Path] {$\PATH{\Mod{\CCat}{\TY}}$};
      \path[lies over] (Gl/T/Path) edge (C/T/Path);
      \path[->] (C/T/Path) edge node [sloped,below] {$\Mod{\CCat}{\TyPath}$} (sw);
      \path[->,bend left=30] (Gl/T/Path) edge (ne);
      \path[->,exists] (Gl/T/Path) edge node [desc] {$\MathSmall{\GlFun{\TyPath}}$} (nw);
      \node (tot) [right = of ne] {$\GlCat$};
      \node (base) [right = of se] {$\DF{\CCat}$};
      \path[fibration] (tot) edge (base);
    \end{tikzpicture}
    \qedhere
  \end{equation*}
\end{con}

\begin{con}\label{con:computable-booleans}
  We may form a computability structure over the booleans by \emph{opcartesian}
  lift, using the fact that the gluing (op)fibration preserves colimits:
  \begin{equation*}
    \begin{tikzpicture}[diagram,baseline = (sw.base)]
      \SpliceDiagramSquare{
        west/style = op lies over,
        east/style = op lies over,
        width = 3.5cm,
        height = 1.5cm,
        north/style = exists,
        nw = \Two{\GlCat},
        sw = \Two{\DF{\CCat}},
        se = \Mod{\CCat}*{\Pb{\El}{\TyBool}},
        ne = \GlFun{\Pb{\El}{\TyBool}},
        north = \MathSmall{\brk{\GlFun{\TmTt}\mid\GlFun{\TmFf}}},
        south = \MathSmall{\brk{\Mod{\CCat}{\TmTt}\mid\Mod{\CCat}{\TmFf}}},
      }
      \node (tot) [right = of ne] {$\GlCat$};
      \node (base) [right = of se] {$\DF{\CCat}$};
      \path[opfibration] (tot) edge (base);
    \end{tikzpicture}
    \qedhere
  \end{equation*}
\end{con}

\begin{lem}\label{lem:glued-booleans-small}
  The computability structure $\GlFun{\Pb{\El}{\TyBool}}$ from Construction~\ref{con:computable-booleans} is small.
\end{lem}

\begin{proof}
  By the characterization theorem, it suffices to check that it lies
  over a small object (obvious), and that vertical map it induces is small in
  $\GlProj\brk{\Mod{\CCat}*{\Pb{\El}{\TyBool}}}$. To see that this is the case,
  we compute this vertical map as follows:
  \begin{equation*}
    \begin{tikzpicture}[diagram]
      \SpliceDiagramSquare{
        west/style = lies over,
        east/style = lies over,
        width = 3cm,
        height = 1.5cm,
        nw = \GlUnProjR{\Mod{\CCat}*{\Pb{\El}{\TyBool}}},
        ne = \GlUnProjR{\Mod{\CCat}{\EMP}} \mathrlap{\cong \GlFun{\EMP}},
        sw = \Mod{\CCat}*{\Pb{\El}{\TyBool}},
        se = \Mod{\CCat}{\EMP},
        south = !_{\Mod{\CCat}*{\Pb{\El}{\TyBool}}},
      }
      \node (nww) [above left = of nw] {$\GlFun{\Pb{\El}{\TyBool}}$};
      \path[lies over,bend right=30] (nww) edge (sw);
      \path[->,exists] (nww) edge (nw);
      \path[->,bend left=30] (nww) edge node [above,sloped] {$!_{\GlFun{\Pb{\El}{\TyBool}}}$} (ne);
    \end{tikzpicture}
  \end{equation*}

  The vertical map above can be seen to be small using the fact that
  $\Mor{\Two{\cSET}}{\F{\Mod{\CCat}*{\Pb{\El}{\TyBool}}}}$ is small.
\end{proof}

\begin{con}[Booleans]
  By Lemma~\ref{lem:glued-booleans-small}, there exists some
  characteristic map $\Mor[\CodeOf{\Pb{\El}{\TyBool}}]{\GlFun{\EMP}}{\GlUTY}$ for
  $\GlFun{\Pb{\El}{\TyBool}}$ lying over a characteristic map for the object
  $\Mod{\CCat}*{\Pb{\El}{\TyBool}}$. Therefore, again using the realignment
  lemma, we may define a suitable code for the booleans using the universal
  property of the cartesian lift:
  \begin{equation*}
    \begin{tikzpicture}[diagram,baseline = (sw.base)]
      \SpliceDiagramSquare{
        west/style = lies over,
        east/style = lies over,
        width = 3cm,
        height = 1.5cm,
        sw = \Mod{\CCat}{\TY},
        se = \CUTY,
        nw = \GlFun{\TY},
        ne = \GlUTY,
        north = \CodeOf{\GlFun{\El}},
        south = \CodeOf{\Mod{\CCat}{\El}},
        north/node/style = upright desc,
      }
      \node (Gl/Emp) [above left = 2.5cm of nw] {$\GlFun{\EMP}$};
      \node (C/Emp) [below = 1.5cm of Gl/Emp] {$\Mod{\CCat}{\EMP}$};
      \path[lies over] (Gl/Emp) edge (C/Emp);
      \path[->] (C/Emp) edge node [sloped,below] {$\Mod{\CCat}{\TyBool}$} (sw);
      \path[->,bend left=30] (Gl/Emp) edge (ne);
      \path[->,exists] (Gl/Emp) edge node [desc] {$\MathSmall{\GlFun{\TyBool}}$} (nw);
      \node (tot) [right = of ne] {$\GlCat$};
      \node (base) [right = of se] {$\DF{\CCat}$};
      \path[fibration] (tot) edge (base);
    \end{tikzpicture}
  \end{equation*}

  To define the elimination form, we must exhibit a choice of
  lifts of the following form natural in $X : \GlCat$, lying over the
  corresponding lifts that we have fixed in $\DF{\CCat}$:
  \begin{equation}\label{diag:glued-bool-lift}
    \begin{tikzpicture}[diagram, baseline=(sw.base)]
      \SpliceDiagramSquare{
        width = 3cm,
        nw = X\times \Two{\GlCat},
        sw = X\times\GlFun{\Pb{\El}{\TyBool}},
        se = \GlFun{\TY},
        ne = \GlFun{\EL},
        east = \GlFun{\El},
        south = B,
        north = b,
        west = \MathSmall{\ArrId{X}\times \brk{\GlFun{\TmTt}\mid\GlFun{\TmFf}}}
      }
      \path[->,exists] (sw) edge node [desc] {$\widetilde{b}$} (ne);
    \end{tikzpicture}
    \qquad
    \begin{tikzpicture}[diagram,baseline = (sw.base)]
      \SpliceDiagramSquare{
        width = 3cm,
        nw = \MathSmall{\GlProj{X}\times \Two{\DF{\CCat}}},
        sw = \MathSmall{\GlProj{X}\times\GlFun{\Pb{\El}{\TyBool}}},
        se = \Mod{\CCat}{\TY},
        ne = \Mod{\CCat}{\EL},
        east = \Mod{\CCat}{\El},
        south = \GlProj{B},
        north = \GlProj{b},
      }
      \path[->] (sw) edge node [desc] {$\widetilde{\GlProj{b}}$} (ne);
    \end{tikzpicture}
  \end{equation}

  For each $X:\GlCat$, the existence of such a lift is guaranteed by the
  universal property of the induced opcartesian map
  $\Mor{X\times\Two{\GlCat}}{X\times\GlFun{\Pb{\El}{\TyBool}}}$:
  \begin{equation*}
    \begin{tikzpicture}[diagram]
      \SpliceDiagramSquare{
        west/style = op lies over,
        east/style = op lies over,
        width = 3.5cm,
        height = 1.5cm,
        nw = X\times\Two{\GlCat},
        sw = \GlProj{X}\times\Two{\DF{\CCat}},
        se = \MathSmall{\GlProj{X}\times\Mod{\CCat}*{\Pb{\El}{\TyBool}}},
        ne = \MathSmall{X\times\GlFun{\Pb{\El}{\TyBool}}},
      }
      \node (nee) [right = 3cm of ne, yshift = 1cm] {$\GlFun{\EL}$};
      \node (see) [right = 3cm of se, yshift = 1cm] {$\Mod{\CCat}{\EL}$};
      \path[op lies over] (nee) edge (see);
      \path[->] (se) edge node [sloped,below] {$\widetilde{\GlProj{b}}$} (see);
      \path[->,exists] (ne) edge node [desc] {$\widetilde{b}$} (nee);
      \path[->,bend left = 30] (nw) edge node [above] {$b$} (nee);
      \node (tot) [right = 5cm of ne] {$\GlCat$};
      \node (base) [right = 5cm of se] {$\DF{\CCat}$};
      \path[opfibration] (tot) edge (base);
    \end{tikzpicture}
  \end{equation*}

  To see that the choice of $\widetilde{b}$ is natural in $X$, we will observe
  the stronger property that it is the \emph{unique} lift lying over
  $\widetilde{\GlProj{b}}$. Computing the opcartesian lift explicitly, we see
  that $\GlEl{\GlFun{\Pb{\El}{\TyBool}}} = \GlEl{\Two{\GlCat}}$; moreover
  $\GlEl$ preserves colimits because $\F$ is left exact~\cite{taylor:1999}, so
  in fact $\GlEl{\GlFun{\Pb{\El}{\TyBool}}} = \Two{\cSET}$. Therefore, the
  (non-unique) lifting situation of Diagram~\ref{diag:glued-bool-lift} becomes
  a unique lifting situation in cubical sets.
\end{con}

\subsection{Universe of Bishop sets}

\begin{con}[Glued universe \`a la Tarski]\label{con:ir-universe}
  By induction-recursion, we may define a universe $\IRTY : \GlCat$
  simultaneously with a decoding function $\Mor[\IRDec]{\IRTY}{\GlFun{\TY}}$
  closed under dependent product, dependent sum, path, and boolean.
  \begin{mathflow}
    \inferrule{
    }{
      \mathtt{bool} : \IRTY &
      \IRDec{\mathtt{bool}} = \GlFun{\TyBool}
    }
    \split
    \inferrule{
      A : \IRTY &
      B : \GlFun{\El}\brk{\IRDec{A}}\Rightarrow \IRTY
    }{
      \mathtt{pi}\prn{A,B} : \IRTY
      &
      \IRDec{\mathtt{pi}\prn{A,B}} = \GlFun{\TyPi}\prn{\IRDec{A},\IRDec\circ B}
    }
    \split
    \inferrule{
      A : \IRTY &
      B : \GlFun{\El}\brk{\IRDec{A}}\Rightarrow \IRTY
    }{
      \mathtt{sg}\prn{A,B} : \IRTY
      &
      \IRDec{\mathtt{sg}\prn{A,B}} = \GlFun{\TySg}\prn{\IRDec{A},\IRDec\circ B}
    }
    \split
    \inferrule{
      A : \GlFun{\DIM}\Rightarrow \IRTY &
      a : \prn{i : \GlFun{\DIM}, \_ : \Bdry{i}}\Rightarrow \GlFun{\El}\brk{\IRDec{A\prn{i}}}
    }{
      \mathtt{path}\prn{A,a} : \IRTY
      &
      \IRDec{\mathtt{path}\prn{A,a}} = \GlFun{\TyPath}\prn{\IRDec\circ A,a}
    }
  \end{mathflow}
  \bigskip

  $\IRTY$ lies \emph{not} over the type-theoretic universe \`a la Tarski
  $\Mod{\CCat}{\TskUTY} : \DF{\CCat}$, but rather over a genuine
  inductive-recursive universe in $\DF{\CCat}$. Because this $\GlProj{\IRTY}$
  is the \emph{least} universe closed under the mentioned connectives, we
  obtain a universal map
  $\Mor[i_{\Mod{\CCat}{\TskUTY}}]{\GlProj{\IRTY}}{\Mod{\CCat}{\TskUTY}}$ which
  automatically commutes with all connectives. We may therefore shift $\IRTY$
  to lie over $\Mod{\CCat}{\TskUTY}$ by opcartesian lift against this
  universal map:
  \begin{equation}
    \DiagramSquare{
      west/style = op lies over,
      east/style = op lies over,
      width = 3cm,
      height = 1.5cm,
      nw = \IRTY,
      north/style = exists,
      sw = \GlProj{\IRTY},
      se = \Mod{\CCat}{\TskUTY},
      south = i_{\Mod{\CCat}{\TskUTY}},
      ne = \GlFun{\TskUTY},
      north = \MathSmall{\prn{i_{\Mod{\CCat}{\TskUTY}}}_\dagger\IRTY}
    }
  \end{equation}

  A decoding map $\Mor[\GlFun{\TskDec}]{\GlFun{\TskUTY}}{\GlFun{\TY}}$ is
  inherited using the universal property of the opcartesian lift:
  \begin{equation}\label{diag:ir-universe:1}
    \begin{tikzpicture}[diagram,baseline = (sw.base)]
      \SpliceDiagramSquare{
        west/style = op lies over,
        east/style = op lies over,
        width = 2.5cm,
        height = 1.5cm,
        nw = \IRTY,
        sw = \GlProj{\IRTY},
        se = \Mod{\CCat}{\TskUTY},
        south = i_{\Mod{\CCat}{\TskUTY}},
        ne = \GlFun{\TskUTY}
      }
      \node (nee) [right = 3cm of ne, yshift = 1cm] {$\GlFun{\TY}$};
      \node (see) [right = 3cm of se, yshift = 1cm] {$\Mod{\CCat}{\TY}$};
      \path[op lies over] (nee) edge (see);

      \path[->] (se) edge node [sloped,below] {$\Mod*{\CCat}{\TskDec}$} (see);
      \path[->,exists] (ne) edge node [desc] {$\MathSmall{\GlFun{\TskDec}}$} (nee);
      \path[->,bend left = 30] (nw) edge node [sloped,above] {$\IRDec$} (nee);
    \end{tikzpicture}
  \end{equation}

  The map $\IRDec$ can be seen to lie strictly over the downstairs composite in
  Diagram~\ref{diag:ir-universe:1} using the uniqueness of maps out of
  inductive-recursive universes. Moreover, the object $\GlFun{\TskUTY}$ is
  small, so we have a characteristic map
  $\Mor[\CodeOf{\GlFun{\TskUTY}}]{\GlFun{\EMP}}{\GlUTY}$.

  The constructions above may be used as the basis for a universe \`a la Tarski
  $\GlFun{\TskU} : \GlFun{\TY}$ lying over the type $\Mod{\CCat}{\TskU} :
  \Mod{\CCat}{\TY}$.
  By Lemma~\ref{lem:fibered-universe-realignment}, we may realign the
  characteristic map $\Mor[\CodeOf{\GlFun{\TskUTY}}]{\GlFun{\EMP}}{\GlUTY}$ to lie
  over the composite
  $\Mor[\CodeOf{\Mod{\CCat}{\El}}\circ\Mod{\CCat}{\TskU}]{\Mod{\CCat}{\EMP}}{\CUTY}$;
  in this way, we obtain an appropriate code for the universe \`a la Tarski by means of
  the universal property of the cartesian lift below:
  \begin{equation*}
    \begin{tikzpicture}[diagram,baseline = (sw.base)]
      \SpliceDiagramSquare{
        west/style = lies over,
        east/style = lies over,
        width = 3cm,
        height = 1.5cm,
        sw = \Mod{\CCat}{\TY},
        se = \CUTY,
        south = \CodeOf{\Mod{\CCat}{\El}},
        ne = \GlUTY,
        nw = \GlFun{\TY},
        north = \CodeOf{\GlFun{\El}},
        north/node/style = upright desc,
      }
      \node (nww) [left = 2.5cm of nw, yshift = 1.5cm] {$\GlFun{\EMP}$};
      \node (sww) [left = 2.5cm of sw, yshift = 1.5cm] {$\Mod{\CCat}{\EMP}$};
      \path[->] (sww) edge node [sloped,below] {$\Mod{\CCat}{\TskU}$} (sw);
      \path[->,exists] (nww) edge node [desc] {$\MathSmall{\GlFun{\TskU}}$} (nw);
      \path[lies over] (nww) edge (sww);
      \path[->,bend left=30] (nww) edge (ne);
    \end{tikzpicture}
    \qedhere
  \end{equation*}
\end{con}

We have \emph{not} shown that the universe \`a la Tarski $\GlFun{\TskUTY}$ is closed under the
appropriate connectives --- we only know that $\IRTY$ is closed under those connectives. Prior to
demonstrating this, however, we must record a few facts about the behavior of $\F$ on pushfowards.
\begin{lemC}[{\cite[Lemma~3.5]{sterling-angiuli:2020}}]\label{lem:pshfwd-comparison}
  Let $\Mor[F]{\XCat}{\ECat}$ be any left exact functor, and let $f_*g :
  \Sl{\XCat}{Z}$ be the pushforward of $\Mor[g]{X}{Y}$ along $\Mor[f]{Y}{Z}$.
  Then we have a canonical (not usually invertible) comparison map
  $\Mor|->,exists|{F\prn{f_*g}}{F(f)_*F(g)}$.
  \qed
\end{lemC}

\begin{cor}
  We have a canonical comparison map commuting the generic dependent product family
  past $\F\circ\Mod{\CCat}$ in the following sense:
  \[
    \Mor|->,exists|[\pprn{-}]{\F{\Mod{\CCat}{\PI{\El}}}}{\PI{\prn{\F{\Mod{\CCat}{\El}}}}}
  \]
\end{cor}
\begin{proof}
  By Lemma~\ref{lem:pshfwd-comparison}, using the fact that $\F$ is left exact.
\end{proof}

\begin{con}[Constructors for the universe \`a la Tarski]
  We will now show that $\GlFun{\TskUTY}$ is closed under the necessary connectives; we consider
  only the case of dependent products, since the remaining constructors work identically.

  We must construct a morphism
  $\Mor[\GlFun{\TmCodePi}]{\GlFun{\PI{\TskUTY}}}{\GlFun{\TskUTY}}$ that lies
  over
  $\Mor[\Mod{\CCat}{\TmCodePi}]{\Mod{\CCat}{\PI{\TskUTY}}}{\Mod{\CCat}{\TskUTY}}$.
  Unfolding this situation, we wish to construct the following dotted map in
  $\cSET$:
  \begin{equation*}
    \DiagramSquare{
      width = 5cm,
      height = 1.5cm,
      sw = \F{\PI{\Mod{\CCat}{\TskUTY}}},
      se = \F{\Mod{\CCat}{\TskUTY}},
      south = \F{\Mod{\CCat}{\TmCodePi}},
      nw = \GlEl{\PI{\GlFun{\TskUTY}}},
      ne = \GlEl{\GlFun{\TskUTY}},
      north = \GlEl{\GlFun{\TmCodePi}},
      north/style = exists,
      west = \PI{\GlFun{\TskUTY}},
      east = \GlFun{\TskUTY},
    }
  \end{equation*}

  We will define this map in the language of cubical sets. Accordingly, we begin by computing the fibers of
  $\GlFun{\TskUTY}$ and $\GlFun{\PI{\TskUTY}}$:
  \begin{align*}
    A : \F{\Mod{\CCat}{\TskUTY}} \mid \GlFun{\TskUTY}[A] &= \{\mathfrak{A} : \IRTY \mid \F{i_{\Mod{\CCat}{\TskUTY}}(\GlProj{\mathfrak{A}})} = A \}\\
    G : \F{\Mod{\CCat}{\PI{\TskUTY}}} \mid \GlFun{\PI{\TskUTY}}[G] &=
    \MathSmall{
      (\mathfrak{A} : \GlFun{\TskUTY}[\pprn{G}_0]) \times
      (a : \F{\Mod{\CCat}{\El}}[\pprn{G}_0])(\mathfrak{a} : \prn{\GlFun{\El}[\IRDec{\mathfrak{A}}]}[a]) \Rightarrow \GlFun{\TskUTY}[\pprn{G}_1~a]
    }
  \end{align*}

  With these fibers in hand, we may define $\GlEl{\GlFun{\TmCodePi}}$:
  \[
    \GlEl{\GlFun{\TmCodePi}}[G](\mathfrak{A}, \mathfrak{B}) = \mathtt{pi}(\mathfrak{A}, \lambda \brk{a,\mathfrak{a}}.~\mathfrak{B}~a~\mathfrak{a})
    \qedhere
  \]
\end{con}

\begin{lem}
  The foregoing construction of $\GlFun{\TskUMor}$ is boundary separated.
\end{lem}

\begin{proof}
  By induction, using the fact that $\Mod{\CCat}{\prn{\TskUMor}}$ is boundary separated.
\end{proof}

\begin{con}[Lifting structure]
  The type-case lifting structure of Specification~\ref{spec:type-case} may be
  constructed using the induction-recursion principle of $\GlFun{\TskUTY}$, and
  using the corresponding lifting structure for $\Mod{\CCat}{\TskUTY}$.
\end{con}

\begin{lem}[Coercion and composition]
  There are coercion and compositions operation defined on glued lines of sets
  $\GlFun{\DIM}\Rightarrow \GlFun{\TskUTY}$ which have the types given in
  Remark~\ref{rem:com-coe}, satisfying the regularity law
  (Specification~\ref{spec:regularity}) as well as the
  connective-specific equations of Specification~\ref{spec:coe-at-connectives}.
\end{lem}

\begin{proof}
  Coercion and \emph{homogeneous} composition operations are first defined
  using the induction principle of the inductive recursive universe of sets,
  taking the equations of Specification~\ref{spec:coe-at-connectives} and
  Lemma~\ref{lem:hcom-at-connectives} respectively as definitions; general
  composition is defined by the standard reduction to coercion and homogeneous
  composition
  (Lemma~\ref{lem:hcom-com-compatibility})~\cite{abcfhl:2021,angiuli-favonia-harper:2017}.
  The case-wise definition of coercion makes crucial use of the tininess of the
  computability interval $\GlFun{\DIM}$ (Lemma~\ref{lem:tininess}).
\end{proof}

\subsection{Summary}

By means of the foregoing constructions, we have defined a functorial
interpretation $\GlFun$ of XTT into a glued logos $\GlCat$ lying over
$\DF{\CCat}$ for an arbitrary XTT-model $\Mod{\CCat}$.  In
Section~\ref{sec:canonicity}, we will reassamble this data into a gluing model
$\Mod{\CptGlCat}$ equipped with a structure map
$\Mor{\Mod{\CptGlCat}}{\Mod{\CCat}}$ that can be used to prove canonicity.

\section{Canonicity for XTT}%
\label{sec:canonicity}

In Section~\ref{sec:computability}, we introduced the logos $\GlCat$ of \emph{computability structures},
along with a representable map functor $\Mor[\GlFun]{\ThCat}{\GlCat}$. These constructions contain
the essence of the proof of canonicity for XTT, but in order to complete the proof we must assemble
$\GlCat$ and $\GlFun$ into a model of $\ThCat$ equipped with a morphism to $\Mod{\CCat}$. In this
section, we construct this \emph{gluing model} of XTT and prove the following canonicity theorem:
\begin{restatable*}[Canonicity]{thm}{CanonicityThm}
  Given a closed term $M:\TyBool$, either $M = \TmTt$ or $M = \TmFf$.
\end{restatable*}

\subsection{The canonicity model of XTT}

Recall from Section~\ref{sec:functorial-semantics} that a model of $\ThCat$ is a category with a
terminal object $\CCat$ paired with a representable map functor
$\Mor[\Mod{\CCat}]{\ThCat}{\DF{\CCat}}$. The category of computability structures $\GlCat$ is not a
category of discrete fibrations, and therefore we cannot directly take $\GlCat,\GlFun$ as the gluing
model.
Instead, following~\cite{sterling-angiuli:2020} we will shift to working with discrete fibrations on
\emph{compact} computability structures $\DF{\CptGlCat}$ and use the representable map functor
$\Mor[\GlNv]{\GlCat}{\DF{\CptGlCat}}$ to uniformly transfer the computability structures from
Section~\ref{sec:computability} from $\GlCat$ to $\DF{\CptGlCat}$.

\begin{con}[Gluing model]
  The \emph{gluing model} $\Mor[\Mod{\CptGlCat}]{\ThCat}{\DF{\CptGlCat}}$ is defined as the
  composition $\GlNv \circ \GlFun$. Diagrammatically:
  \begin{equation*}
    \begin{tikzpicture}[diagram, node distance = 2.5cm]
      \node (T) {$\ThCat$};
      \node (G) [right = of T] {$\GlCat$};
      \node (DF) [right = of G] {$\DF{\CptGlCat}$};
      \path[exists, bend left, ->] (T) edge node [upright desc] {$\Mod{\CptGlCat}$} (DF);
      \path[->] (T) edge node [below] {$\GlFun$} (G);
      \path[->] (G) edge node [below] {$\GlNv$} (DF);
    \end{tikzpicture}
    \qedhere%
    \label{diag:gluing-model}
  \end{equation*}
\end{con}

It remains to construct the morphism
$\Mor{\Mod{\CptGlCat}}{\Mod{\CCat}}:\ModCat{\ThCat}$. We will begin by
constructing the data of this morphism, checking the Beck-Chevalley condition
afterward.

\begin{con}[Gluing homomorphism data]\label{con:gluing-homomorphism-data}
  We expect a morphism of models tracked by the gluing fibration
  $\FibMor[\CptGlProj]{\CptGlCat}{\CCat}$ at the level of contexts; it remains
  to make a choice of functors $\Mor{\Mod{\CptGlCat}{J}}{\Mod{\CCat}{J}} :
  \DF{}$ lying over $\CptGlProj$ in $\FibMor{\DF{}}{\CAT}$ (natural in
  $J:\ThCat$) to exhibit the action of the homomorphism on judgments.
  Fixing a judgment $J:\ThCat$, we construct each of the components as follows:
  \begin{equation}\label{diag:gluing-homomorphism}
    \begin{tikzpicture}[diagram, node distance=2.5cm,baseline = (base.base)]
      \node (0) {$\Mod{\CptGlCat}{J}$};
      \node (1) [right = of 0] {$\GlNv{\GlFun{J}}$};
      \node (2) [right = 3.5cm of 1] {$\GlProj{\GlFun{J}}$};
      \node (3) [right = of 2] {$\Mod{\CCat}{J}$};

      \node (0') [below = 1.5cm of 0] {$\CptGlCat$};
      \node (1') [below = 1.5cm of 1] {$\CptGlCat$};
      \node (2') [below = 1.5cm of 2] {$\CCat$};
      \node (3') [below = 1.5cm of 3] {$\CCat$};

      \path[->] (0) edge node [above] {$\cong$} (1);
      \path[->] (1) edge node [above] {$\ArrCat{\GlProj}$} (2);
      \path[->] (2) edge node [above] {$\GlCell{J}$} (3);

      \path[->] (0') edge node [below] {$\ArrId{\CptGlCat}$} (1');
      \path[->] (1') edge node [below] {$\CptGlProj$} (2');
      \path[->] (2') edge node [below] {$\ArrId{\CCat}$} (3');

      \path[lies over] (0) edge (0');
      \path[lies over] (1) edge (1');
      \path[lies over] (2) edge (2');
      \path[lies over] (3) edge (3');

      \node (tot) [right = of 3] {$\DF{}$};
      \node (base) [below = 1.5cm of tot] {$\CAT$};
      \path[fibration] (tot) edge (base);
    \end{tikzpicture}
  \end{equation}

  Diagram~\ref{diag:gluing-homomorphism} above exhibits a natural
  transformation, because $\GlCell{J}$ is a component of the natural
  transformation $\Mor[\GlCell]{\GlProj\circ\GlFun}{\Mod{\CCat}}$.
\end{con}

It remains to check that the naturality squares induced by
Construction~\ref{con:gluing-homomorphism-data} at representable maps satisfy
the Beck-Chevalley condition.  First, we record a simple characterization of
the right adjoint $\ReprAdj{f}$ of a representable map in $\DF{\CCat}$:
\begin{lemC}[{\cite[Corollary~3.11]{uemura:2019}}]%
  \label{lem:repr-adj-pb}
  \pushQED{\qed}%
  If $\Mor[f]{J}{I} : \DF{\CCat}$ is a representable map, the right adjoint $\ReprAdj{f}$ sends an
  element $\Mor[y]{\DFYo[\CCat]{C}}{I}$ to the upstairs element $x$ determined by the following pullback square:
  \[
    \DiagramSquare[node distance=2.5cm]{
      nw = \DFYo[\CCat]\prn{C.y},
      nw/style = pullback,
      ne = J,
      sw = \DFYo[\CCat]{C},
      se = I,
      east = f,
      south = y,
      north = x,
      north/style = exists,
    }
    \qedhere
  \]
  \popQED
\end{lemC}

\begin{lem}%
  \label{lem:cptglproj-beck-chevalley}
  If $\Mor[f]{J}{I} : \ThCat$ is a representable map, the following naturality square satisfies
  the Beck-Chevalley condition:
  \[
    \DiagramSquare[node distance=2.5cm]{
      nw = \Mod{\CptGlCat}{J},
      ne = \Mod{\CCat}{J},
      se = \Mod{\CCat}{I},
      sw = \Mod{\CptGlCat}{I},
      east = \Mod{\CCat}{f},
      west = \Mod{\CptGlCat}{f},
      north = \Mod{\CptGlProj}{J},
      south = \Mod{\CptGlProj}{I}
    }
  \]
\end{lem}

\begin{proof}
  We must show that the following square commutes up to canonical isomorphism:
  \begin{equation}
    \DiagramSquare[node distance=2.5cm]{
      sw = \Mod{\CptGlCat}{J},
      se = \Mod{\CCat}{J},
      ne = \Mod{\CCat}{I},
      nw = \Mod{\CptGlCat}{I},
      east = \ReprAdj{\Mod{\CCat}{f}},
      west = \ReprAdj{\Mod{\CptGlCat}{f}},
      south = \Mod{\CptGlProj}{J},
      north = \Mod{\CptGlProj}{I}
    }
  \end{equation}

  To begin with, we fix an element $\Mor[y]{\CptIncl{\Gamma}}{\GlFun{J}} : \Mod{\CptGlCat}{J}$. Inspecting the
  definition of $\Mod{\CptGlProj}{I}$ from Construction~\ref{con:gluing-homomorphism-data}, we see that
  $\Mod{\CptGlProj}{I}\prn{y} = \GlCell{I} \circ \GlProj{y}$. Therefore, by
  Lemma~\ref{lem:repr-adj-pb},
  $\ReprAdj{\Mod{\CCat}{f}}(\Mod{\CptGlProj}{I}\prn{y}) = \ReprAdj{\Mod{\CCat}{f}}(\GlCell{I} \circ
  \GlProj{y})$ is the top map of the following pullback square:
  \begin{equation}
    \DiagramSquare{
      width = 3cm,
      ne = \Mod{\CCat}{J},
      nw/style = pullback,
      se = \Mod{\CCat}{I},
      nw = \DFYo[\CCat]{\cdots},
      sw = \DFYo[\CCat]{\CptGlProj{\Gamma}},
      east = \Mod{\CCat}{f},
      north/style = exists,
      south = {\GlCell{I} \circ \GlProj{y}},
    }%
    \label{diag:bc-map-one}
  \end{equation}

  Similarly, we may compute $\ReprAdj{\Mod{\CCat}{f}}(\Mod{\CptGlProj}{I}(y))$ as the top of the
  following composite square:
  \begin{equation}
    \begin{tikzpicture}[diagram, baseline = (sw.base)]
      \SpliceDiagramSquare<l/>{
        width = 2.5cm,
        ne = \GlProj{\GlFun{J}},
        nw/style = pullback,
        se = \GlProj{\GlFun{I}},
        nw = \DFYo[\CCat]{\CptGlProj{\cdots}},
        sw = \DFYo[\CCat]{\CptGlProj{\Gamma}},
        east = \GlProj{\GlFun{f}},
        east/node/style = {upright desc},
        south = {\GlProj{y}},
      }
      \SpliceDiagramSquare<r/>{
        glue = west,
        glue target = l/,
        ne = \Mod{\CCat}{J},
        se = \Mod{\CCat}{I},
        east = \Mod{\CCat}{f},
        south = \GlCell{I},
        north = \GlCell{J}
      }
    \end{tikzpicture}%
    \label{diag:bc-map-two}
  \end{equation}

  To show that the top of Diagram~\ref{diag:bc-map-one} is isomorphic to
  Diagram~\ref{diag:bc-map-two}, it suffices to check that the outer square of
  Diagram~\ref{diag:bc-map-two} is cartesian. Both $\GlCell{J}$ and
  $\GlCell{I}$ are isomorphisms so the right-hand square of
  Diagram~\ref{diag:bc-map-two} is cartesian, so the result follows immediately
  from the pullback pasting lemma.
\end{proof}

\begin{cor}%
  \label{cor:cptglproj-mor}
  The natural transformation $\Mod{\CptGlProj}$ from Construction~\ref{con:gluing-homomorphism-data} is a morphism of models.
  \qed
\end{cor}

\subsection{The canonicity theorem}

Having constructed the gluing model and the projection onto $\Mod{\CCat}$, we are now in a position
to prove canonicity. First, we stop considering an arbitrary model $\Mod{\CCat}$ and work
exclusively with $\Mod{\ICat}$, the bi-initial model of $\ThCat$
(Theorem~\ref{thm:bi-initial-model}). Bi-initiality ensures that there is a morphism
$\Mor[\IMor{\CptGlCat}]{\Mod{\ICat}}{\Mod{\CptGlCat}}$ and, because $\ArrId{\Mod{\ICat}}$ and
$\Mod{\CptGlProj} \circ \Mod{\IMor{\CptGlCat}}$ are both objects of
$\Hom[\ModCat{\ThCat}]{\Mod{\ICat}}{\Mod{\ICat}}$, there is a unique invertible 2-morphism
$\Mor[\ICell]{\ArrId{\Mod{\ICat}}}{\Mod{\CptGlProj} \circ \Mod{\IMor{\CptGlCat}}}$.

In prior presentations of gluing with strict
homomorphisms~\cite{sterling-angiuli-gratzer:2019,kaposi-huber-sattler:2019,coquand-huber-sattler:2019,gratzer-kavvos-nuyts-birkedal:2020},
$\Mod{\ICat}$ was initial in the \emph{1-categorical} sense, so
$\ArrId{\Mod{\ICat}}$ and $\Mod{\CptGlProj} \circ \Mod{\IMor{\CptGlCat}}$ were
equal on this nose. This in turn implied that for every map
$\Mor[f]{\Mod{\ICat}{\EMP}}{\Mod{\ICat}*{\Pb{\El}{\TyBool}}}$, there existed a map
$\Mor[\bbrk{f}]{\Mod{\CptGlCat}{\EMP}}{\Mod{\CptGlCat}*{\Pb{\El}{\TyBool}}}$
such that $\CptGlProj{\bbrk{f}} = f$. Canonicity followed more-or-less
immediately by inspection of $\bbrk{f}$.

In this work, we have used a weaker notion of morphism and as a consequence, $\Mod{\ICat}$ is merely
\emph{bi}-initial. Accordingly, we cannot immediately conclude that $f$ is in the image of
$\Mod{\CptGlProj}$. In fact, while it is not generally the case that
$f = \Mod{\CptGlProj \circ \IMor{\CptGlCat}}{f}$, it is possible to \emph{realign}
$\Mod{\IMor{\CptGlCat}}{x}$ to an isomorphic arrow which does lie strictly over $f$.

\begin{lem}[Realignment]%
  \label{lem:alignment}
  Given a morphism $\Mor[x]{\DFYo[\ICat]{\Gamma}}{\Mod{\ICat}{X}} : \DF{\ICat}$, there exists an object
  $\bbrk{\Gamma} : \CptGlCat$ and a morphism
  $\Mor[\bbrk{x}]{\DFYo[\CptGlCat]{\bbrk{\Gamma}}}{\Mod{\CptGlCat}{X}} : \DF{\CptGlCat}$ such that
  $\CptGlProj{\bbrk{x}} = x$ and $\IMor{\CptGlCat}{\Gamma} \cong \bbrk{\Gamma}$.
\end{lem}

\begin{proof}
  First we note that $\IMor{\CptGlCat}{x}$ lies over $\CptGlProj{\IMor{\CptGlCat}{x}}$
  by definition. Moreover, because
  $\Mor|fibration|[\CptGlProj]{\CptGlCat}{\CCat}$ is a fibration, we may construct a cartesian lift
  of $\Mor[\ICell]{\Gamma}{\CptGlProj{\IMor{\CptGlCat}{\Gamma}}}$. Diagrammatically, there exists
  two squares:
  \begin{equation*}
    \DiagramSquare{
      height = 1.5cm,
      width = 3cm,
      ne = \IMor{\CptGlCat}{\Gamma},
      se = \CptGlProj{\IMor{\CptGlCat}{\Gamma}},
      nw = \ICell^*{\IMor{\CptGlCat}{\Gamma}},
      sw = \Gamma,
      north/style = exists,
      north = {\ICell^{\dagger}},
      south = {\ICell},
      east/style = lies over,
      west/style = lies over,
    }
    \qquad
    \DiagramSquare{
      height = 1.5cm,
      width = 3.5cm,
      nw = \DFYo[\CptGlCat]{\IMor{\CptGlCat}(\Gamma)},
      sw = \DFYo[\ICat]{\CptGlProj{\IMor{\CptGlCat}{\Gamma}}},
      ne = \Mod{\CptGlCat}{X},
      se = \Mod{\ICat}{X},
      north = \IMor{\CptGlCat}{x},
      south = \CptGlProj{\IMor{\CptGlCat}{x}},
      east/style = lies over,
      west/style = lies over,
    }
  \end{equation*}

  Observe that $\ICell^\dagger$ is an isomorphism because it is cartesian over an isomorphism.
  Unfolding the definition of a 2-morphism, we see that
  $\CptGlProj\prn{\IMor{\CptGlCat}{x} \circ \DFYo[\CptGlCat]{\ICell}} = x$. Accordingly, we may
  paste together these two diagrams to obtain the following:
  \[
    \begin{tikzpicture}[diagram]
      \SpliceDiagramSquare{
        width = 6cm,
        height = 1.5cm,
        nw = \DFYo[\CptGlCat]{\ICell^*{\IMor{\CptGlCat}{\Gamma}}},
        sw = \DFYo[\ICat]{\Gamma},
        ne = \Mod{\CptGlCat}{X},
        se = \Mod{\ICat}{X},
        north = \IMor{\CptGlCat}{x} \circ \DFYo[\CptGlCat]{\ICell^{\dagger}},
        south = \CptGlProj\prn{\IMor{\CptGlCat}{x} \circ \DFYo[\CptGlCat]{\ICell}},
        south/node/style = upright desc,
        east/style = lies over,
        west/style = lies over,
      },
      \path[->, bend right] (sw) edge node[below] {$x$} (se);
    \end{tikzpicture}
  \]

  Therefore, $\bbrk{x} \defeq\IMor{\CptGlCat}{x} \circ \DFYo[\CptGlCat]{\ICell^\dagger}$ is a morphism lying
  strictly over $x$. The final condition is immediate because
  $\bbrk{\Gamma} \cong \IMor{\CptGlCat}(\Gamma)$ by definition.
\end{proof}

\begin{rem}
  One might hope that Lemma~\ref{lem:alignment} implies the existence of a morphism
  $\Mor[\Mod{j}]{\Mod{\ICat}}{\Mod{\CptGlCat}}$ which satisfies the identity
  $\Mod{\CptGlProj} \circ \Mod{j} = \ArrId{\ICat}$ on the nose. Simply applying the realignment
  procedure to every element of $\Mod{\ICat}$ does not result in a morphism, however, because it is
  not functorial, merely pseudo-functorial. This should not be surprising: realignment relies on a
  choice of cartesian lift, which is only pseudo-functorial in general.
\end{rem}

\CanonicityThm
\begin{proof}
  By such a closed term, we mean a morphism
  $\Mor[M]{\DFYo[\ICat]{\ObjTerm{\ICat}}}{\Mod{\ICat}*{\Pb{\El}{\TyBool}}} : \DF{\ICat}$; it
  suffices to prove that $M = \Mod{\ICat}{\TmTt}$ or $M = \Mod{\ICat}{\TmFf}$.
  First, by Lemma~\ref{lem:alignment} we obtain a morphism
  $\Mor[\bbrk{M}]{\DFYo[\ICat]{\bbrk{\EMP}}}{\Mod{\CptGlCat}*{\Pb{\El}{\TyBool}}}
  : \DF{\CptGlCat}$ lying over $M$.
  Next, by definition we have $\Mod{\CptGlCat}*{\Pb{\El}{\TyBool}} =
  \GlNv{\GlFun{\Pb{\El}{\TyBool}}}$ and so $\bbrk{M}$ is uniquely determined by
  morphism $\Mor[\widetilde{\bbrk{M}}]{\ObjTerm{\GlCat}}{\GlFun{\Pb{\El}{\TyBool}}} : \GlCat$.
  Unfolding further, any morphism
  $\Mor[\widetilde{\bbrk{M}}]{\ObjTerm{\GlCat}}{\GlFun{\Pb{\El}{\TyBool}}}$ must be a commuting square
  of the following shape in $\cSET$:
  \begin{equation}\label{diag:canonicity:0}
    \DiagramSquare{
      width = 3.5cm,
      height = 1.5cm,
      nw = \ObjTerm{\cSET},
      sw = \F{\Mod{\CCat}{\EMP}},
      se = \F{\Mod{\CCat}*{\Pb{\El}{\TyBool}}},
      ne = \Two{\cSET},
      south = \F{\GlProj{M}},
      north = \GlEl{\widetilde{\bbrk{M}}},
      west = \cong,
    }
  \end{equation}

  It is immediate that $\GlEl{\widetilde{\bbrk{M}}} = \mathbf{inl}$ or $\GlEl{\widetilde{\bbrk{M}}} = \mathbf{inr}$;
  the fact that Diagram~\ref{diag:canonicity:0} commutes ensures that $\widetilde{\bbrk{M}}$ lies over either
  $\Mod{\ICat}{\TmTt}$ or $\Mod{\ICat}{\TmFf}$.
\end{proof}

\section{Perspective and outlook}\label{sec:outlook}

For decades now, the puzzle of Martin-L\"of's intensional identity type has
remained at the center of type theorists' minds. In 1994, Streicher showed
that intensional type theory was independent of seemingly sensible principles
(like function extensionality) by constructing extremely intensional
counter-models~\cite{streicher:1994}; in 1998, Hofmann and Streicher went a
step further and demonstrated that the same type theory was independent of the
uniqueness of identity proofs (UIP) principle by constructing a model of type
theory in groupoids~\cite{hofmann-streicher:1998}.

Hofmann and Streicher's contribution showed that it was possible for identity
in a universe of sets to ``mean'' bijection, a precursor to the univalence
principle of Voevodsky~\cite{voevodsky:2006,voevodsky:2010:cmu}, later codified
in the language of homotopy type theory (HoTT)~\cite{hottbook}. Later, it was
discovered that the identity type conferred an infinite-dimensional structure
already familiar in the context of homotopy
theory~\cite{awodey-warren:2009,van-den-berg-garner:2011,lumsdaine:2010}.

Cubical type theories were invented in order to repair several semantic and
syntactic anomalies of the new homotopy type theory; homotopy type theory lacks
\emph{canonicity}, a property closely related to but distinct from the
existence of a computational ``proofs-as-programs'' interpretation of the
language. On the other hand, the standard model of homotopy type theory in
\emph{simplicial sets}~\cite{kapulkin-lumsdaine:2021} must be formulated in a
boolean metatheory~\cite{bezem-coquand:2015}.  The discovery of a constructive
model for univalent type theory in cubical sets~\cite{bch:2014} sparked a
flurry of work on explicitly cubical type
theories~\cite{cchm:2017,angiuli-favonia-harper:2017,abcfhl:2021} which resolved
both the matters of canonicity and computational
interpretation~\cite{huber:2018,angiuli-favonia-harper:2017}.

While the benefits of cubical ideas for solving problems in
infinite-dimensional type theory are clear, we believed that it might be
possible to bring the cubical perspective to bear on the problems of
\emph{traditional} one-dimensional type theory, in which the intensional
identity type is augmented with enough uniqueness and extensionality principles
for it to behave like classical mathematical equality. In the context of the
strongest possible such uniqueness principle, \emph{equality reflection}, it
remains an open question whether it is possible to implement a usable proof
assistant; on the other hand, extending type theories with axioms for function
extensionality destroys canonicity and has significant usability problems.

Inspired by the work of Altenkirch, McBride, and Swierstra on
\emph{Observational Type
Theory}~\cite{altenkirch-mcbride:2006,altenkirch-mcbride-swierstra:2007}, which
internalized aspects of the setoid model of type theory, we sought to
internalize Coquand's semantic universe of Bishop sets~\cite{coquand:2017:bish}
as a type theory in its own right, XTT\@.
We believe that XTT is an ideal language for \emph{dependently typed
programming}, in which it is very important that coercions may be erased prior
to execution (a procedure that cannot be applied to coercions arising
from the univalence principle).
Unfortunately, there are several obstacles rendering both OTT and XTT
unsuitable for use as languages for formalizing general mathematics,
disadvantages not shared by homotopy type theory or its cubical variants.

\subsection{Trade-offs with universes}\label{sec:bad-universes}

In mathematics, a universe is a ``family of (some) families'', an object from
which every family in some class arises by pullback; we cannot have a universe
of all families for general reasons, but there are several restrictions of this
na\"{\i}ve idea that make sense, such as a universe of all monomorphisms (a
subobject classifier), or a universe of $\kappa$-compact families for some
regular or inaccessible cardinal $\kappa$.

Universes in mathematics are important for two reasons: first, they tame subtle
but essentially bureaucratic questions of size~\cite[Expos\'e I, Ch.~0]{sga:4},
and second, they provide a (stronger) alternative to the set-theoretical
\emph{axiom of replacement} that enables concepts to be formulated in the
convenient fiber-wise style familiar from dependent type
theory~\cite{streicher:2005}. Here, it is very important to ensure that the
universe imposes no spurious structures on the maps it classifies; for
instance, if $U$ is the universe of $\kappa$-compact maps, we may develop the
theory of $\kappa$-compact groups in terms of $U$. Then, a predicate defined over
$\kappa$-compact groups should have the same meaning as a predicate defined over
$U$-groups.

\NewDocumentCommand\TTObs{}{\textsf{TT}\textsuperscript{obs}}

\subsubsection{Induction-recursion and type-case}

The published models of both
OTT~\cite{altenkirch-mcbride:2006,altenkirch-mcbride-swierstra:2007} and XTT
(both~\cite{sterling-angiuli-gratzer:2019} and the present paper) share an
infelicitous interpretation of universes from the perspective of mathematics:
universes are modeled by \emph{closed} inductive-recursive types, hence in the
semantics there is a corresponding \emph{elimination principle} for the
universe that proceeds by cases on whether a given code is (\eg) a dependent
product type, a dependent sum type, the boolean type, \etc.  This defect
appears in both the syntax and semantics of OTT and XTT:\@ in the former case,
because OTT is not a type theory \emph{per se} but rather a syntactical model
construction, and in the latter, because we have explicitly added a type-case
operator to the theory of XTT\@.

The reason we view the closed universe semantics as infelicitous is that it
changes the meaning of quantification over (\eg) algebraic structures in a way
that is not really compatible with ordinary mathematical usage.  The existence
of a code $\hat{A}$ for a type $A$ in these universes expresses not only the
smallness of $A$, but also the fact that $A$ is either a dependent product, a
dependent sum, an equation, or it is the booleans (\etc), and the same for all
of $A$'s subterms. Hence a mathematical statement quantifying over the elements
of $\hat{A}$ does not have the right meaning when interpreted into such a
model, considering these additional assumptions on the form of $A$. At the very
least, a universe that is suitable for use by mathematicians should be
characterized (up to equivalence) by some intrinsic property such as smallness
or compactness, rather than by the syntactical form of the classified types.

As we have discussed in Section~\ref{sec:typecase}, the reason for
internalizing the type-case principle in the syntax of XTT is to ensure that
type constructors are \emph{internally injective}, a prerequisite for deciding
type checking in the presence of boundary separation in OTT/XTT-style
theories.\footnote{Note that we are referring to injectivity with respect to
propositional equality, not judgmental equality; in any well-adapted type
theory, one expects to have admissible injectivity up to judgmental equality.}
Although it is obvious in \emph{syntax} that one could choose a less
heavy-handed implementation of internal injectivity than type-case, \eg by
adding constants witnessing the internal injectivity of each individual type
constructor, we did not at the time believe that such a generalization would
meaningfully enlarge the space of semantic models of XTT\@.

\subsubsection{Syntax and semantics of internal injectivity}

Subsequent to the introduction of XTT by the present authors~~\cite{sterling-angiuli-gratzer:2019}, Pujet and
Tabareau~\cite{pujet-tabareau:2022} have introduced a new non-cubical
reconstruction of observational type theory called \TTObs{} that has a number
of attractive properties going beyond the ones proved in the present paper,
including normalization and decidability of type checking. Just as XTT,
\TTObs{} improves on OTT by defining a true type theory that is distinguished
from any particular syntactic model or translation; \TTObs{} deviates from XTT,
however, by witnessing the internal injectivity of type constructors directly
rather than by means of a type-case primitive.

Gratzer~\cite{gratzer:2022:universe} has recently demonstrated a construction
of \emph{open} inductive-recursive universes that are generic for a suitable
class of families (\eg relatively $\kappa$-small or relatively $\kappa$-compact
families for some a strongly inaccessible cardinal $\kappa$) and yet validate
the injectivity assumptions needed by XTT and \TTObs. Gratzer's construction
refutes our previous assumption that weakening type-case to a direct account of
internal injectivity would not lead to new and more useful models of XTT, and
provides some vindication to the choice of Pujet and Tabareau to treat internal
injectivity directly.

Therefore we conclude in hindsight that a type theory for boundary-separated sets,
whether treated cubically or not, need not internalize a closed universe but
can instead be usefully equipped with explicit laws governing the internal
injectivity of type constructors. We do maintain, however, that there remains a
contradiction between the intended meaning of universes in mathematical
practice and the need for injectivity in the implementation of boundary
separation in OTT/XTT/\TTObs: the statement ``If $A\to B$ is equal to $C\to D$, then $A$ is equal
to $C$'' is arguably a \emph{junk theorem} akin to $1\in \pi$ or $4\subseteq 9$.

\NewDocumentCommand\Squash{d<>m}{
  \verts{#2}{\IfValueT{#1}{_{#1}}}
}

\subsection{What is a proposition?}

A considerably more subtle obstacle for using either OTT or XTT in the
formalization of mathematics is to be found when choosing a suitable
notion of \emph{proposition} or \emph{relation}. In type theory, there are
\emph{a priori} two ways to formalize propositions:

\begin{enumerate}
  \item A \emph{strict} proposition is a type whose elements are all judgmentally equal.
  \item A \emph{weak} proposition is a type $X$ together with a function $\prn{x,y : X}\to \CTyPath{X}{x}{y}$.
\end{enumerate}

\noindent
In OTT/XTT, the strict and weak notions of proposition do not agree,
though they could be forced to agree by adding equality reflection.
Unfortunately, when investigating the interplay between the indispensable
principles of function comprehension and effectivity of equivalence relations,
we will find that this mismatch cannot be resolved by favoring either the
strict or the weak notion.

As soon as one has chosen a notion of proposition, one may consider the
corresponding ``squash type'', the reflection of types into propositions:
\begin{enumerate}

  \item Given a type $A$, the strict squash type $\Squash<s>{A}$ is a strict
    proposition; a function $\Squash<s>{A}\to B$ is a function $f : A\to B$ such
    that every $f(x)$ is judgmentally equal to $f(y)$.
    The strict squash type
    was investigated by Awodey and Bauer~\cite{awodey-bauer:2004}, and appears
    in recent versions of Coq and Agda~\cite{gilbert-cockx-sozeau-tabareau:2019}.

  \item Given a type $A$, the weak squash type $\Squash<w>{A}$ is a weak
    proposition; a function $\Squash<w>{A}\to B$ is a function $f : A\to B$
    together with a function assigning to each $x,y:A$ an element of
    $\CTyPath{B}{f(x)}{f(y)}$. The weak squash type appears in homotopy type
    theory as \emph{propositional truncation}~\cite{hottbook}.

\end{enumerate}

\noindent
Dependent product and binary product preserve the property of being a (strict,
weak) proposition, and may therefore be used as universal quantification and
conjunction in a logic of propositions. Dependent sum and binary sum do not
preserve this property, but they can be squashed in order to define existential
quantification and disjunction. The logic of (strict, weak) propositions is
summarized below:
\begin{align*}
  \forall x : A.P(x) &\defeq (x : A) \to P(x)\\
  P \supset Q &\defeq P \to Q\\
  P\land Q &\defeq P \times Q\\
  \exists^{s/w} x:A.P(x) &\defeq \Squash<s/w>{(x:A)\times P(x)}\\
  P\lor^{s/w} Q &\defeq \Squash<s/w>{(x : \CTyBool)\times \CTmIf{}{x}{P}{Q}}\\
  \prn{M =_{A} N} &\defeq \CTyPath{A}{M}{N}
\end{align*}

\subsubsection{Function comprehension}

What is the meaning of ``function''? There is only one possible answer: it is
an element of an exponential object. In some categories, however, these
exponentials can be reconstructed as representing objects for collections of
\emph{functional relations}. This isomorphism between the collection of
functions $A\to B$ and subobjects of $A\times B$ satisfying a unique existence
property is traditionally referred to as the ``axiom of unique choice'', though
it is perhaps better to refer to it as \emph{function comprehension}.  Function
comprehension is a crucial feature of both classical \emph{and} constructive
mathematics, and life becomes very difficult in categories where function
comprehension fails.

A (strict, weak) functional relation from $A$ to $B$ is a family of (strict,
weak) propositions $x:A,y:B\vdash R(x,y)$ together with a proof of the
following (strict, weak) proposition:
\[
  \forall x : A.\
  \exists^{s/w} y:B.\
  R(x,y) \land
  \forall y' : B.\
  R(x,y') \supset
  y =_{B} y'
  \tag{functionality}
\]

The proposition above may be rendered into the language of types as follows:
\[
  (x : A) \to \Squash<s/w>{
    \prn{y : B}
    \times
    R(x,y)
    \times
    \prn{
      \prn{y' : B}\to
      R(x,y') \to
      \CTyPath{B}{y}{y'}
    }
  }
\]

The function comprehension principle is immediate for weak propositions in OTT
and XTT, but fails for strict propositions.

\begin{enumerate}

  \item To exhibit a function $A\to B$ from a weak functional relation, we may
    extract the $y : B$ using the universal property of the weak squash type,
    fulfilling the auxiliary obligation using the weak uniqueness of $y$ with
    $R(x,y)$.

  \item In doing the same with a strict functional relation, we run into a
    problem: to make a function out of an element of the strict squash type, we end up needing
    that $y$ is unique with $R(x,y)$ \emph{up to judgmental equality}, but we
    have only an element of $\CTyPath{B}{y}{y'}$ for each $y'$ such that
    $R(x,y')$.

\end{enumerate}

\noindent
Therefore, short of adding equality reflection, we must conclude that the weak
notion of proposition is the ``correct'' one, and the strict one is not
particularly useful for mathematics in an environment without equality
reflection. Unfortunately, we will see that another indispensable reasoning
principle in constructive and classical mathematics, the effectivity of
equivalence relations, appears to be compatible only with the strict notion in
a boundary separated environment lacking equality reflection. (In contrast,
full cubical type theory satisfies a vastly stronger exactness condition called
\emph{descent}, generalizing both the disjointness of coproducts and the
effectivity of equivalence relations.)

\subsubsection{Effectivity of equivalence relations}

It is possible to add quotient types to both XTT and OTT (an extension
implemented, for instance, as part of the experimental Epigram~2 proof
assistant~\cite{epigram:blog:quotients}); likewise, Nuprl has supported a
version of quotient types for decades.  Unfortunately, these quotients can be
made to have good properties only for \emph{certain} equivalence relations:
\begin{enumerate}

  \item In Nuprl, only equivalence relations valued in ``strong propositions''
    (types having at most one element up to the intensional \emph{untyped}
    equivalence of Howe~\cite{howe:1989}) have good quotients. Equivalence relations valued in general
    propositions (types having at most one element up to extensional equality) do
    \emph{not} necessarily have good quotients, a serious problem alluded to in
    the work of Nogin~\cite{nogin:2002}.

  \item In OTT and XTT, only equivalence relations valued in strict
    propositions can have good quotients.

\end{enumerate}

\noindent
Writing $[-] : A \to A/R$ for the quotient map, the quotient $A/R$ is ``good''
when each type $\CTyPath{A/R}{[x]}{[y]}$ is equivalent to $R(x,y)$; this
property, called the effectivity of $R$, does not follow from the rules of
quotient types alone and in fact fails in many categories (such as categories of
partial equivalence relations). The effectivity of \emph{all} equivalence
relations is, however, indispensable for practical use of quotients in
mathematics.

In type theory, the effectivity of equivalence relations follows from
propositional extensionality, a restricted version of the univalence principle
that places bi-implications $\prn{f,g} : P\leftrightarrow Q$ into
correspondence with proofs of equality $\CKwd{pua}\prn{f,g} :
\CTyPath{}{P}{Q}$; in topos theory, this corresponds to the existence of a
subobject classifier. We will see, however, that it is not possible to extend
either XTT or OTT (or any type theory satisfying boundary separation or
definitional UIP) with a univalence principle for weak propositions without
some fundamentally new ideas.

In cubical type theories, univalence is supported by means of a special
connective taking an equivalence of types and returning a path between the
corresponding types~\cite{angiuli:2019}; we might attempt to extend XTT by a
version of this connective restricted to weak propositions:

\medskip
\begin{mathflow}
  \inferrule{
    \Gamma\vdash r : \DIM\qquad
    \Gamma, r = 0\vdash P\ \mathit{prop}\qquad
    \Gamma \vdash Q\ \mathit{prop} \\
    \Gamma, r = 0\vdash f : P \to Q \qquad
    \Gamma, r = 0\vdash g : Q \to P
  }{
    \Gamma\vdash \CKwd{V}_r\prn{P,Q,f,g}\ \mathit{prop}\ \brk{
      \Bdry{r}\to
      \brk{r=\CDim0\to P\mid r=\CDim1\to Q}
    }
  }
\end{mathflow}
\medskip

Unfortunately, we can show that a univalent universe of weak propositions $\CKwd{Prop}$ cannot be boundary separated.

\begin{lem}\label{lem:prop:equivalences-unique}
  Suppose that we have a boundary separated universe of propositions closed
  under $\CKwd{V}$-types (thence univalent); then, all equivalences between two
  propositions are judgmentally equal.
\end{lem}

\begin{proof}
  Let $P,Q$ be two propositions classified by the univalent universe of
  propositions, and let $(f,g)$ and $(f',g')$ be two equivalences between them.
  Abstracting a dimension $i:\DIM$, we therefore have two $\CKwd{V}$-types
  $\CKwd{V}_i\prn{P,Q,f,g}$ and $\CKwd{V}_i\prn{P,Q,f',g'}$; by boundary separation,
  we in fact have $\CKwd{V}_i\prn{P,Q,f,g} = \CKwd{V}_i\prn{P,Q,f',g'}$.

  We will show that $f = f'$ judgmentally by using coercion in the
  $\CKwd{V}$-type, following the computation rules described by
  Angiuli~\cite{angiuli:2019}.
  \begin{align*}
    \lambda x.f(x) &= \lambda x.\CTmCoe{0}{1}{\_.Q}{f(x)} && \text{regularity}
    \\
    &= \lambda x.\CTmCoe{\CDim0}{\CDim1}{i.\CKwd{V}_i\prn{P,Q,f,g}}{x} && \text{\cite[p.163]{angiuli:2019}}
    \\
    &= \lambda x.\CTmCoe{\CDim0}{\CDim1}{i.\CKwd{V}_i\prn{P,Q,f',g'}}{x} && \text{boundary separation}
    \\
    &= \lambda x.\CTmCoe{\CDim0}{\CDim1}{\_.Q}{f'(x)} && \text{\cite[p.163]{angiuli:2019}}
    \\
    &= \lambda x.f'(x) && \text{regularity}
  \end{align*}

  To see that $g=g'$, simply repeat the procedure with the inverse equivalences.
\end{proof}

\begin{rem}
  Lemma~\ref{lem:prop:equivalences-unique} can likewise be replayed when
  $\CKwd{glue}$-types \`a la~\cite{abcfhl:2021} are used instead of
  $\CKwd{V}$-types \`a la~\cite{angiuli-favonia-harper:2017}. In either case,
  coercion can be used (modulo regularity) to recover the equivalence from the
  line of types.
\end{rem}

The assumptions of Lemma~\ref{lem:prop:equivalences-unique}
imply some intensional type theoretic taboos.

\begin{cor}\label{cor:prop:all-strict}
  Under the assumptions of Lemma~\ref{lem:prop:equivalences-unique}, all propositions are strict propositions.
\end{cor}

\begin{proof}
  Let $P$ be a proposition, and let $x,y$ be proofs of $P$. The constant
  functions $\lambda\_.x$ and $\lambda\_.y$ are both equivalences
  $P\leftrightarrow P$; by Lemma~\ref{lem:prop:equivalences-unique}, they are
  judgmentally equal. Therefore, $x$ and $y$ are judgmentally equal.
\end{proof}

\begin{cor}
  Under the assumptions of Lemma~\ref{lem:prop:equivalences-unique}, equality reflection holds.
\end{cor}

\begin{proof}
  Let $A:\CU$ and let $a\in A$; then $S_A(a) \defeq \prn{x\in
  A}\times\CTyPath{\CEl{A}}{a}{x}$ is a weak proposition. By
  Corollary~\ref{cor:prop:all-strict}, $S_A(a)$ is moreover a strict
  proposition. Let $a' \in A$ and $p : \CTyPath{\CEl{A}}{a}{a'}$; therefore we have
  $\CTmPair{a}{\lambda \_.a} = \CTmPair{a'}{p} : S_A(a)$, whence $a = a'$ judgmentally.
\end{proof}

\section*{Acknowledgment}

\noindent We thank Thorsten Altenkirch, Mathieu Anel, Steve Awodey, Lars
Birkedal, Evan Cavallo, David Thrane Christiansen, Thierry Coquand, Kuen-Bang
Hou (Favonia), Marcelo Fiore, Jonas Frey, Ambrus Kaposi, Krzysztof Kapulkin,
Alex Kavvos, Andr\'as Kov\'acs, Dan Licata, Conor McBride, Darin Morrison,
Anders M\"ortberg, Michael Shulman, Bas Spitters, and Thomas Streicher for
helpful conversations about extensional equality, algebraic type theory, and
categorical gluing.  We thank our anonymous reviewers for their insightful
comments, and especially thank Robert Harper for valuable conversations
throughout the development of this work. We thank Tristan Nguyen at AFOSR for
his support.

The authors gratefully acknowledge the support of the Air Force Office of
Scientific Research through MURI grants FA9550--15--1--0053 and FA9550--21--8. Any
opinions, findings and conclusions or recommendations expressed in this material
are those of the authors and do not necessarily reflect the views of the AFOSR\@.

\bibliographystyle{alphaurl}
\bibliography{refs,temp-refs}

\end{document}